\documentclass[12pt,draft]{article}

\usepackage[margin=0.7in]{geometry} %
\usepackage{setspace}               %

\usepackage{siunitx}       %
\usepackage{amsmath}    %
\usepackage{amssymb}                %
\usepackage{amsthm}                 %

\usepackage{authblk}

\usepackage[font=footnotesize,hypcap=true]{caption}
\usepackage[font=footnotesize,hypcap=true]{subcaption}

\usepackage{graphicx}
\usepackage[hyperref,x11names,cmyk,pdftex,table]{xcolor}
\usepackage{tikz}

\usepackage{booktabs}
\usepackage{multirow}
\usepackage{dcolumn}
\usepackage{arydshln} %

\usepackage{algorithmic}
\usepackage{enumerate}
\usepackage{paralist}

\usepackage[resetfonts]{cmap}
\usepackage[T1]{fontenc}
\usepackage[utf8]{inputenc}

\usepackage[final,layout={inline},author=]{fixme}

\usepackage[pdftex,breaklinks=true,bookmarks=true]{hyperref}

\usepackage[
  backend=biber,
  style=ieee,  %
  sorting=nty,
  doi=false,
  url=false,
  isbn=false,
  giveninits=true,
  maxbibnames=9999,
  hyperref]{biblatex}

\AtEveryBibitem{%
  \clearlist{language}%
}

\AtEveryBibitem{\clearfield{month}}
\AtEveryCitekey{\clearfield{month}}

\DeclareFieldFormat{titlecase}{#1}

\usepackage{cleveref}

\newcommand{\gsurf}{G-SURF}

\title{A systematic association of subgraph counts over a network}

\makeatletter
\let\titletext\@title
\makeatother

\author[1]{Dimitris~Floros}
\author[1,2]{Nikos~Pitsianis}
\author[2]{Xiaobai~Sun}

\affil[1]{Department~of~Electrical~and~Computer~Engineering,
  Aristotle~University~of~Thessaloniki, Thessaloniki~54124,~Greece}
\affil[2]{Department~of~Computer~Science, Duke~University,
  Durham,~NC~27708,~USA}

\newcommand{\pdfauthors}{%
  Floros D., Pitsianis N., Sun X.
}

\hypersetup{%
  pdfauthor={\pdfauthors},%
  pdftitle={\titletext}%
}

\makeatletter
\def\th@plain{%
  \thm@notefont{}%
  \itshape %
}
\def\th@definition{%
  \thm@notefont{}%
  \normalfont %
}
\makeatother

\newtheorem{theorem}{Theorem}
\newtheorem{corollary}[theorem]{Corollary}

\newtheorem{proposition}{Proposition}
\newtheorem{definition}{Definition}
\newtheorem{example}{Example}

\Crefname{ALC@unique}{Line}{Lines} %

\setkeys{Gin}{draft=false}

\pdfpageattr{/Group <</S /Transparency /I true /CS /DeviceRGB>>}

\let\leftorig\left
\let\rightorig\right
\renewcommand{\left}{\mathopen{}\mathclose\bgroup\leftorig}
\renewcommand{\right}{\aftergroup\egroup\rightorig}

\hypersetup{final}
\DeclareSIUnit{\nothing}{\relax}

\newbibmacro{string+doiurlisbn}[1]{%
  \iffieldundef{doi}{%
    \iffieldundef{url}{%
      \iffieldundef{isbn}{%
        \iffieldundef{issn}{%
          #1%
        }{%
          \href{http://books.google.com/books?vid=ISSN\thefield{issn}}{#1}%
        }%
      }{%
        \href{http://books.google.com/books?vid=ISBN\thefield{isbn}}{#1}%
      }%
    }{%
      \href{\thefield{url}}{#1}%
    }%
  }{%
    \href{http://dx.doi.org/\thefield{doi}}{#1}%
  }%
}

\DeclareFieldFormat{title}{\usebibmacro{string+doiurlisbn}{\mkbibemph{#1}}}
\DeclareFieldFormat[article,incollection,book,inproceedings]{title}%
{\usebibmacro{string+doiurlisbn}{\mkbibquote{#1}}}

\begin{document}

\doublespacing  %
\pdfbookmark[1]{Title}{sec:title}
\maketitle

\addcontentsline{toc}{section}{Abstract}
\begin{abstract}

We associate all small subgraph counting problems with a systematic
graph
encoding/rep{-}resentation
system which makes a coherent use of
graphlet structures. The system can serve as a unified foundation for
studying and connecting many important graph problems in theory and
practice.  We describe topological relations among graphlets (graph elements)
in rigorous mathematics language and from the perspective of graph
encoding. We uncover, characterize and utilize algebraic and numerical
relations in graphlet counts/frequencies. We present a novel algorithm
for efficiently counting small subgraphs as a practical product of
our theoretical findings.

 \end{abstract}

\clearpage

\section{Introduction}
\label{sec:introduction}

Graph or network studies, classical or modern, inevitably examine
subgraph structures and counts, especially small subgraphs.
Detecting the presence or absence of a small subgraph $H$ in
a larger graph $G$ under consideration is fundamental to several
classical graph-theoretic problems such as graph recognition and graph
classification~\cite{olariu1988,faudree1997,kloks2000,corneil1985}.  In
the early 1930s, Whitney investigated graph
connectivity with several small subgraph patterns~\cite{whitney1932},
which in modern terms are triangles, claws, diamonds, and four-node cliques.
Since then, if not earlier, the counts or distributions of small
subgraphs with prescribed patterns have been
persistently used as primitives to characterize, recognize and
categorize graphs. We give in \Cref{sec:conclusion} a brief review of
subgraph counting problems.

At the turn of the 21\textsuperscript{st} century, subgraph counts and
distributions gained unprecedented attention and applications with the
advent of real-world networks and the advance in network
analysis techniques. Two seminal papers, among others, made a massive
impact on applied network studies by using subgraph structures and
counts. In 1998, Watts and Strogatz used triangles in their network
centrality measure and network model~\cite{watts1998c}. In 2002, Milo
and his five co-authors found and defined network motifs as simple
network building blocks~\cite{milo2002}. The seminal works inspired
many new approaches in improtant applications such as
in biochemistry for investigating gene interactions,
in neurobiology for mapping neural pathways in
the brain, in computer vision and graphics for image alignment.
Frequently occurring subgraphs are used to analyze protein-protein
interaction networks and metabolic networks for drug target
discovery~\cite{csermely2013}.

In fact, 2004 saw another important work by Pr\v{z}ulj, Corneil and
Jurisica~\cite{przulj2004}.  The authors introduced the concept and
use of graphlets for network analysis, which extend the conventional
approach with edges and triangles to small subgraphs of various
topological structures and their statistical distributions. The work
has gained more attention and appreciation over the years, mostly in
the community of researchers investigating biological networks with
statistical methods~\cite{ribeiro2020}.

\begin{figure}[!t]
  \centering
  \begin{subfigure}{.48\linewidth}
    \centering
    \includegraphics[width=.9\linewidth]{%
      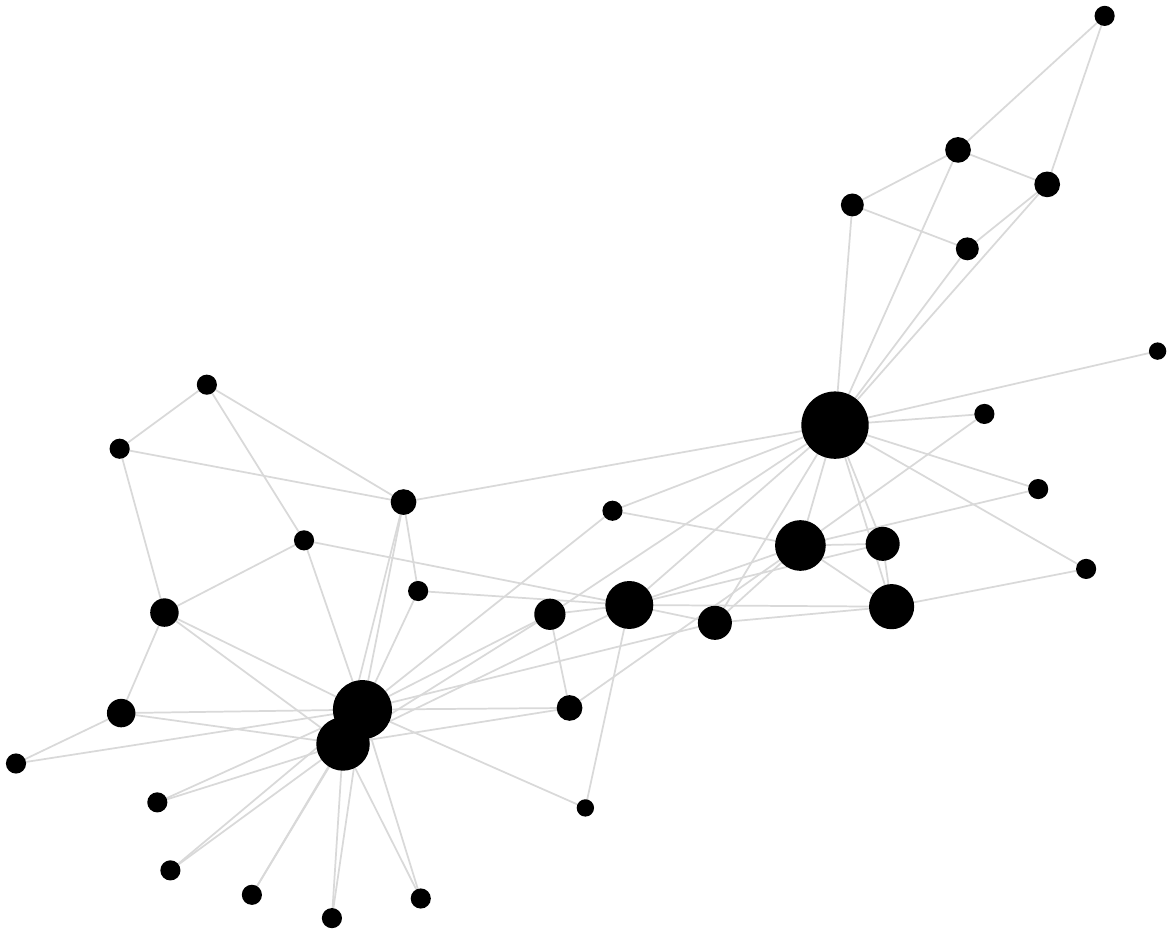}
  \end{subfigure}
  \begin{subfigure}{.48\linewidth}
    \centering
    \includegraphics[width=.8\linewidth]{%
      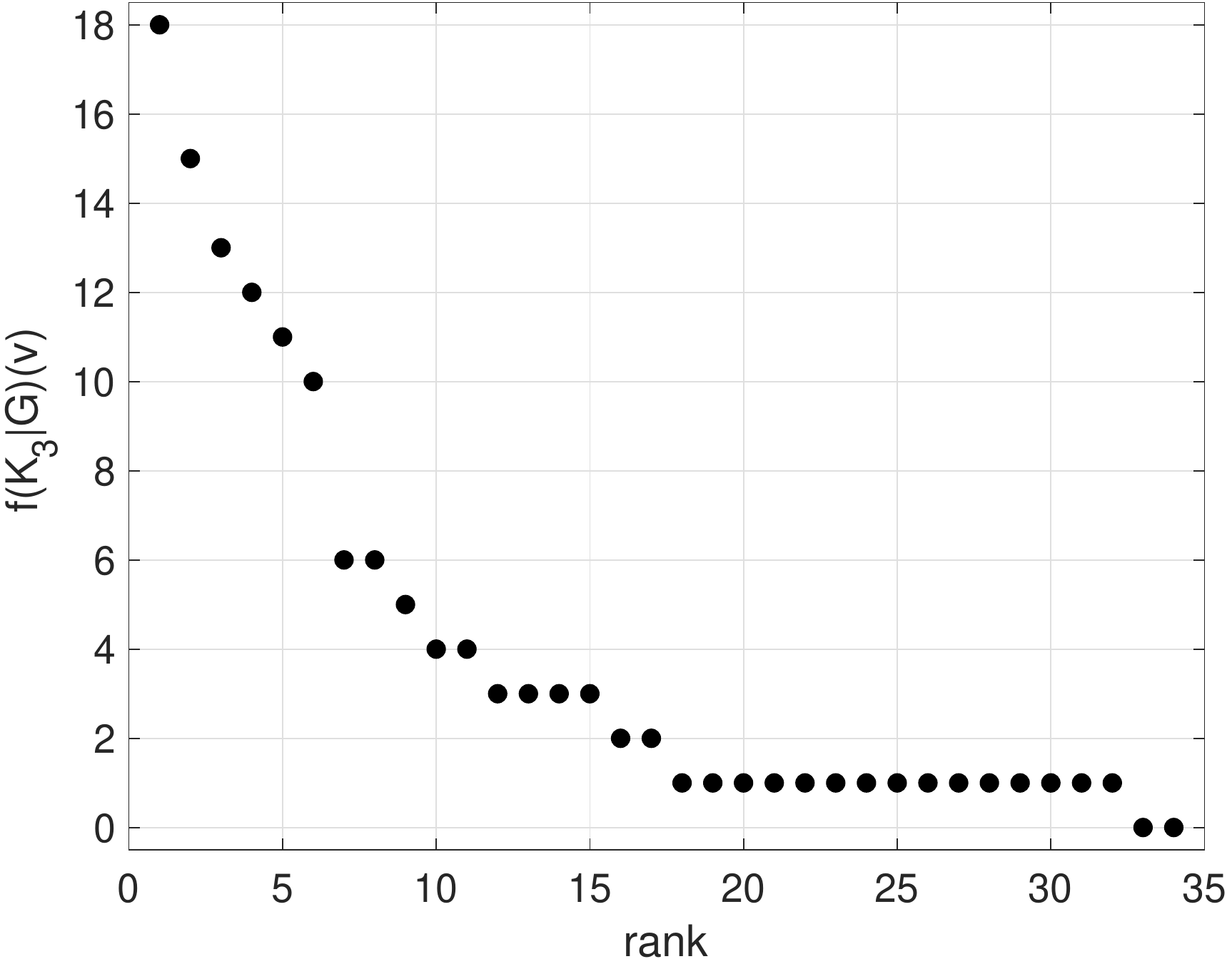}
  \end{subfigure}
  \caption{Triangle-frequency map (left) on Zachary's karate club
    friendship network~\cite{zachary1977b} and the triangle-frequency
    sequence (right).  The network has \num{34} club members and
    \num{78} friendship links.  The counts in the
    (spatial/topological) map have 1-1 association with club members
    (vertex labels), members in more triangle links are shown with
    larger markers. The counts in the (statistical) sequence are
    sorted in non-ascending order. }
  \label{fig:frequency-map-sequence}
\end{figure}

Applied network analysis and graph data mining with motifs or graphlet
distributions remain ad hoc and in flux, by and large, with
undiminished interest and enthusiasm yet lack of coherent
understanding and principled decision making at multiple data analysis
stages.  This situation is reflected in multiple surveys and
reviews~\cite{washio2003,lee2010,jiang2013,alhasan2018,ribeiro2020,bouhenni2021}.
Rarely graphlet frequencies are connected to motif detection or discovery.

In the present work, we make a systematic association of all small
subgraph counting problems with a graph encoding system which makes a
coherent use of graphlet structures. The graph encoding/representation
system can serve as a unified foundation for studying and connecting
many significant graph problems in theory as well as in practice.  In
\Cref{sec:formal-description}, we first give a formal description of
multi-channel graph encoding with template graphs, in rigorous
mathematics language and from the graph encoding perspective. We then
focus on a system of graph encoding elements using what are known as
graphlets.  In \Cref{sec:feature-conversion}, based on topological
relations among the graphlets, we uncover, characterize, classify and utilize
algebraic and quantitative relations in graphlet counts or frequencies. In
\Cref{sec:g-surf}, we present a novel algorithm for efficiently
counting small subgraphs as a practical product of our theoretical
findings. We show a significant reduction in the computation cost of
generating graphlet maps on a real-world network.
In \Cref{sec:conclusion}, we comment on the connections made by our
analysis among previous subgraph counting problems and
methods. Certain shortcomings in some previous works become evident
consequently. We also remark on potential applications of the present
work.

\section{Problem description \& preliminary}
\label{sec:formal-description}

\begin{figure}[!t]
  \definecolor{tabintroblue}{rgb}{0,0.6470,0.9410}
  \newcommand{\mbf}[1]{\textbf{#1}}
  \centering
  \begin{subfigure}{0.45\linewidth}
    \centering
    \begin{subfigure}{0.4\linewidth}
      \centering
      \includegraphics[width=.7\linewidth]{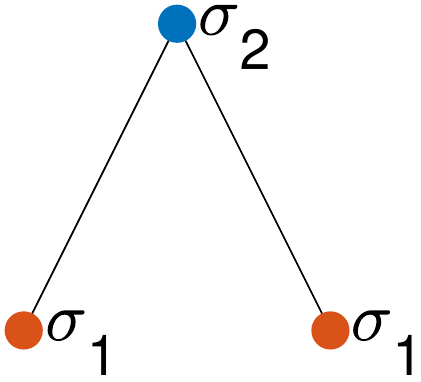}
      \\[-1em]
      \caption*{$H$}
    \end{subfigure}
    \\[0.5em]
    \begin{subfigure}{.9\linewidth}
      \centering
      \includegraphics[width=.7\linewidth]{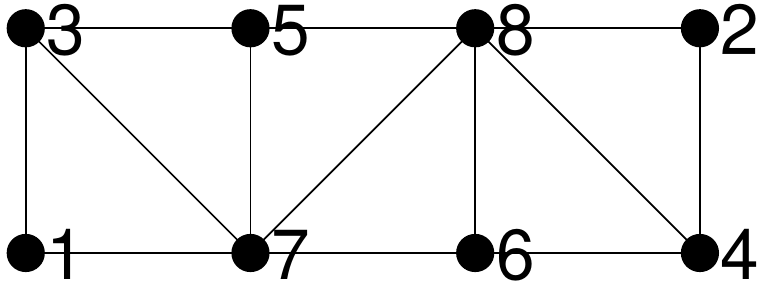}
    \end{subfigure}
  \end{subfigure}
  \begin{subfigure}{0.45\linewidth}
    \centering
    \begin{subfigure}{.7\linewidth}
      \centering
      \includegraphics[width=.7\linewidth]{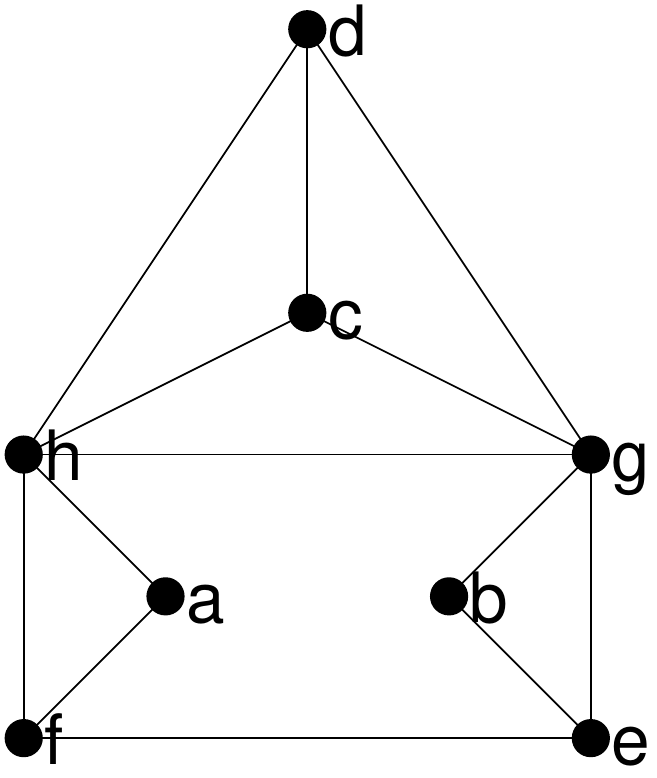}
    \end{subfigure}
  \end{subfigure}
  \\[0em]
  \begin{subfigure}{0.45\linewidth}
    \caption*{$G_{1}$}
  \end{subfigure}
  \begin{subfigure}{0.45\linewidth}
    \caption*{$G_{2}$}
  \end{subfigure}
  \\[0em]
  \begin{subfigure}{0.45\linewidth}
    \begin{subfigure}{0.9\linewidth}
\resizebox{\linewidth}{!}{%
\begin{tabular}{lcccccccc}
\toprule
$v$ & 1 & 2 & 3 & 4 & 5 & 6 & 7 & 8 \\
\midrule
$d(v)$ & 2 & 2 & 3 & 3 & 3 & 3 & 5 & 5 \\
$f(H|G)(v)$ & 4 & 4 & 4 & 4 & 7 & 7 & 9 & 9 \\
$f({\color{red}H_{\sigma_{1}}}|G)(v)$ & 4 & 4 & \textbf{3} & \textbf{3} & \textbf{6} & \textbf{6} & 3 & 3 \\
$f({\color{tabintroblue}H_{\sigma_{2}}}|G)(v)$ & 0 & 0 & \textbf{1} & \textbf{1} & \textbf{1} & \textbf{1} & 6 & 6 \\
\bottomrule
\end{tabular}%
}
     \end{subfigure}
  \end{subfigure}
  \begin{subfigure}{0.45\linewidth}
    \begin{subfigure}{0.9\linewidth}
      \centering
\resizebox{\linewidth}{!}{%
\begin{tabular}{lcccccccc}
\toprule
$v$ & a & b & c & d & e & f & g & h \\
\midrule
$d(v)$ & 2 & 2 & 3 & 3 & 3 & 3 & 5 & 5 \\
$f(H|G)(v)$ & 4 & 4 & 4 & 4 & 7 & 7 & 9 & 9 \\
$f({\color{red}H_{\sigma_{1}}}|G)(v)$ & 4 & 4 & \textbf{4} & \textbf{4} & \textbf{5} & \textbf{5} & 3 & 3 \\
$f({\color{tabintroblue}H_{\sigma_{2}}}|G)(v)$ & 0 & 0 & \textbf{0} & \textbf{0} & \textbf{2} & \textbf{2} & 6 & 6 \\
\bottomrule
\end{tabular}%
}
     \end{subfigure}
  \end{subfigure}
  \caption{Graph differentiation by orbit-specific graphlet frequency sequences.
    The template graph $H$ is $K_{1,2}$, larger than $K_{2}$ and smaller than $K_{3}$.
    The leaf nodes of $H$ in red are in orbit $\sigma_{1}$,
    the root node in blue is in orbit $\sigma_{2}$.
    Graph $G_{1}$ and graph $G_{2}$ are attributed with their respective
    frequency sequence tables at the bottom. They
    have identical degree sequences,  as shown in the first rows of the two tables.
    Their frequency sequences with respect to $H$ are also identical, as shown
    in the second rows of the tables. Graph $G_1$
    and graph $G_2$ are differentiated by their orbit-specific vertex maps or
    sequences shown in the third and fourth rows, which decompose the sequences
    in the second rows.
  }
  \label{fig:graph-differentiation}
\end{figure}

We give a formal description of the basic subgraph counting problems.
Any graph or network addressed in this paper is undirected, with
simple edges and without self-loops. A graph is denoted by
$G=G(V, E)$, with $V$ or $V(G)$ as the set of vertices or nodes and
$E$ or $E(G) \subseteq V(G) \times V(G) $ as the set of edges or
links. The primary sizes of graph $G$ are specified by
$n=n(G) = |V(G)|$ and $m=m(G) = |E(G)|$.  Graph $G'$ is a subgraph of
$G$ if $V(G') \subseteq V(G)$ and $E(G') \subseteq E(G)$.  If
$E(G') = E(G) \cap ( V(G) \times V(G) )$ in addition, $G'$ is an
induced subgraph of $G$. A graph is connected if every pair of
vertices is connected by a path.
Two graphs $G_1$ and $G_2$ are isomorphic in topological link
structure, $G_1 \cong G_2$, if there is a bijection $\phi$ between
$V(G_1)$ and $V(G_2)$ such that $(u,v) \in E(G_1)$ implies
$ (\phi(u), \phi(v)) \in E(G_2)$ and vice versa.
Graph $G$ is labeled if every vertex of $G$ has a unique label, or
equivalently, the vertices are labeled from $1$ to $n(G)$. A graph to
be characterized by its subgraph structures is referred to as a source
graph.

\subsection{One subgraph template: counts \& maps}
\label{sec:problem-description}

\begin{definition}
  \label{def:global-counts}
  {\rm (Subgraph counts over a source graph)}
  Let $G$ be a labeled source graph. Let $H$ be a connected, unlabeled
  \emph{template} graph, $n(H) \leq n(G)$.
  The \emph{gross} count (or frequency) of $H$-isomorphic subgraphs in $G$ is
  $g(H|G) = |\Gamma_{g}(H|G)|$,
  where
  $\Gamma_g(H | G) = \left\{ H' \subseteq G \mid H' \cong H \right\}$.
  Any two elements $H', H''$ in $\Gamma_{g}(H|G)$ are two different
  subgraphs, $V(H') \neq V(H'')$.
  The \emph{net} count is $f(H|G) = |\Gamma_{f}(H|G)|$,
  where
  $\Gamma_f(H | G) = \left\{ H' \in \Gamma_{g}(H | G) \mid E(H') =
    E(G) \cap (V(H')\times V(H')) \right\}$. 
  \end{definition}
  The template size $n(H)$ is assumed bounded, throughout the rest of the
  paper, independent of any source graph.
  By \Cref{def:global-counts} the count of the induced subgraphs is
  more constrained, $f(H|G) \leq g(H|G)$.
  For any $K_p$, the clique
  with $p$ nodes, $p\geq 0$, $g( K_p |G ) = f( K_p | G)$. In
  particular, $f( K_2 | G)$ is the total number of edges in $G$,
  $m(G) = f( K_2 | G)$.  Two graphs with the same number of edges may
  be differentiated by their degree sequences. By extension, two
  networks with equal global counts with respect to (w.r.t.)  the same
  template $H$, by \Cref{def:global-counts}, may be differentiated by
  local counts at vertices.

  \begin{definition}
   \label{def:local-counts-at-vertices}    
  {\rm (Subgraph counts at incidence vertices)}
  Let $G$ and $H$ be defined as in \Cref{def:global-counts}.  For
  every vertex $v$ in $V(G)$, let 
  $\Gamma_{g}(H|G)(v) = \{ H' \in \Gamma_{g}(H|G) \mid v \in V(H') \}$
  and let
  $ \Gamma_{f}(H|G)(v) = \{ H' \in \Gamma_{f}(H|G) \mid v \in V(H')\}
  $.
  The gross count (or frequency) of $H$-isomorphic subgraphs incident with $v$ is
  $g(H|G)(v) = |\Gamma_{g}(H|G)(v)|$; the net count at $v$ is
  $f(H|G)(v) = |\Gamma_{f}(H|G)(v)|$.
\end{definition}
In particular, $f(K_2|G)(v)$ is $d(v)$, the degree of $v$. The essence
of \Cref{def:local-counts-at-vertices} is the introduction of a vertex
map over $V(G)$ with respect to any particular template $H$. For
instance, $f(K_2|G)(v)$ is a map onto vertex
positions/locations/labels. The degree map gives rise to the degree
sequence, which sorts the degrees in descending or ascending order
with the vertex label or position information discarded.  Similarly,
$f(H | G)(v)$ gives rise to the frequency sequence of $G$ {w.r.t.}
$H$, see the triangle-frequency map and triangle-frequency sequence in
\Cref{fig:frequency-map-sequence}.

The subgraph counts and maps defined above are graph invariants. They
encode graph information.

\subsection{Multi-channel graph encoding}
\label{sec:feature-vector-map}

The coding capacity for graph representation and differentiation can
be increased by using multiple templates
${\cal H} = \{H_p,\, p=1,2,\cdots, P \}$, $P>1$. For instance, in
addition to the edge graph $K_{2}$ for encoding the degree
information, one may also use $K_{1,2}$ of bi-fork pattern to encode
more structural information.  In \Cref{fig:graph-differentiation}, two
simple templates $K_2$ and $K_{1,2}$ are used to compare and
differentiate graph $G_1$ and graph $G_2$.

Graph encoding with multiple templates has more discriminative
capacity than with a single template.  Let ${\cal H} $ be a collection
of template graphs. Let $G$ be a source graph. By
\Cref{def:local-counts-at-vertices}, each template $H\in {\cal H} $
identifies with a unique net-frequency map (or heat map) over
$V(G)$. There are multiple views at multiple granularity levels:
local, regional and global.
Locally at each vertex $v \in V(G)$, a unique
$| {\cal H} |$-dimensional frequency (feature) vector
$ f( {\cal H} |G)(v) $ is uniquely defined.  The frequency vector at
vertex $v$ encodes the topological structures in a neighborhood of the
vertex;
the frequency maps capture spatial and statistical information of
pattern distributions and inter-pattern association or disassociation
over the entire source graph or large regional subgraphs.

Besides the use of multiple template patterns, we describe the concept and
approach of sub-channel decomposition for increasing code capacity without
resorting to a new template pattern. For example, in \cite{floros2020a}
only two template patterns $K_2$ and $K_{1,2}$ are used to detect
dynamic changes in temporal sequences of large, real-world networks.
Three net-frequency maps are generated per network with greater
discriminative power at about the same cost for generating two maps.

Sub-channel decomposition applies to any template graph $H$ with more than
one orbits. The node set $V(H)$ can be uniquely partitioned into
orbits, namely, disjoint subsets of equivalent nodes that are
reflective, symmetric and transitive under automorphisms.
An automorphism on $H$ is a link-invariant bijection $\phi$ on $V(H)$,
i.e., $(u,v) \in E(H)$ if and only if
$(\phi(u), \phi(v)) \in E(H)$. When $H\cong H'$, we denote by
$\sigma \cong \sigma'$ the correspondence between orbit $\sigma$ of
$H$ and orbit $\sigma'$ of $H'$. In \Cref{fig:graph-differentiation},
the template $K_{1,2}$ has two orbits, $\sigma_{1}$ and $\sigma_{2}$,
which are color coded in red and blue, respectively.

\begin{definition}
  \label{def:orbit-specific-counts}
{\rm (Subgraph counts at orbit-specific incidence vertices.)} 
Let $G$ and $H$ be defined as in \Cref{def:global-counts}.
Denote by $H_{\sigma}$ the template $H$ with orbit $\sigma$ designated
for incidence.  Let
$\Gamma_{g}(H_{\sigma}|G)(v) = \{ H'_{\sigma'} \in \Gamma_{g}(H|G)(v) \mid v \in \sigma'\, \mbox{\rm
  and }\, \sigma' \cong \sigma \}$.
Let
$\Gamma_{f}(H_{\sigma}|G)(v) = \{ H'_{\sigma'} \in \Gamma_{f}(H|G)(v)
\mid v \in \sigma'\, \mbox{\rm and }\, \sigma' \cong \sigma \}$.
The gross count (or frequency) of $H$-isomorphic subgraphs with $\sigma$-specific
incidence vertex at $v$ is
$g(H_{\sigma}|G)(v) = |\Gamma_{g}(H_{\sigma}|G)(v)|$.
The net count at $v$ is
$f(H_{\sigma}|G)(v) = |\Gamma_{f}(H_{\sigma}|G)(v)|$.
\end{definition}
We can precisely describe the sub-channel decomposition property as
follows,  $\forall v \!\in\!  V(G)$, 
\begin{equation}
  \label{eq:sigma-frequency-split}
  \begin{array}{rc} 
    f( H|G)(v) \! & \! = \displaystyle 
                 \sum_{\sigma \subset V(H) } f( H_{\sigma} | G )(v), 
    \\ 
    g( H|G)(v) \! &\! = \displaystyle 
                 \sum_{\sigma \subset V(H) } g( H_{\sigma} | G )(v) . 
  \end{array}
\end{equation}
In \Cref{fig:graph-differentiation}, graph $G_1$ and graph $G_2$ have
identical degree sequences and identical $K_{1,2}$-frequency sequences.
They are differentiated by the orbit-specific sequences {w.r.t} $K_{1,2}$.

\begin{table}
  \centering
  \caption{The size sequence of graphlet families ${\cal H}_s$ (graphlets
    with designated orbits) and the size sequence of
    families $\hat{\cal H}_s$ (graphlets without orbit partition),
    $1 \leq s \leq 8$. Each sequence grows
    exponentially with $s$, the number of nodes in a family. }
  \label{tab:number-graphlets}
    \begin{tabular}{lrrrrrrrr}
      \toprule
      $s$
      & 1 & 2 & 3 & 4 & 5 & 6 & 7 & 8 \\
      \midrule
      $|\mathcal{\hat{H}}_{s}|$
      & \num{1} & \num{1} & \num{2} & \num{6} & \num{21}
      & \num{112} & \num{853} & \num{11117}
      \\
      $|\mathcal{H}_{s}|$
      & \num{1} & \num{1} & \num{3} & \num{11} & \num{58}
      & \num{407} & \num{4306} & \num{72489}
      \\
      \bottomrule
    \end{tabular}
\end{table}

By utilizing the topological structure in a pattern template, the
sub-channel encoding approach increases the discriminative power
without or with little increase in the computation cost of subgraph counting.
Assume the templates in ${\cal H} $ are made orbit-specific out of $P$
mutually non-isomorphic patterns, as in
\Cref{def:orbit-specific-counts}. Denote by $a_p$ the number of orbits
in a unique pattern $H_p$. Then, the code length of the frequency vector is
$| {\cal H} | = \sum_{1\leq p \leq P} a_p $. The code length
is greater than $P$. The difference is an indicator of the
increased coding capacity for differentiating local structures.
The first use of orbit-specific subgraph templates is seen in the
original work of biological network analysis with graphlet degree
distributions by Pr{\v z}ulj, Corneil and Jurisica in
2004~\cite{przulj2004}.  The above description by the present work
elucidates, explains and characterizes the orbit-specific templates as
sub-channel decomposition from the multi-channel encoding perspective.

By multi-channel graph encoding we refer to the use of multiple templates
with optional use of the sub-channel decomposition approach.  There are
deeper potential and additional benefits with multi-channel graph
coding. 
  \begin{inparaenum}[(a)]
  \item The frequency features are derived, self-learned, from the
    source graph only. They can be joined with other attributes, or
    used to validate other features learned by different approaches.
    \item The frequencies of a source graph are coupled by vertex
    collocation.  This ``spatial'' collocation property increases the
    discriminative power. The frequency sequences are not independent
    of each other. Even when two different source graphs have
    identical sequences with respect to every template individually,
    their difference may be detected by comparing their vector-valued
    frequency sequences in lexicographical order.
    \item The frequency vectors can be used to assess or detect self
    similarities within a source graph, among vertices or vertex
    subsets.
  \item By using multiple connected templates, we have effectively
    included the case in which a template is composed of more than one
    connected components, especially for motif detection or discovery.
  \end{inparaenum}

\begin{figure}[!t]
  \centering
  \newcommand{\figGraphletWidth}{0.07\linewidth}
  \newcommand{\familyID}{0}
  \newcommand{\graphletID}{0}
  \newcommand{\orbitID}{0}
  \renewcommand{\familyID}{1}
  \hspace*{\fill}
  \renewcommand{\graphletID}{1}\renewcommand{\orbitID}{1}
\begin{subfigure}{\figGraphletWidth}
  \centering
  \includegraphics[width=\linewidth]{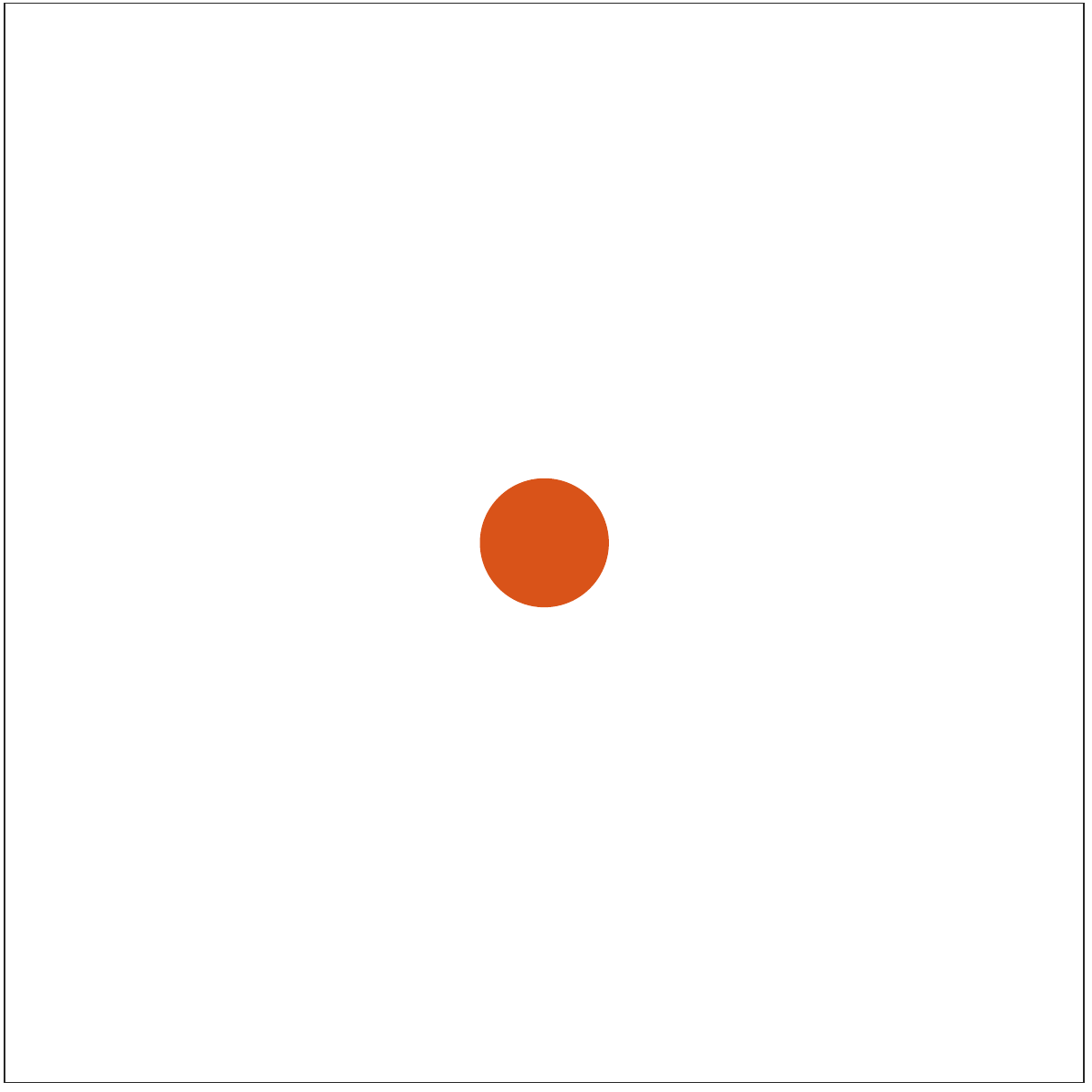}
  \\[-0.5em]
  \caption*{\tiny $H_{\familyID,\graphletID,\orbitID}$}
\end{subfigure}
\hspace*{\fill}
  \renewcommand{\familyID}{2}
  \renewcommand{\graphletID}{1}\renewcommand{\orbitID}{1}
\begin{subfigure}{\figGraphletWidth}
  \centering
  \includegraphics[width=\linewidth]{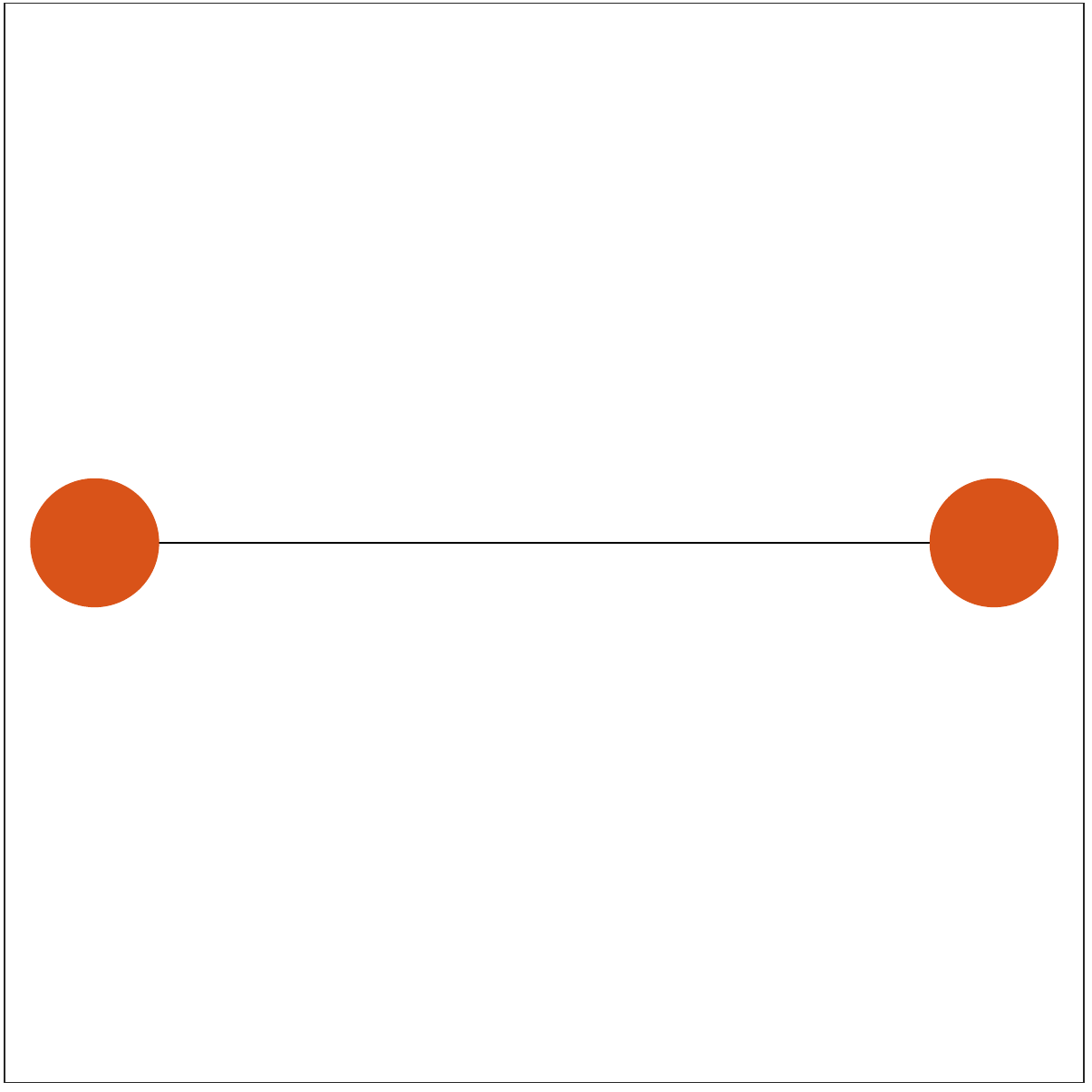}
  \\[-0.5em]
  \caption*{\tiny $H_{\familyID,\graphletID,\orbitID}$}
\end{subfigure}
\hspace*{\fill}
  \renewcommand{\familyID}{3}
  \renewcommand{\graphletID}{1}\renewcommand{\orbitID}{1}
\begin{subfigure}{\figGraphletWidth}
  \centering
  \includegraphics[width=\linewidth]{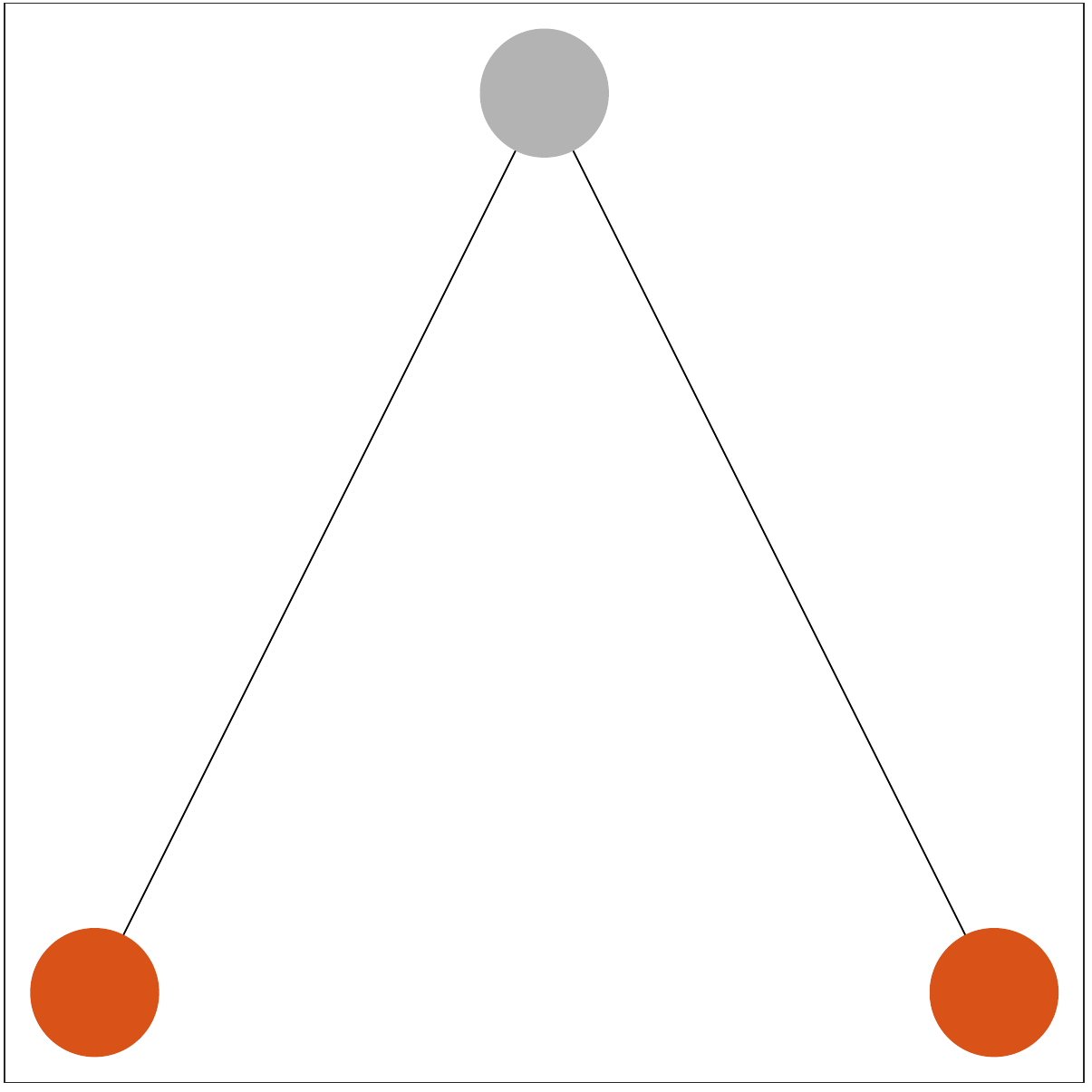}
  \\[-0.5em]
  \caption*{\tiny $H_{\familyID,\graphletID,\orbitID}$}
\end{subfigure}
\hspace*{\fill}
   \renewcommand{\graphletID}{1}\renewcommand{\orbitID}{2}
\begin{subfigure}{\figGraphletWidth}
  \centering
  \includegraphics[width=\linewidth]{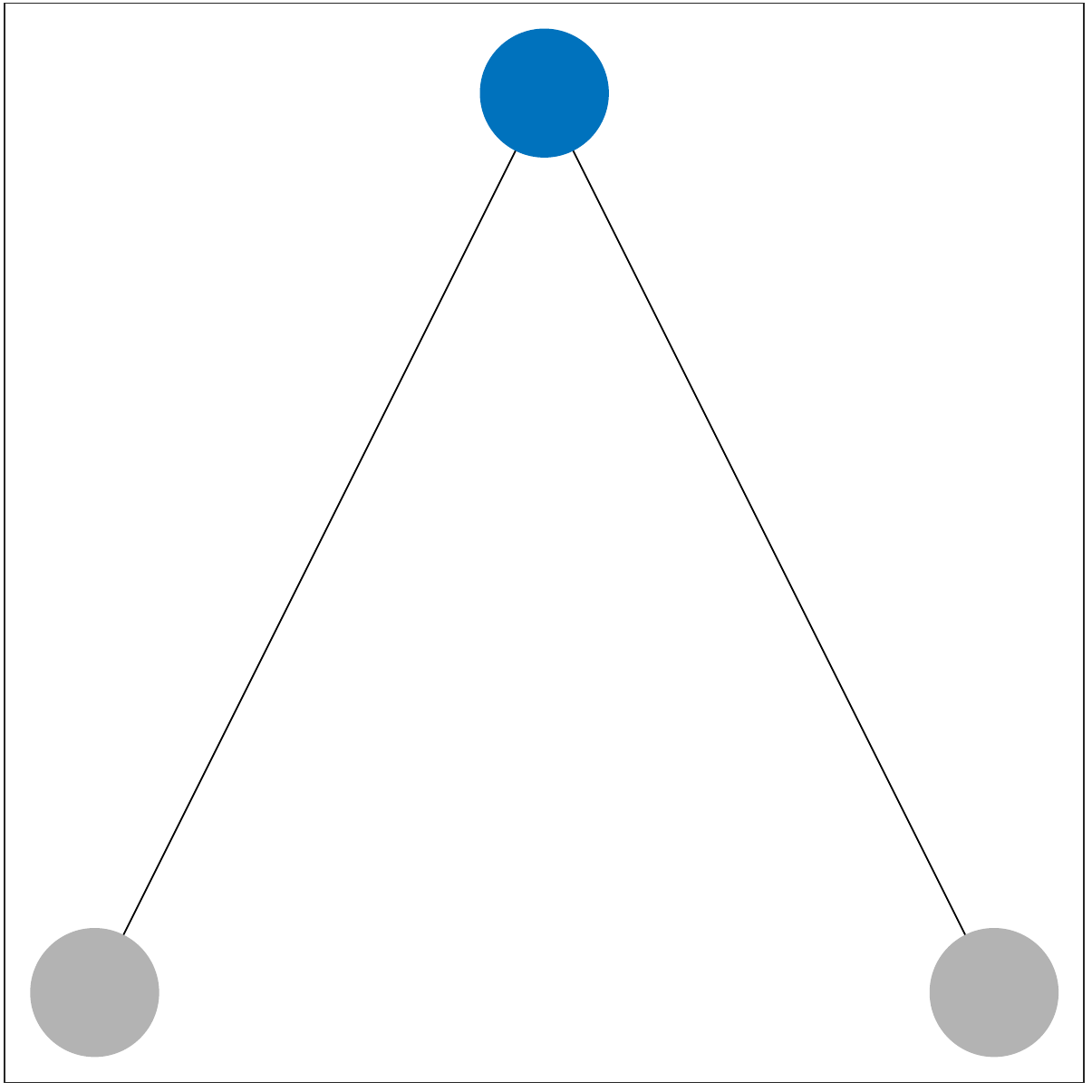}
  \\[-0.5em]
  \caption*{\tiny $H_{\familyID,\graphletID,\orbitID}$}
\end{subfigure}
\hspace*{\fill}
  \renewcommand{\graphletID}{2}\renewcommand{\orbitID}{1}
\begin{subfigure}{\figGraphletWidth}
  \centering
  \includegraphics[width=\linewidth]{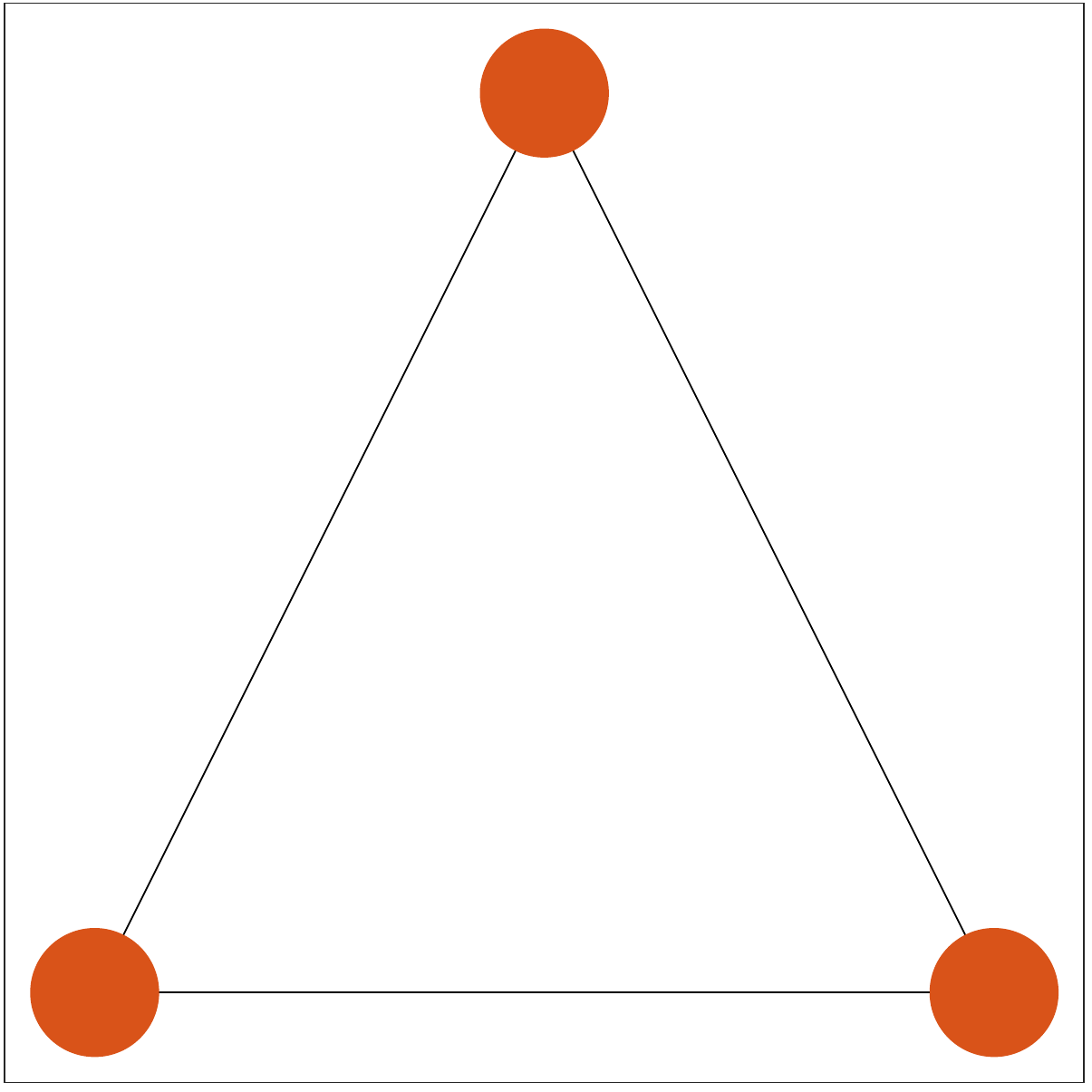}
  \\[-0.5em]
  \caption*{\tiny $H_{\familyID,\graphletID,\orbitID}$}
\end{subfigure}
\hspace*{\fill}
  \renewcommand{\familyID}{4}
  \renewcommand{\graphletID}{1}\renewcommand{\orbitID}{1}
\begin{subfigure}{\figGraphletWidth}
  \centering
  \includegraphics[width=\linewidth]{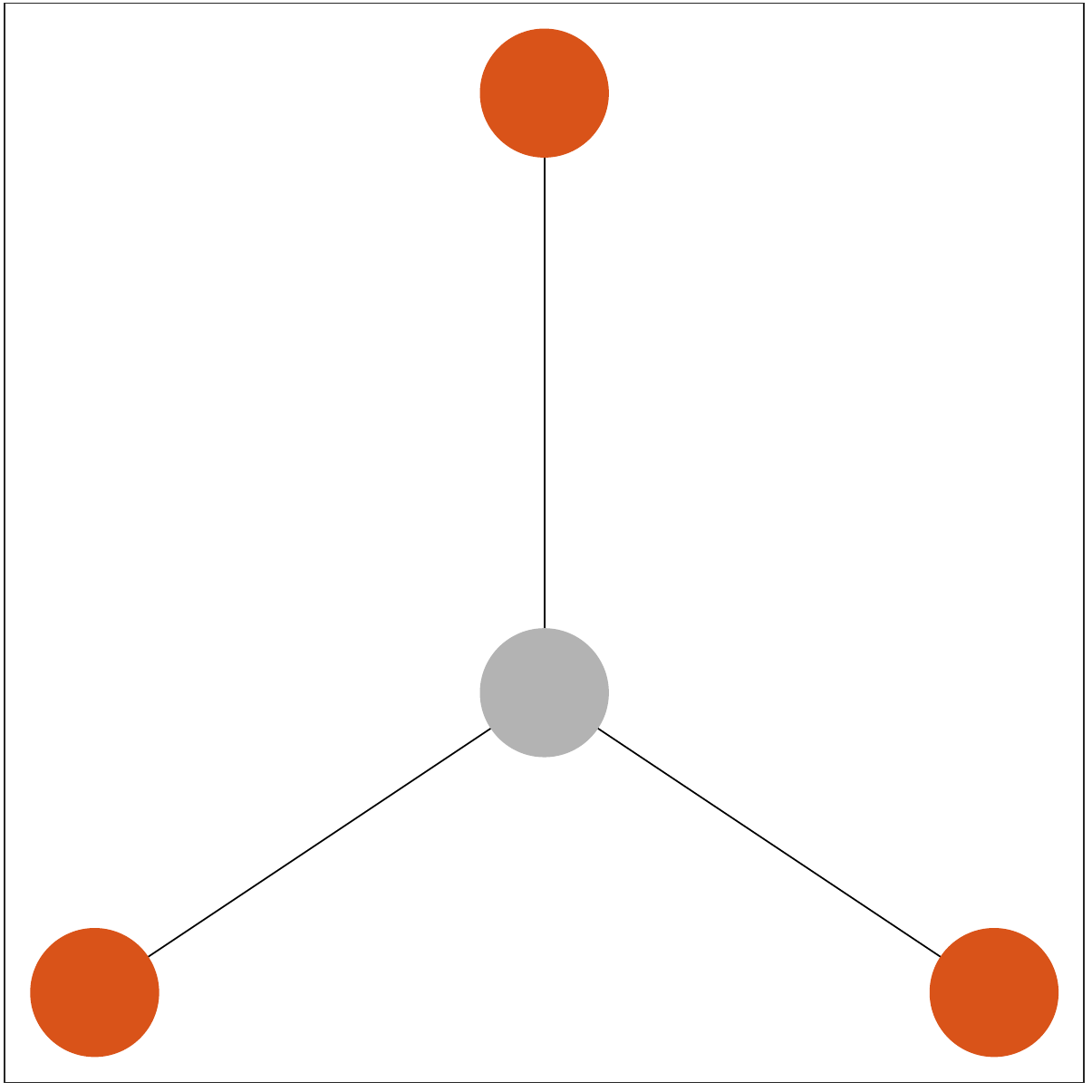}
  \\[-0.5em]
  \caption*{\tiny $H_{\familyID,\graphletID,\orbitID}$}
\end{subfigure}
\hspace*{\fill}
   \renewcommand{\graphletID}{1}\renewcommand{\orbitID}{2}
\begin{subfigure}{\figGraphletWidth}
  \centering
  \includegraphics[width=\linewidth]{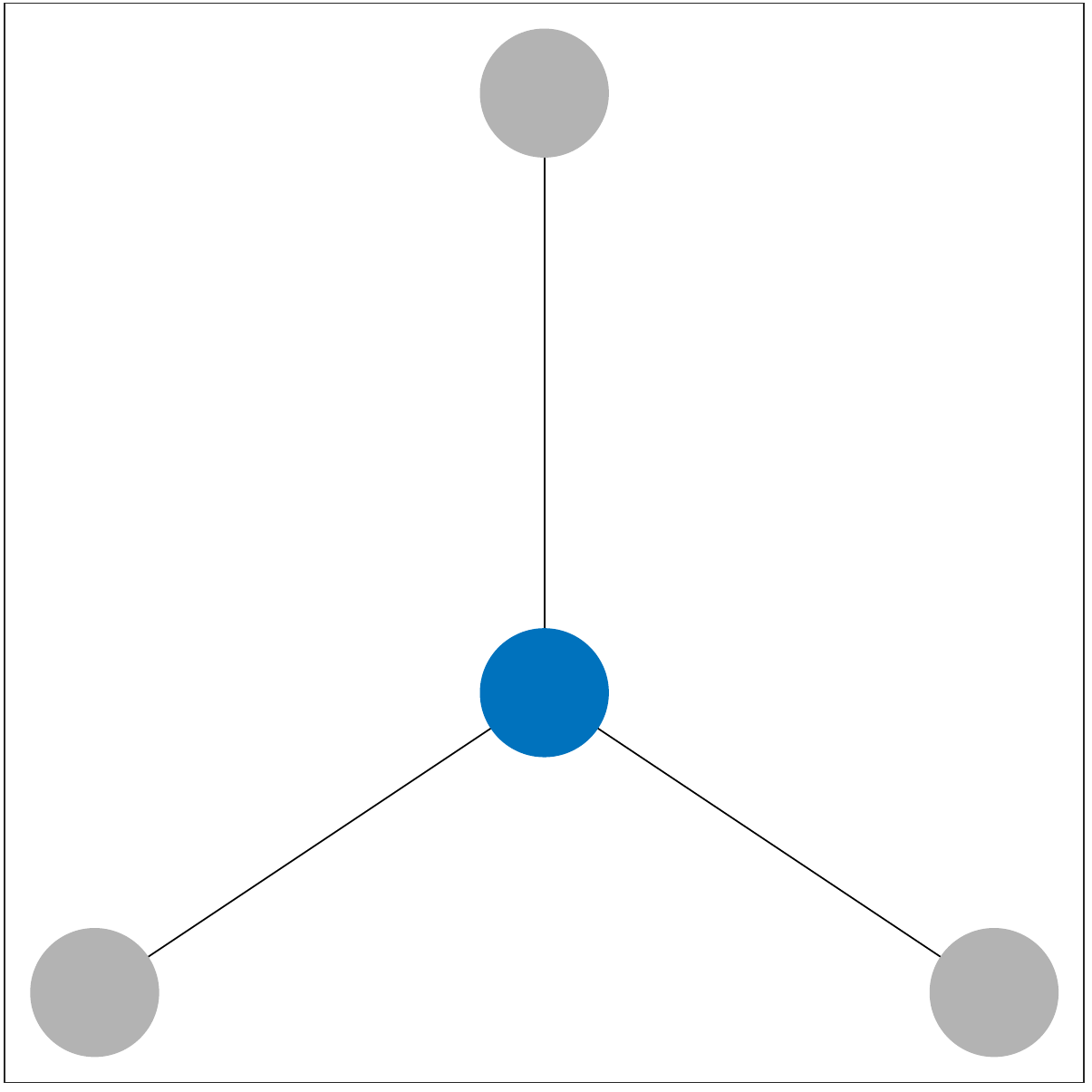}
  \\[-0.5em]
  \caption*{\tiny $H_{\familyID,\graphletID,\orbitID}$}
\end{subfigure}
\hspace*{\fill}
   \renewcommand{\graphletID}{2}\renewcommand{\orbitID}{1}
\begin{subfigure}{\figGraphletWidth}
  \centering
  \includegraphics[width=\linewidth]{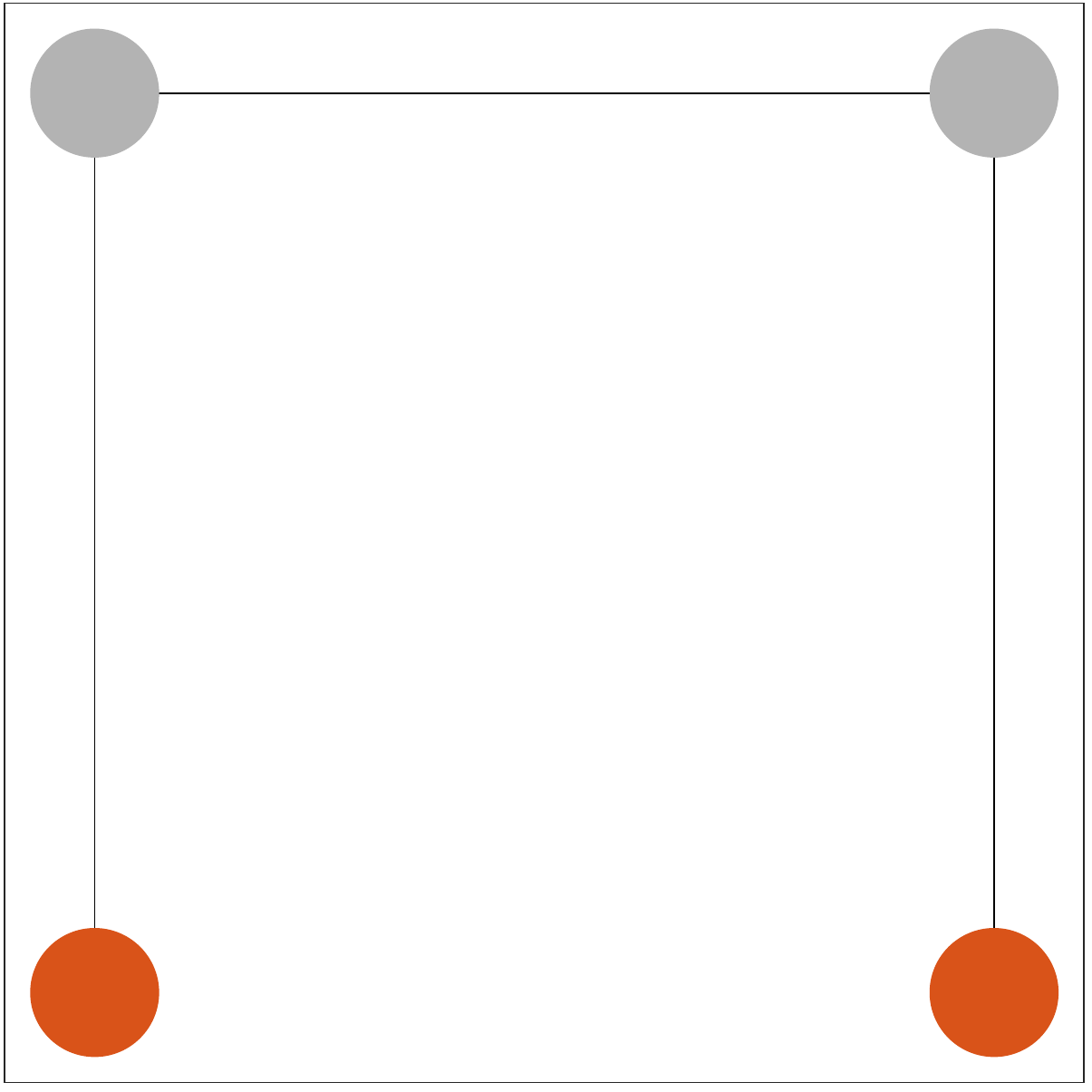}
  \\[-0.5em]
  \caption*{\tiny $H_{\familyID,\graphletID,\orbitID}$}
\end{subfigure}
\hspace*{\fill}
   \\[0.5em]
  \hspace*{\fill}
  \renewcommand{\graphletID}{2}\renewcommand{\orbitID}{2}
\begin{subfigure}{\figGraphletWidth}
  \centering
  \includegraphics[width=\linewidth]{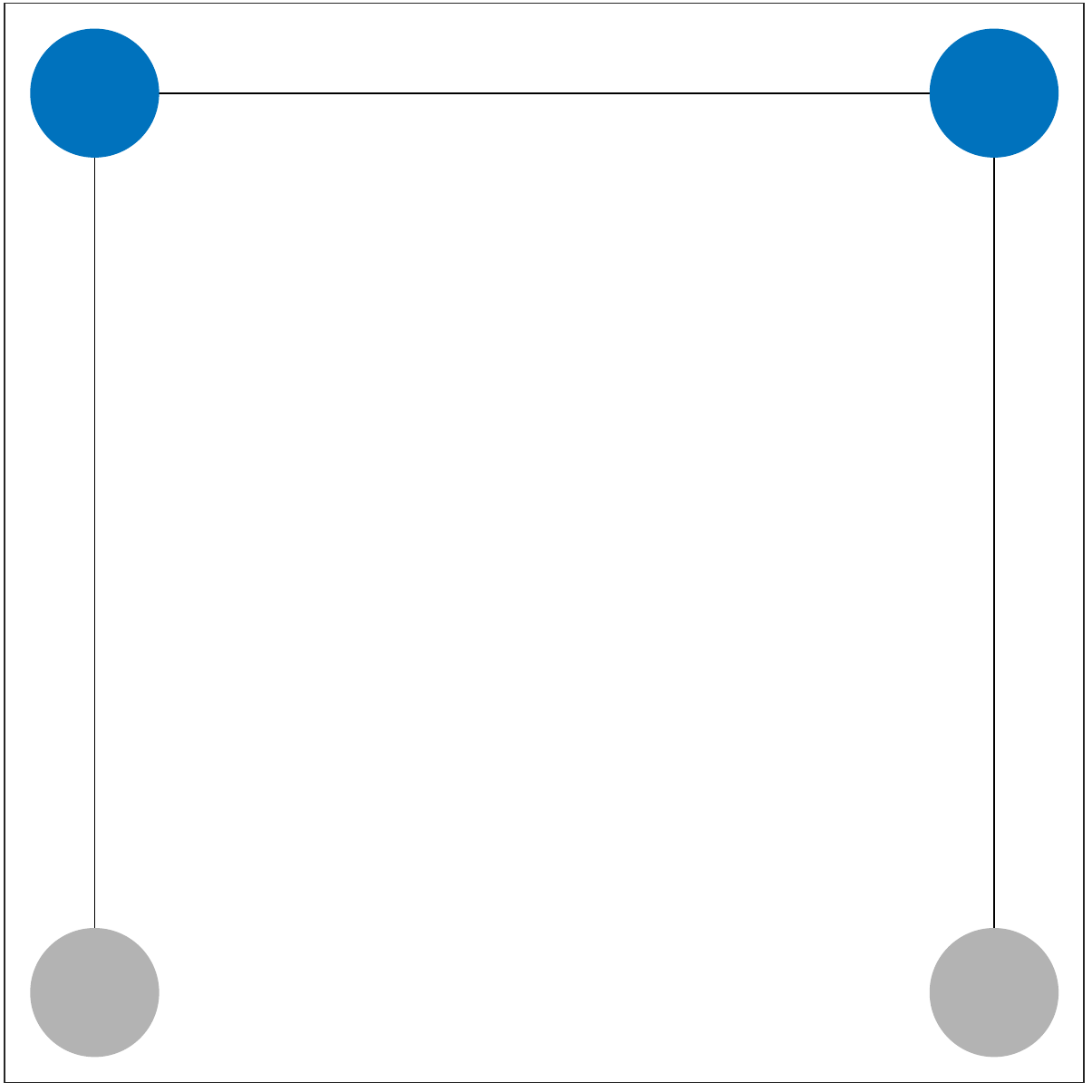}
  \\[-0.5em]
  \caption*{\tiny $H_{\familyID,\graphletID,\orbitID}$}
\end{subfigure}
\hspace*{\fill}
   \renewcommand{\graphletID}{3}\renewcommand{\orbitID}{1}
\begin{subfigure}{\figGraphletWidth}
  \centering
  \includegraphics[width=\linewidth]{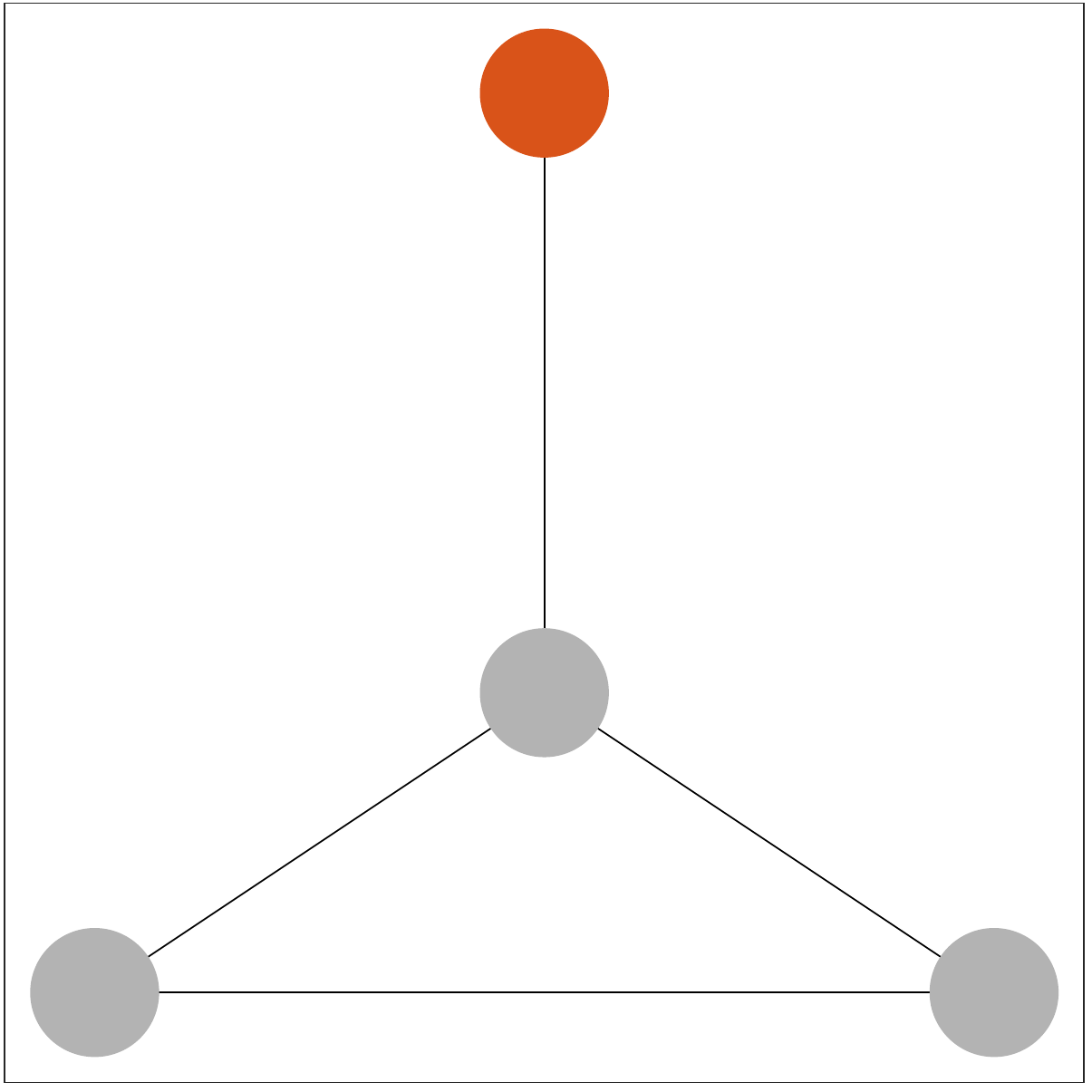}
  \\[-0.5em]
  \caption*{\tiny $H_{\familyID,\graphletID,\orbitID}$}
\end{subfigure}
\hspace*{\fill}
   \renewcommand{\graphletID}{3}\renewcommand{\orbitID}{2}
\begin{subfigure}{\figGraphletWidth}
  \centering
  \includegraphics[width=\linewidth]{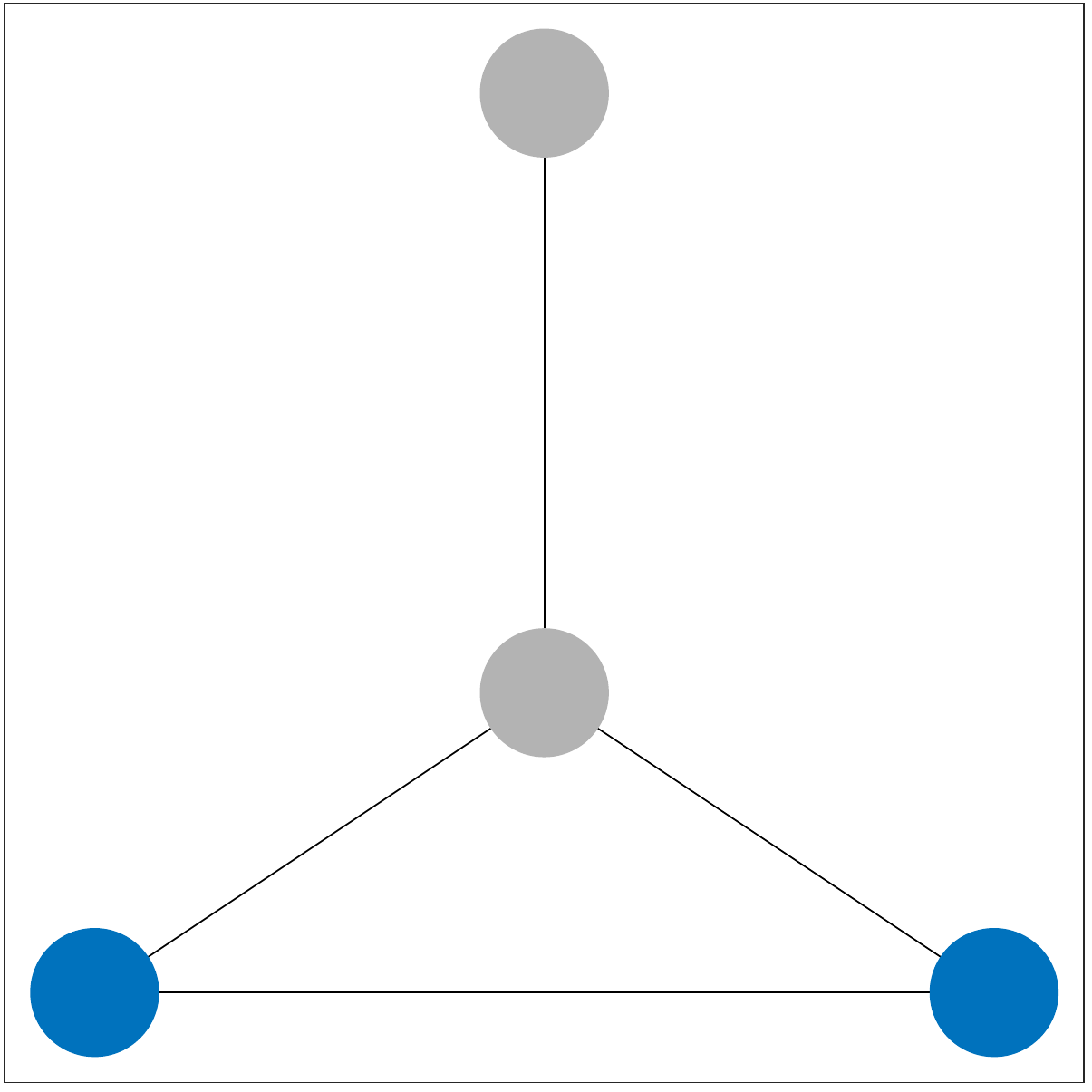}
  \\[-0.5em]
  \caption*{\tiny $H_{\familyID,\graphletID,\orbitID}$}
\end{subfigure}
\hspace*{\fill}
   \renewcommand{\graphletID}{3}\renewcommand{\orbitID}{3}
\begin{subfigure}{\figGraphletWidth}
  \centering
  \includegraphics[width=\linewidth]{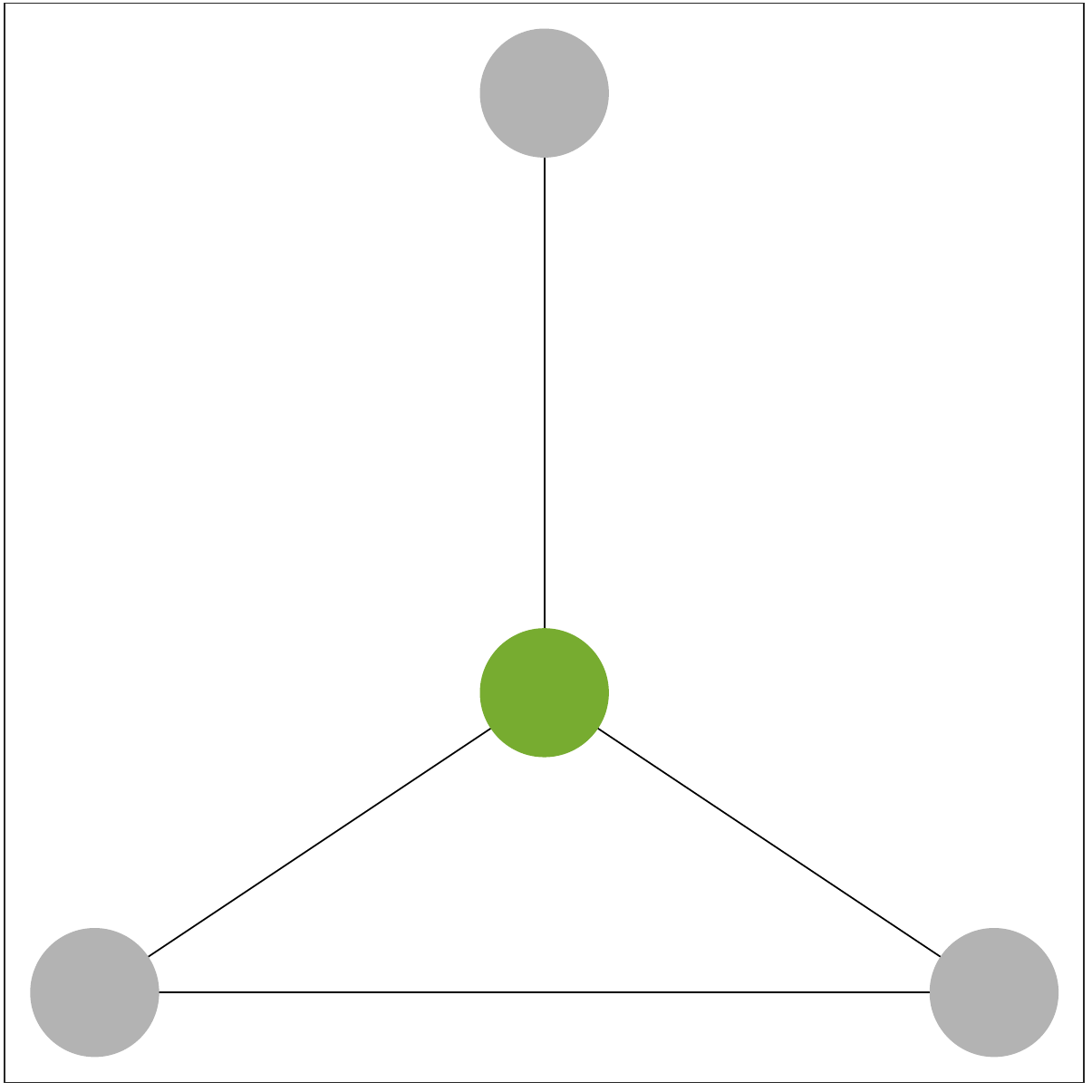}
  \\[-0.5em]
  \caption*{\tiny $H_{\familyID,\graphletID,\orbitID}$}
\end{subfigure}
\hspace*{\fill}
   \renewcommand{\graphletID}{4}\renewcommand{\orbitID}{1}
\begin{subfigure}{\figGraphletWidth}
  \centering
  \includegraphics[width=\linewidth]{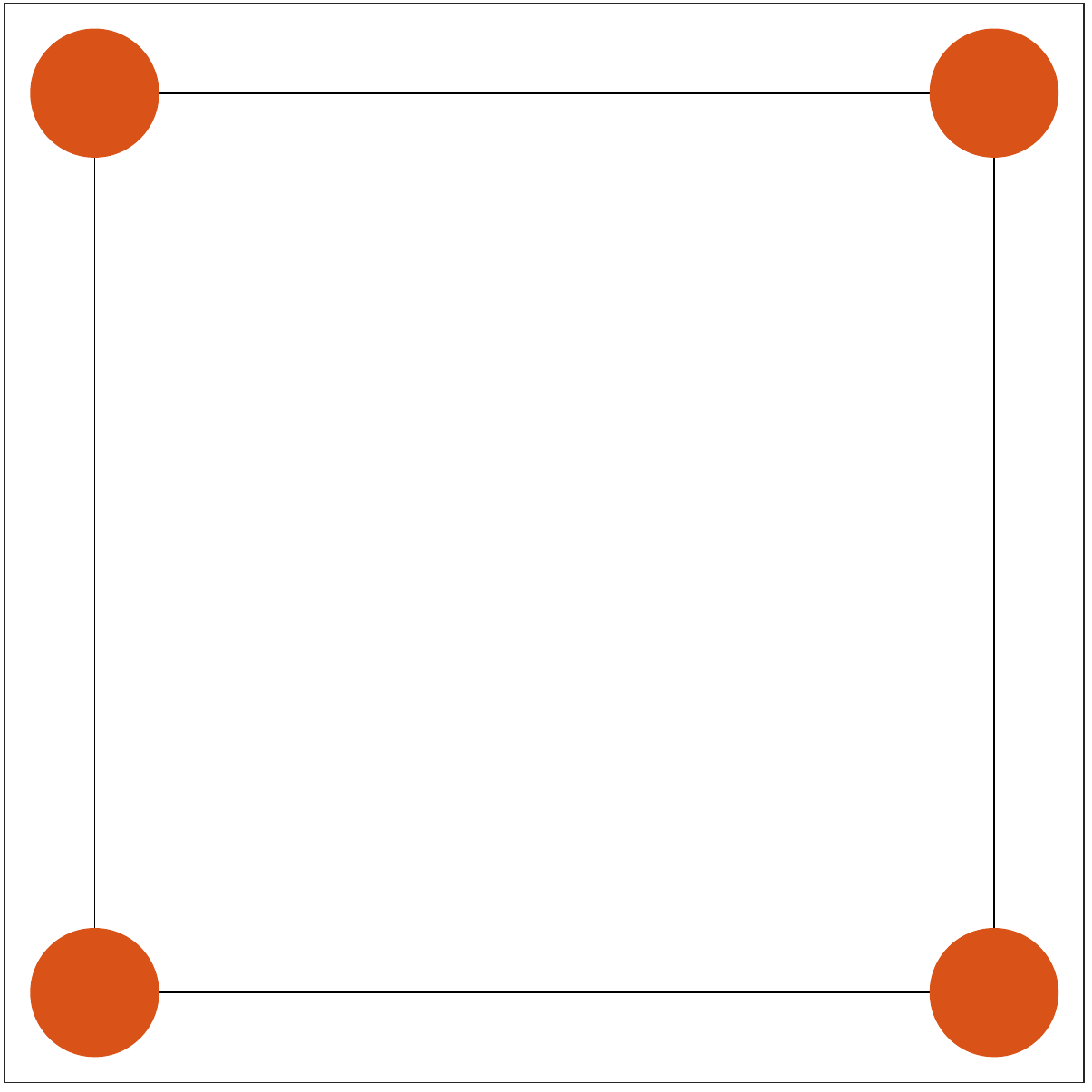}
  \\[-0.5em]
  \caption*{\tiny $H_{\familyID,\graphletID,\orbitID}$}
\end{subfigure}
\hspace*{\fill}
   \renewcommand{\graphletID}{5}\renewcommand{\orbitID}{1}
\begin{subfigure}{\figGraphletWidth}
  \centering
  \includegraphics[width=\linewidth]{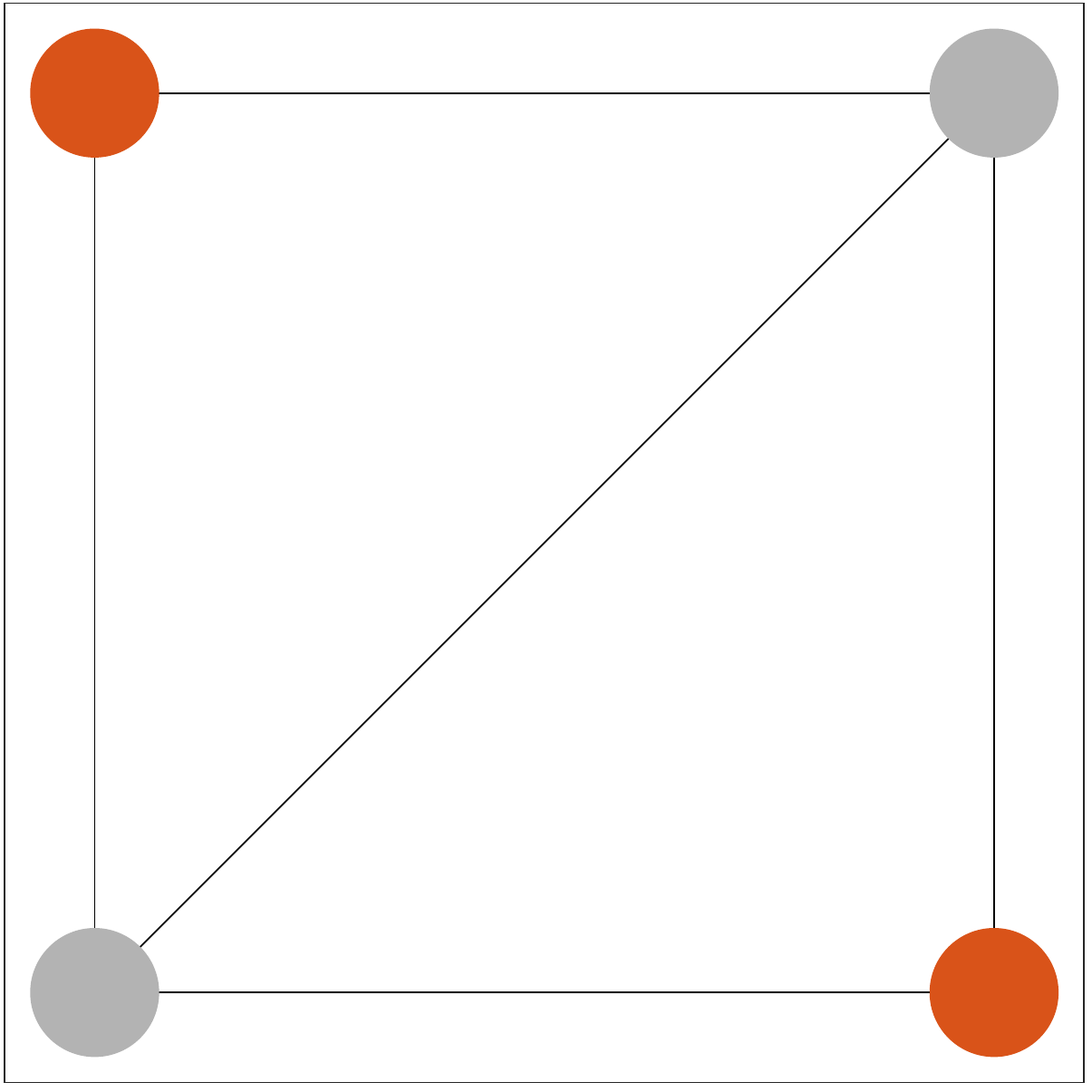}
  \\[-0.5em]
  \caption*{\tiny $H_{\familyID,\graphletID,\orbitID}$}
\end{subfigure}
\hspace*{\fill}
   \renewcommand{\graphletID}{5}\renewcommand{\orbitID}{2}
\begin{subfigure}{\figGraphletWidth}
  \centering
  \includegraphics[width=\linewidth]{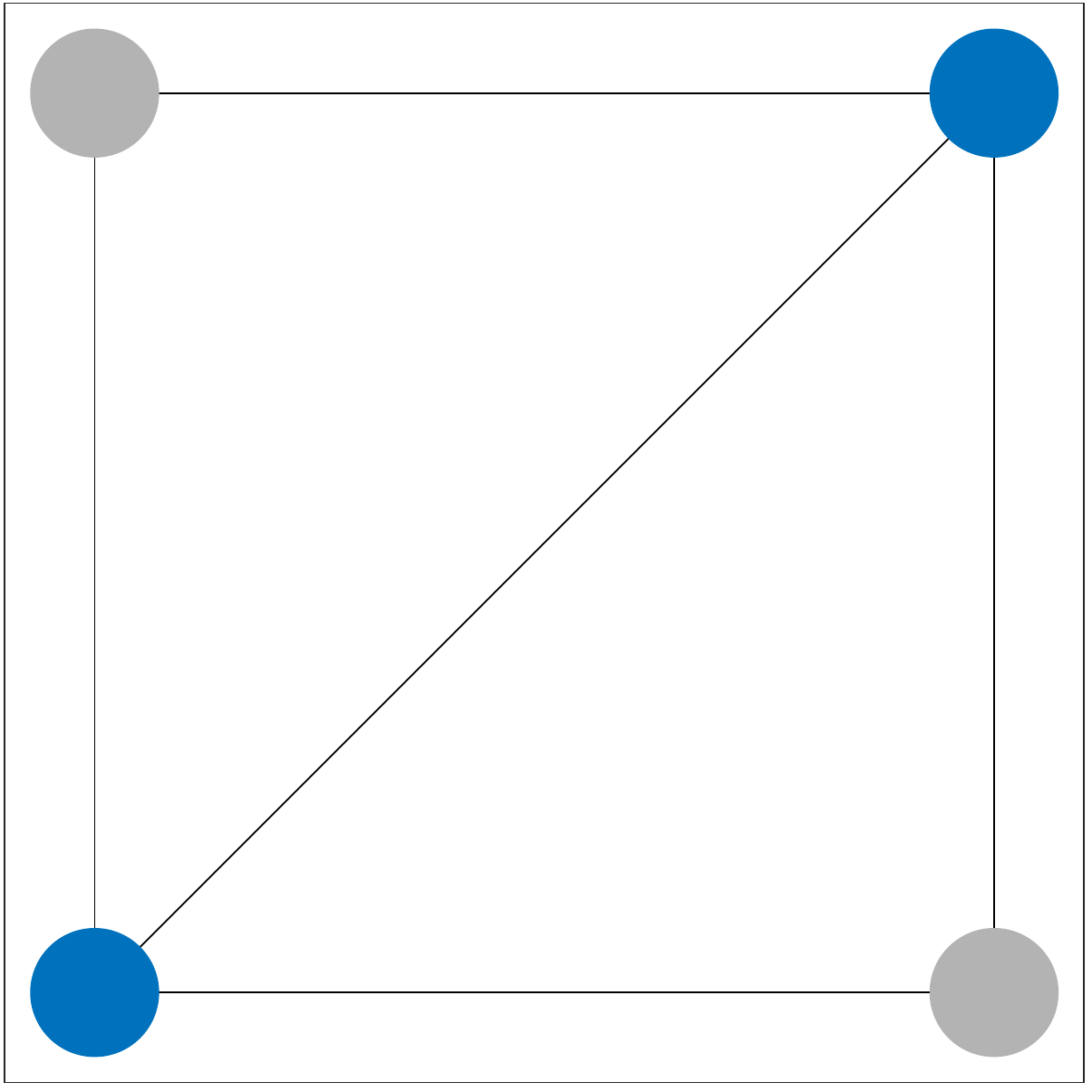}
  \\[-0.5em]
  \caption*{\tiny $H_{\familyID,\graphletID,\orbitID}$}
\end{subfigure}
\hspace*{\fill}
   \renewcommand{\graphletID}{6}\renewcommand{\orbitID}{1}
\begin{subfigure}{\figGraphletWidth}
  \centering
  \includegraphics[width=\linewidth]{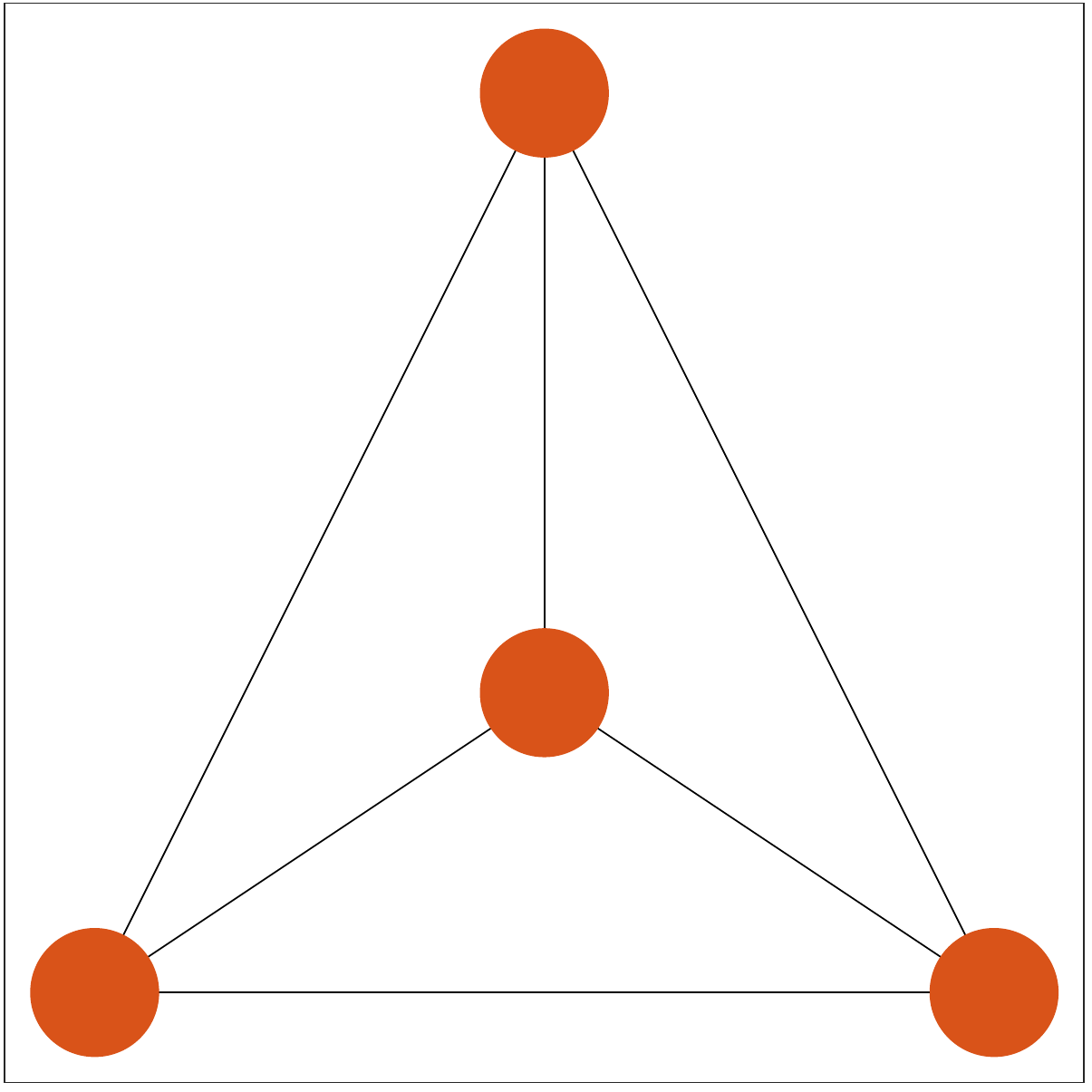}
  \\[-0.5em]
  \caption*{\tiny $H_{\familyID,\graphletID,\orbitID}$}
\end{subfigure}
\hspace*{\fill}
  \\[0.5em]
  \caption{%
    The dictionary of graphlets in the first four families ${\cal H}_s$,
    $s=1,2,3,4$. 
    Each graphlet $H$ is identified by an index triplet
    $(s,p,\sigma)$, with $s = n(H)$, $p$ indexed to a unique
    topological pattern, and $\sigma$ to a specific orbit, as
    described in \Cref{subsec:graphlet-dictionary}. Orbits indexed by
    $1$, $2$ and $3$ are color coded red, blue and green,
    respectively.  The graphlets are placed from left to right, top to
    bottom, by the ordering scheme {\sc seira} in \Cref{subsec:graphlet-dictionary}. }
  \label{fig:16-graphlet-portraits}
\end{figure}

\subsection{Graphlet lattice neighborhoods}
\label{subsec:graphlet-dictionary}

We focus on a system of multi-channel encoding graph elements known as
graphlets. 
In this section, we
give a clarified description of graphlets independent of source graphs
and graphlet frequencies on any given source graph.
More importantly, we introduce intrinsic topological relations among
graphlets in the language of graph
theory~\cite{diestel2017,george1993} and lattice
theory~\cite{birkhoff1938,birkhoff1948}. These topological
relationships are the foundation of the algebraic and quantitative
relations in graphlet frequencies we further uncover and present in
the rest of the paper.

\begin{figure*}[!t]
  \centering
  \newcommand{\figGraphletWidth}{0.07\linewidth}
  \newcommand{\familyID}{0}
  \newcommand{\graphletID}{0}
  \renewcommand{\familyID}{5}
  \hspace*{\fill}
  \renewcommand{\graphletID}{1}%
\begin{subfigure}{\figGraphletWidth}
  \includegraphics[width=\linewidth]{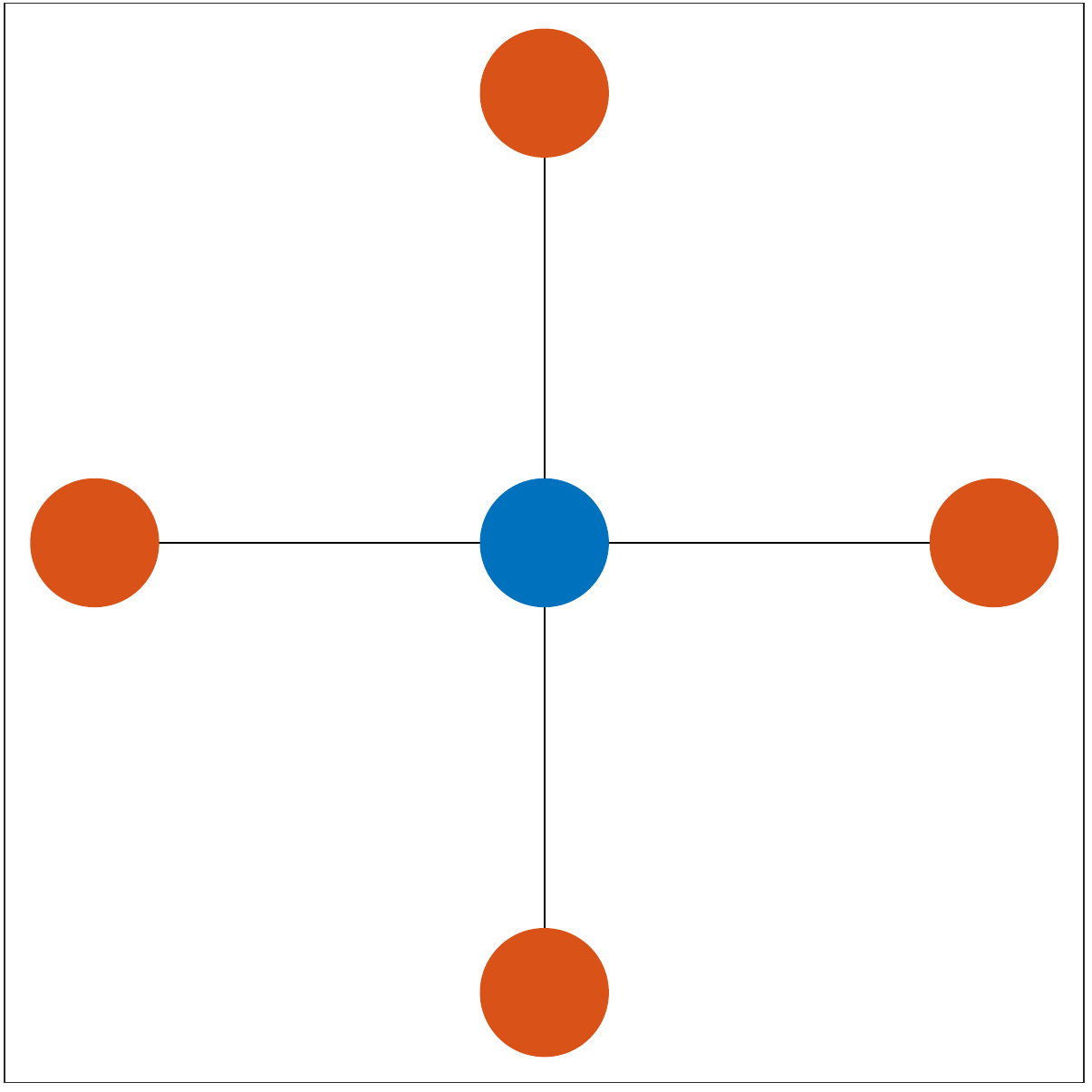}
  \caption*{$H_{\familyID,\graphletID}$}
\end{subfigure}
\hspace*{\fill}
   \renewcommand{\graphletID}{2}%
\begin{subfigure}{\figGraphletWidth}
  \includegraphics[width=\linewidth]{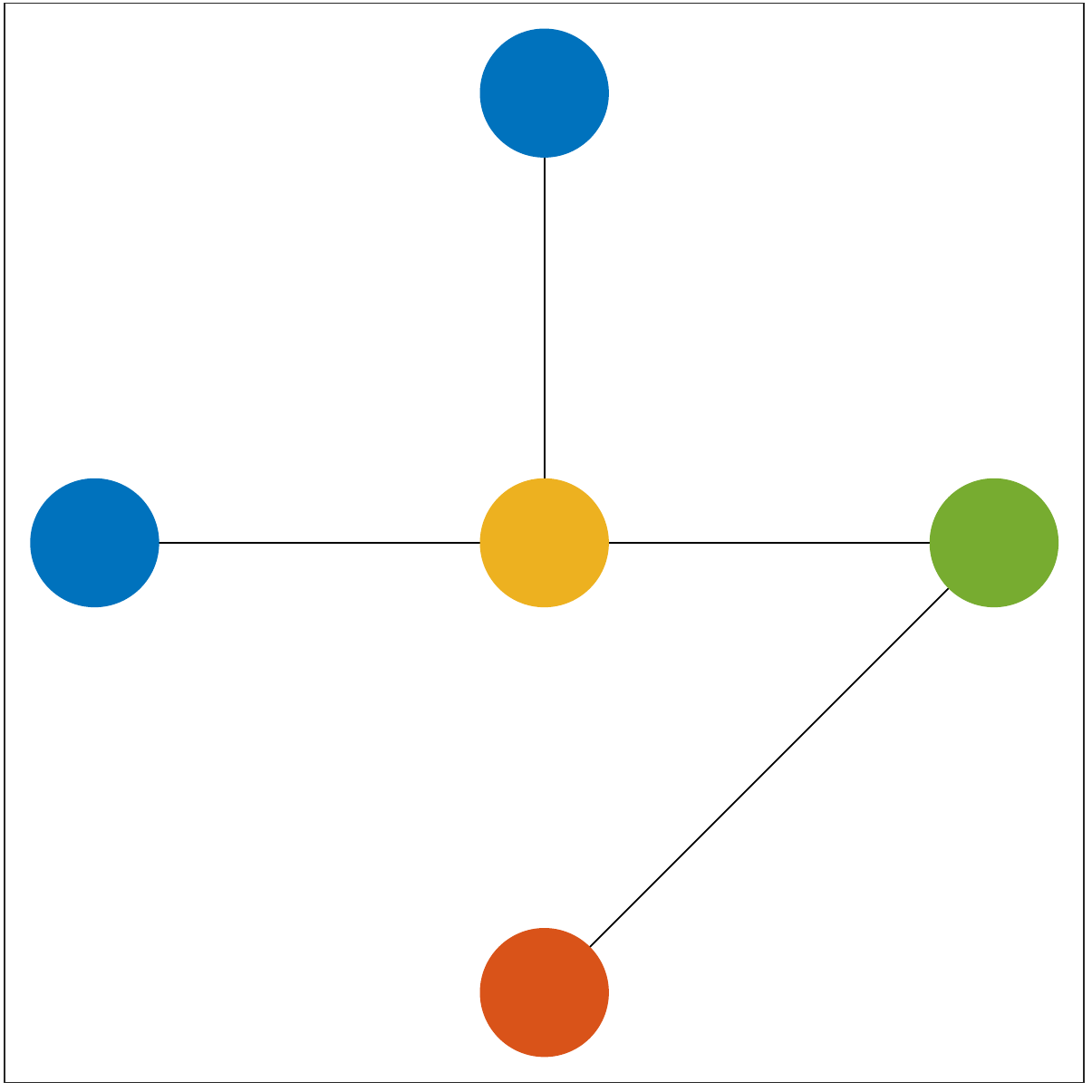}
  \caption*{$H_{\familyID,\graphletID}$}
\end{subfigure}
\hspace*{\fill}
   \renewcommand{\graphletID}{3}%
\begin{subfigure}{\figGraphletWidth}
  \includegraphics[width=\linewidth]{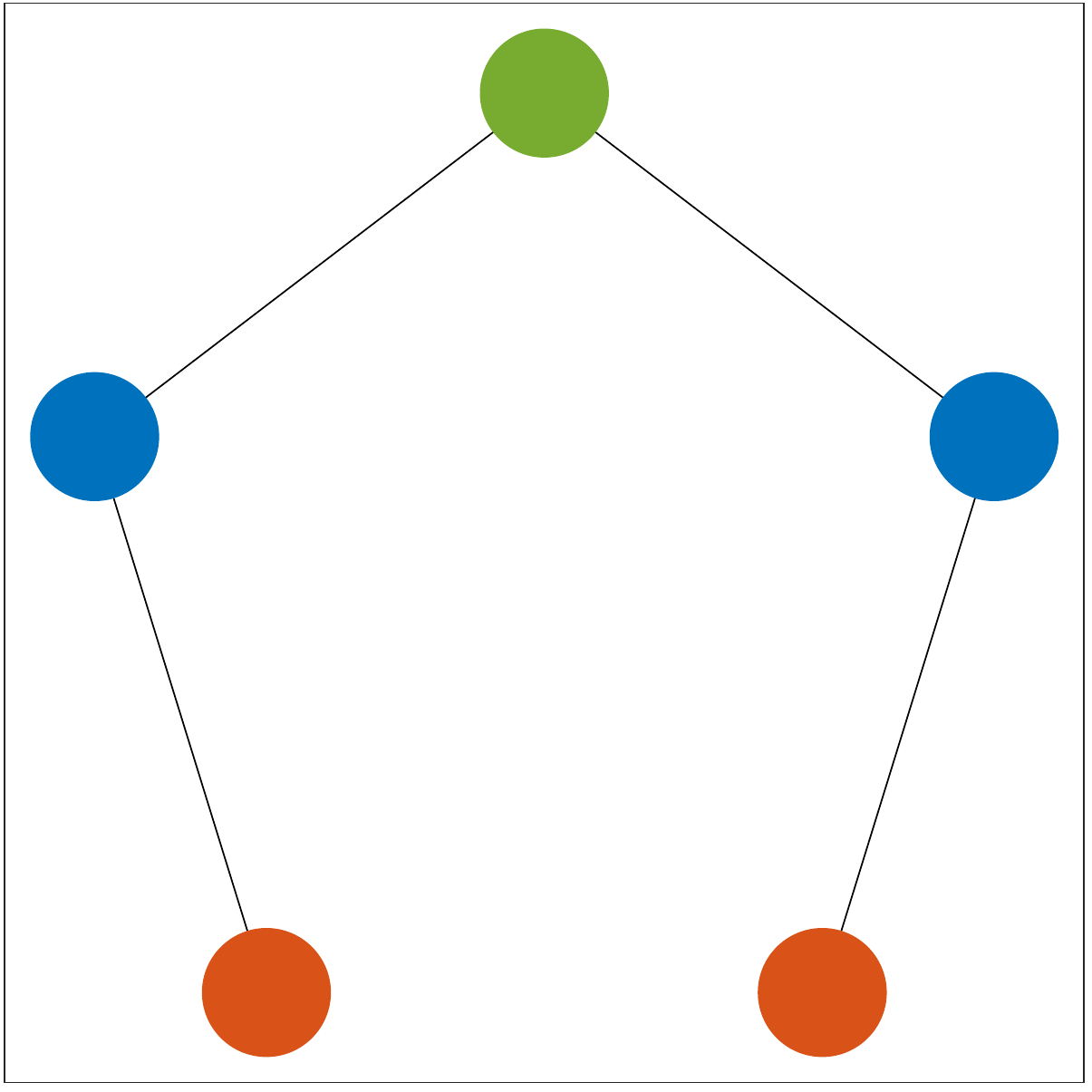}
  \caption*{$H_{\familyID,\graphletID}$}
\end{subfigure}
\hspace*{\fill}
   \renewcommand{\graphletID}{4}%
\begin{subfigure}{\figGraphletWidth}
  \includegraphics[width=\linewidth]{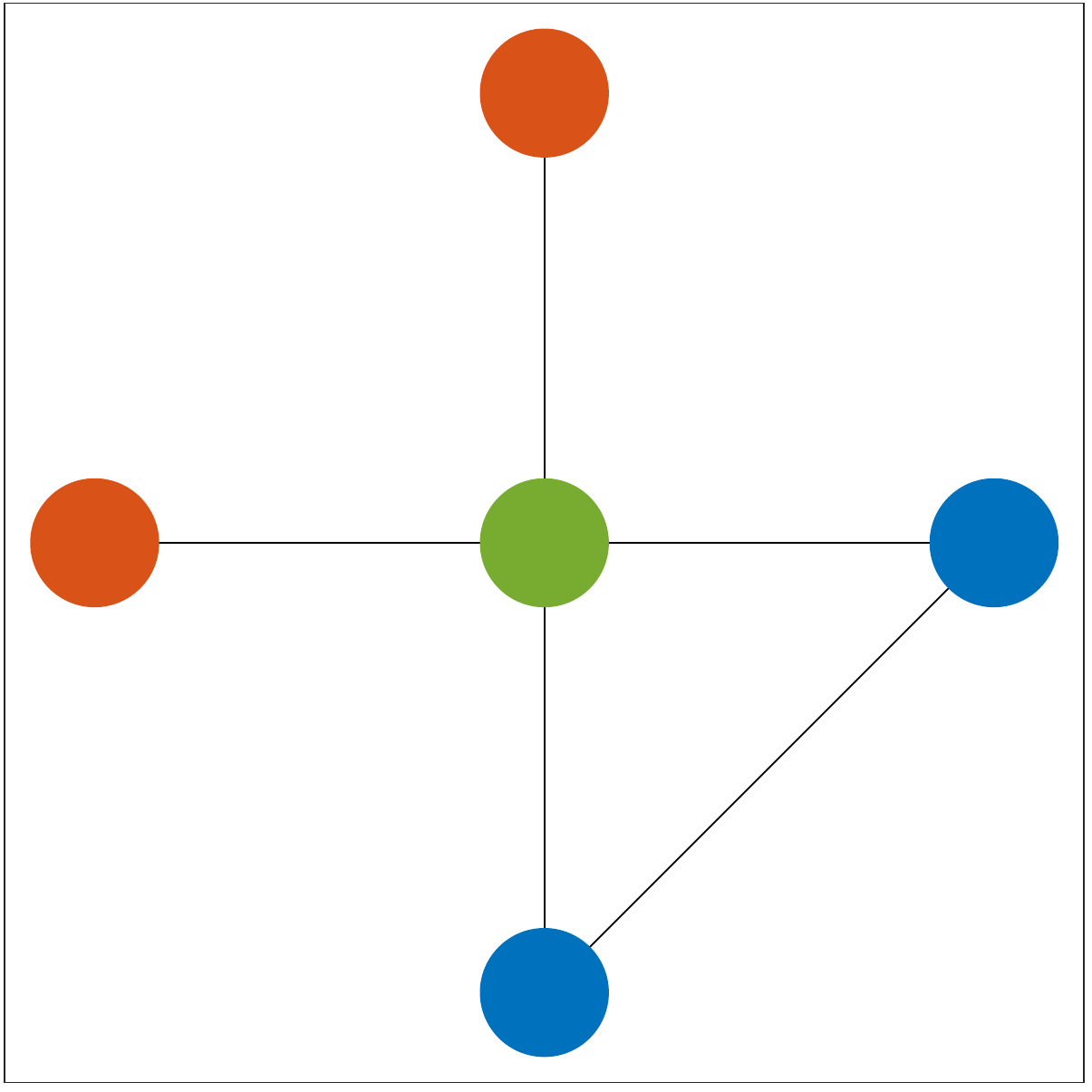}
  \caption*{$H_{\familyID,\graphletID}$}
\end{subfigure}
\hspace*{\fill}
   \renewcommand{\graphletID}{5}%
\begin{subfigure}{\figGraphletWidth}
  \includegraphics[width=\linewidth]{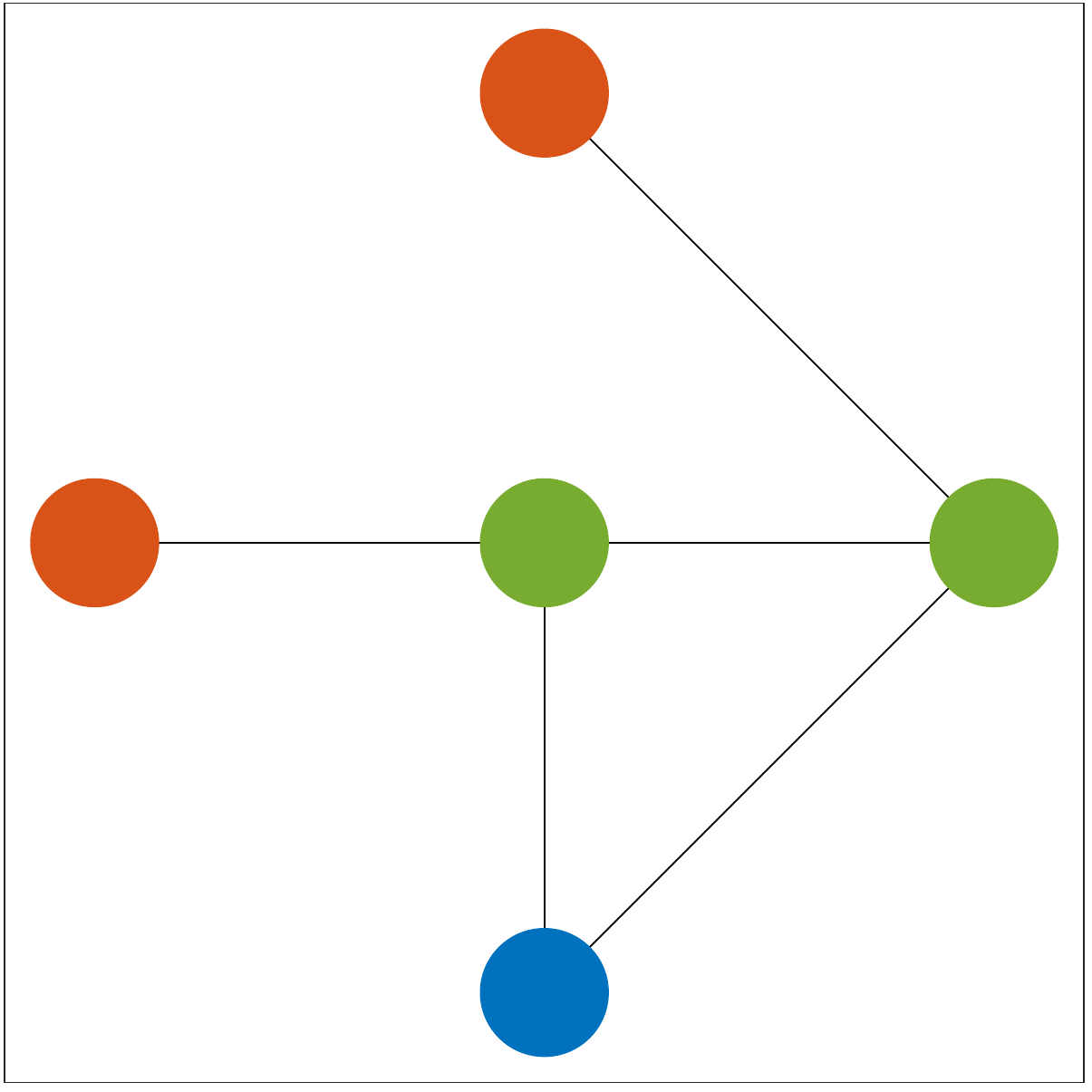}
  \caption*{$H_{\familyID,\graphletID}$}
\end{subfigure}
\hspace*{\fill}
   \renewcommand{\graphletID}{6}%
\begin{subfigure}{\figGraphletWidth}
  \includegraphics[width=\linewidth]{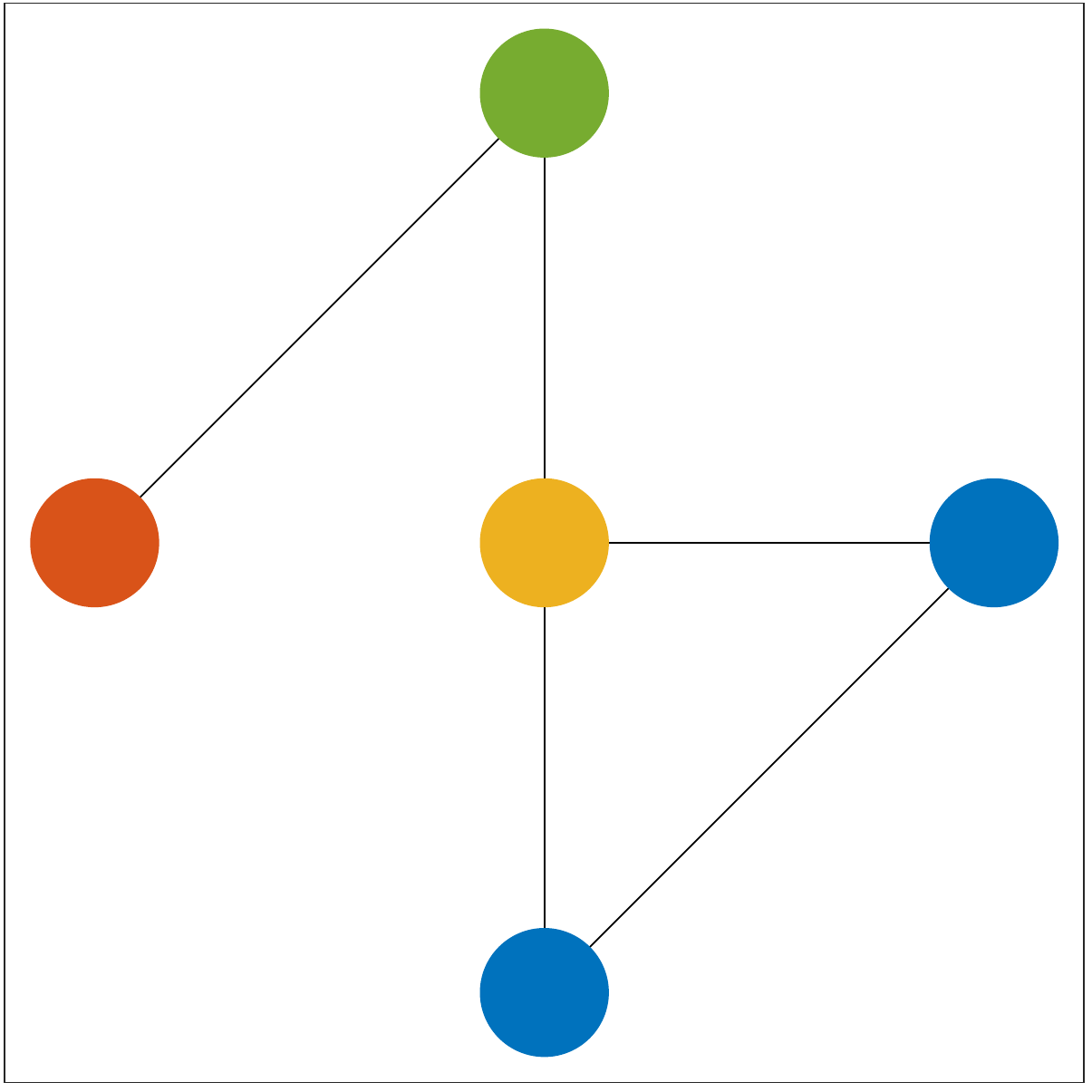}
  \caption*{$H_{\familyID,\graphletID}$}
\end{subfigure}
\hspace*{\fill}
   \renewcommand{\graphletID}{7}%
\begin{subfigure}{\figGraphletWidth}
  \includegraphics[width=\linewidth]{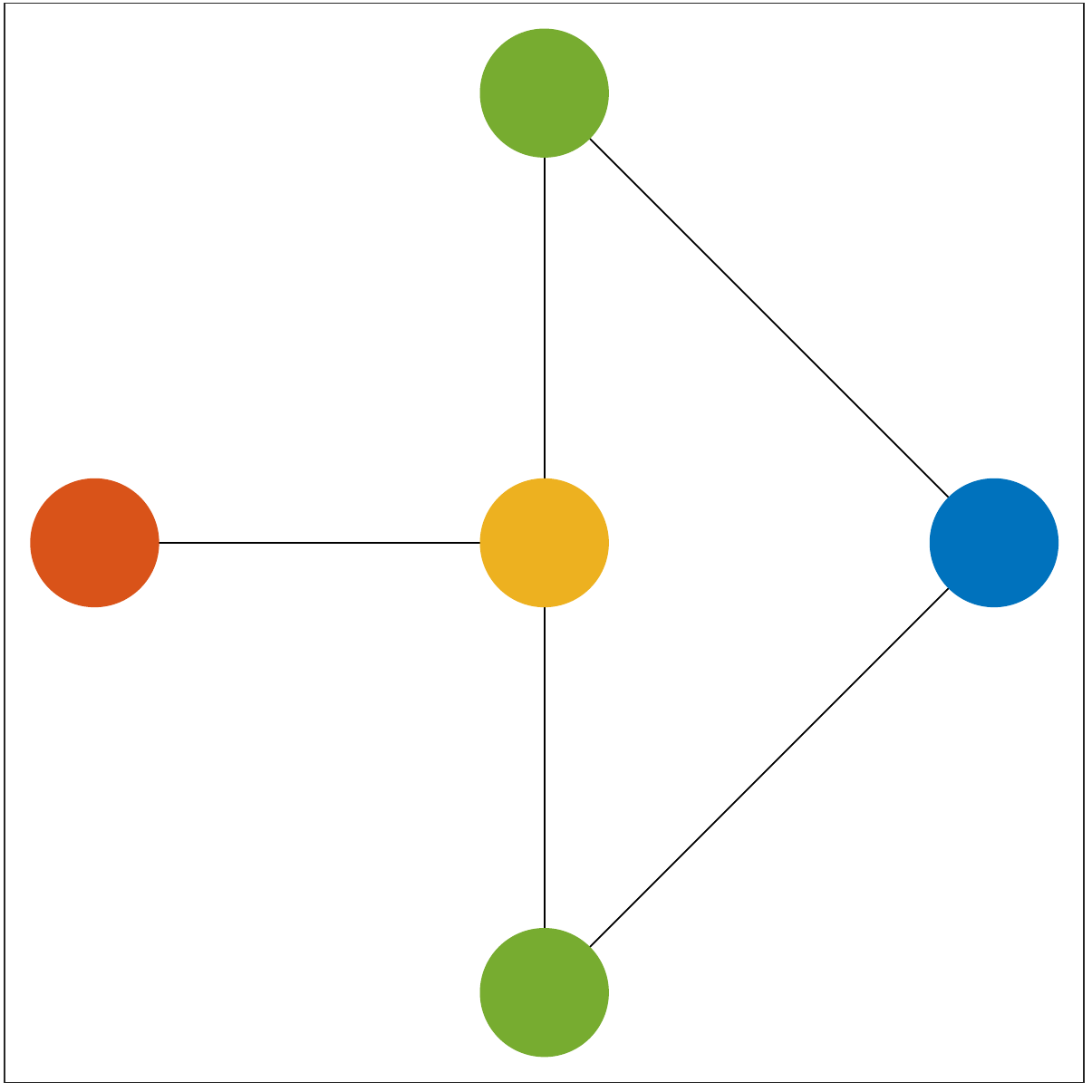}
  \caption*{$H_{\familyID,\graphletID}$}
\end{subfigure}
\hspace*{\fill}
   \\[0.5em]
  \hspace*{\fill}
  \renewcommand{\graphletID}{8}%
\begin{subfigure}{\figGraphletWidth}
  \includegraphics[width=\linewidth]{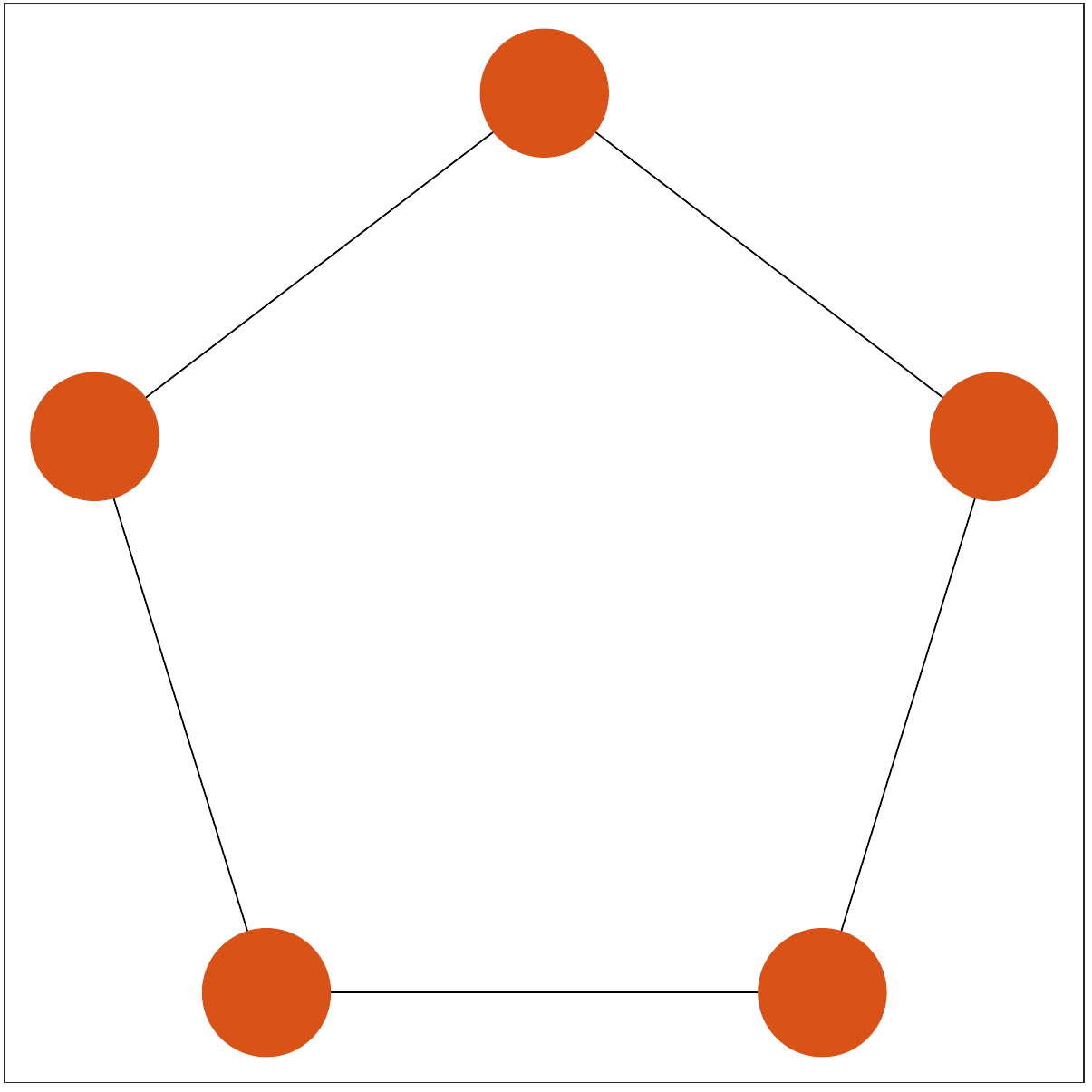}
  \caption*{$H_{\familyID,\graphletID}$}
\end{subfigure}
\hspace*{\fill}
   \renewcommand{\graphletID}{9}%
\begin{subfigure}{\figGraphletWidth}
  \includegraphics[width=\linewidth]{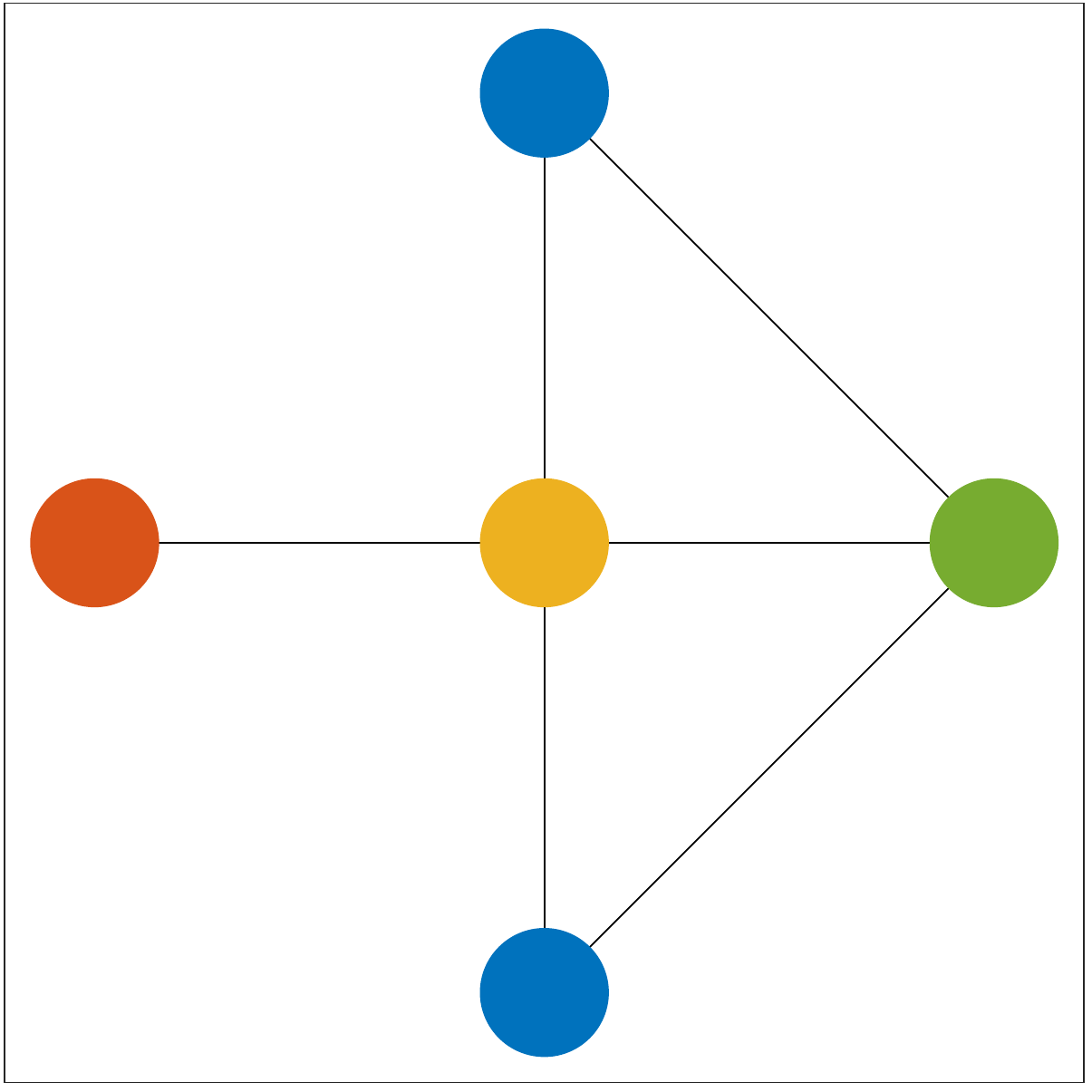}
  \caption*{$H_{\familyID,\graphletID}$}
\end{subfigure}
\hspace*{\fill}
   \renewcommand{\graphletID}{10}%
\begin{subfigure}{\figGraphletWidth}
  \includegraphics[width=\linewidth]{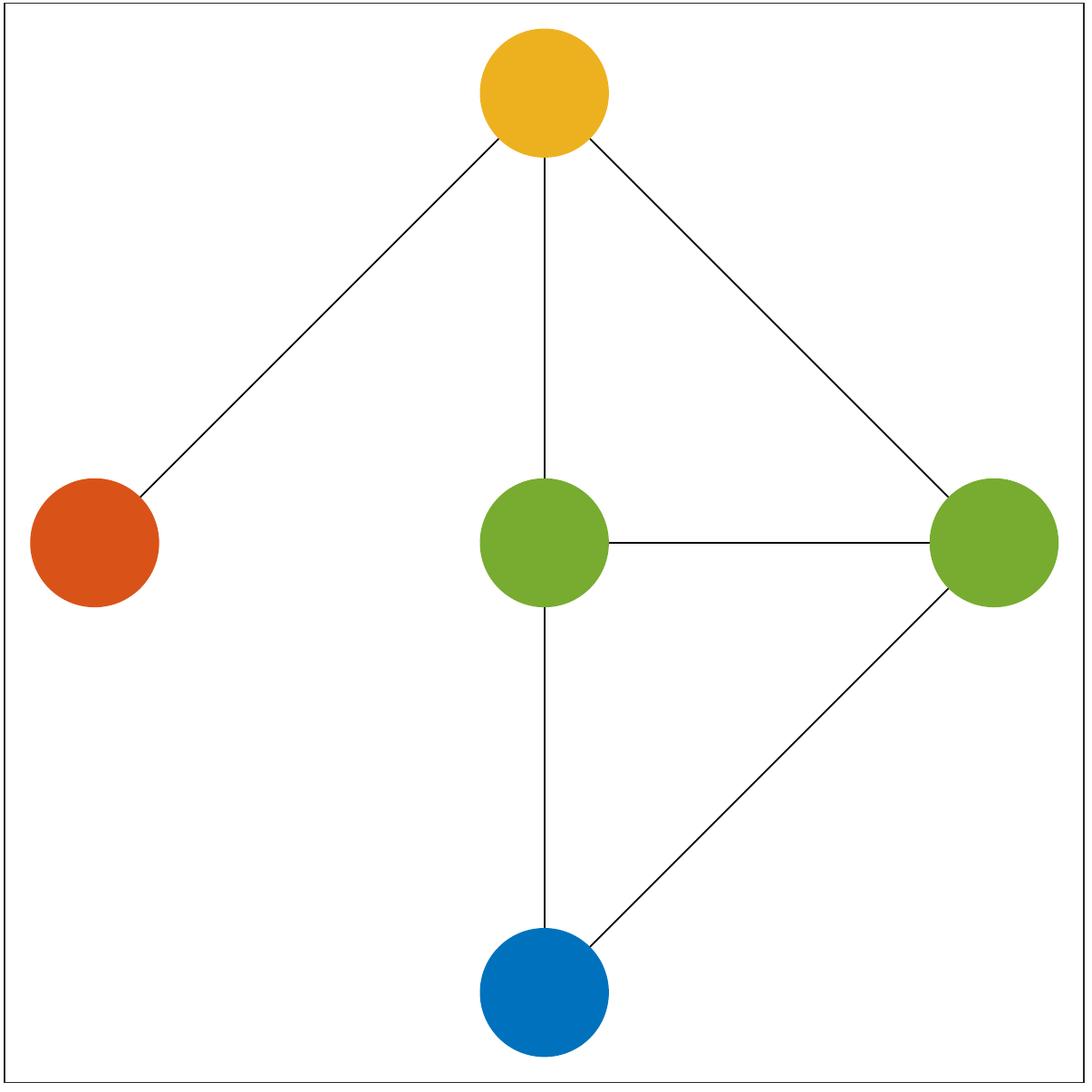}
  \caption*{$H_{\familyID,\graphletID}$}
\end{subfigure}
\hspace*{\fill}
   \renewcommand{\graphletID}{11}%
\begin{subfigure}{\figGraphletWidth}
  \includegraphics[width=\linewidth]{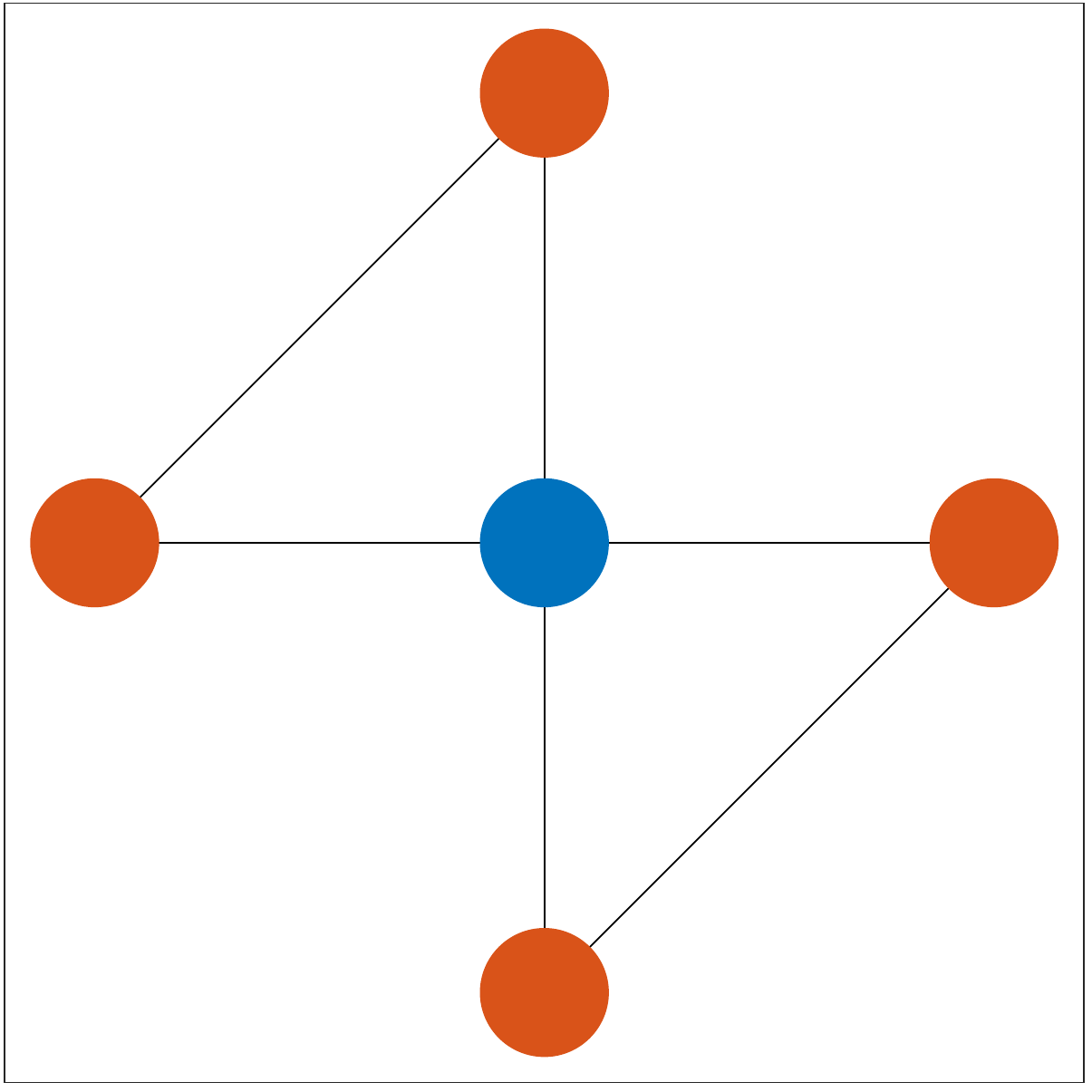}
  \caption*{$H_{\familyID,\graphletID}$}
\end{subfigure}
\hspace*{\fill}
   \renewcommand{\graphletID}{12}%
\begin{subfigure}{\figGraphletWidth}
  \includegraphics[width=\linewidth]{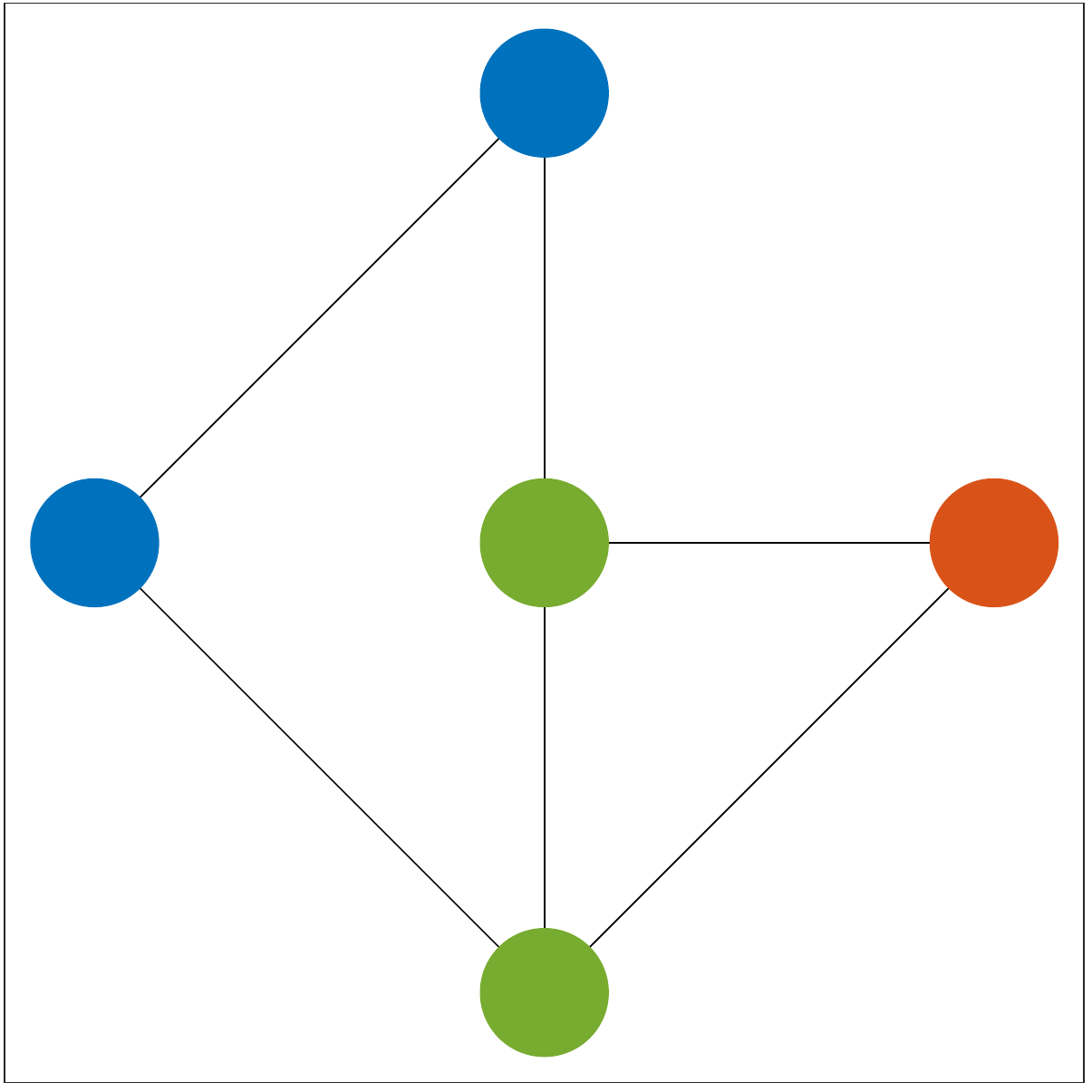}
  \caption*{$H_{\familyID,\graphletID}$}
\end{subfigure}
\hspace*{\fill}
   \renewcommand{\graphletID}{13}%
\begin{subfigure}{\figGraphletWidth}
  \includegraphics[width=\linewidth]{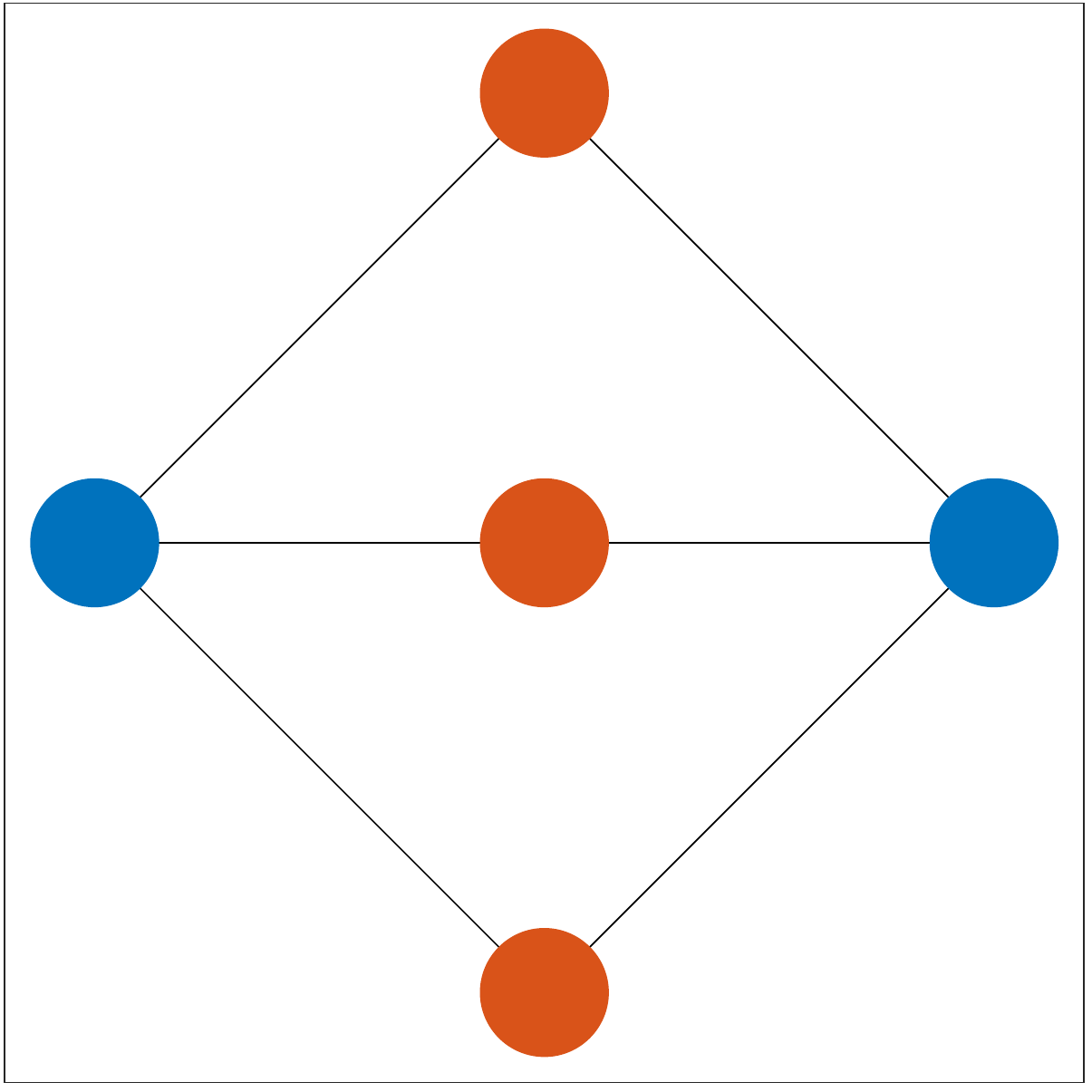}
  \caption*{$H_{\familyID,\graphletID}$}
\end{subfigure}
\hspace*{\fill}
   \renewcommand{\graphletID}{14}%
\begin{subfigure}{\figGraphletWidth}
  \includegraphics[width=\linewidth]{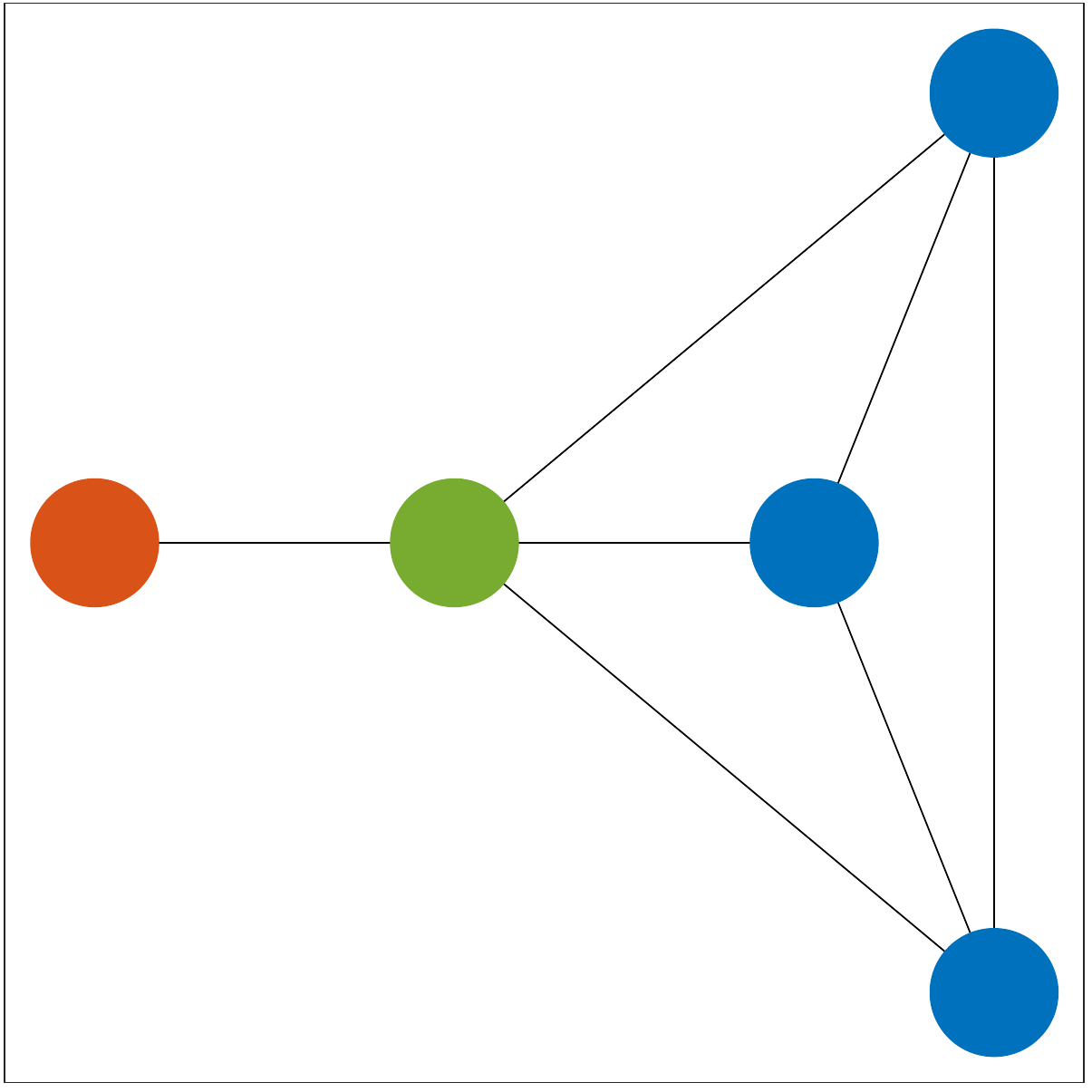}
  \caption*{$H_{\familyID,\graphletID}$}
\end{subfigure}
\hspace*{\fill}
   \\[0.5em]
  \hspace*{\fill}
  \renewcommand{\graphletID}{15}%
\begin{subfigure}{\figGraphletWidth}
  \includegraphics[width=\linewidth]{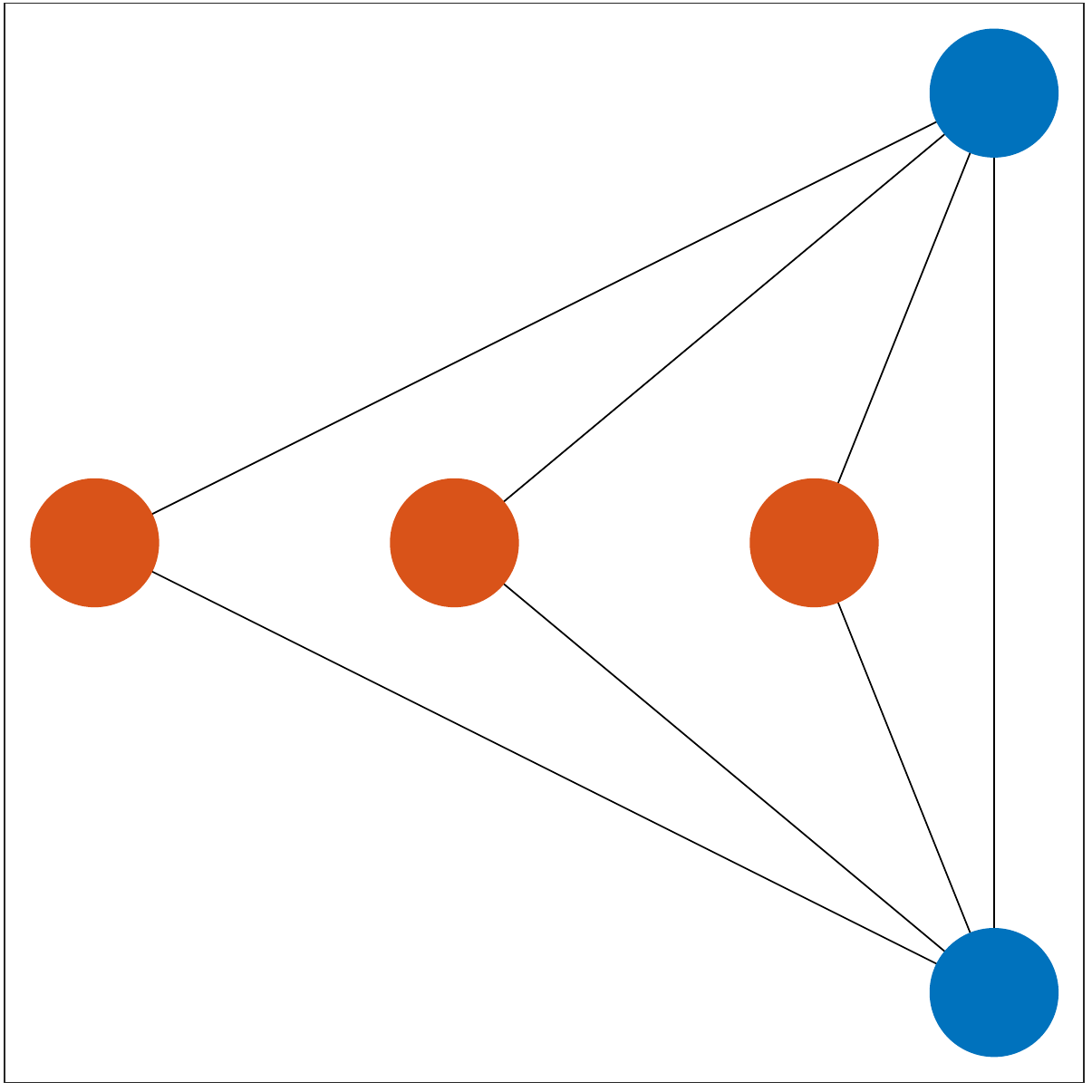}
  \caption*{$H_{\familyID,\graphletID}$}
\end{subfigure}
\hspace*{\fill}
   \renewcommand{\graphletID}{16}%
\begin{subfigure}{\figGraphletWidth}
  \includegraphics[width=\linewidth]{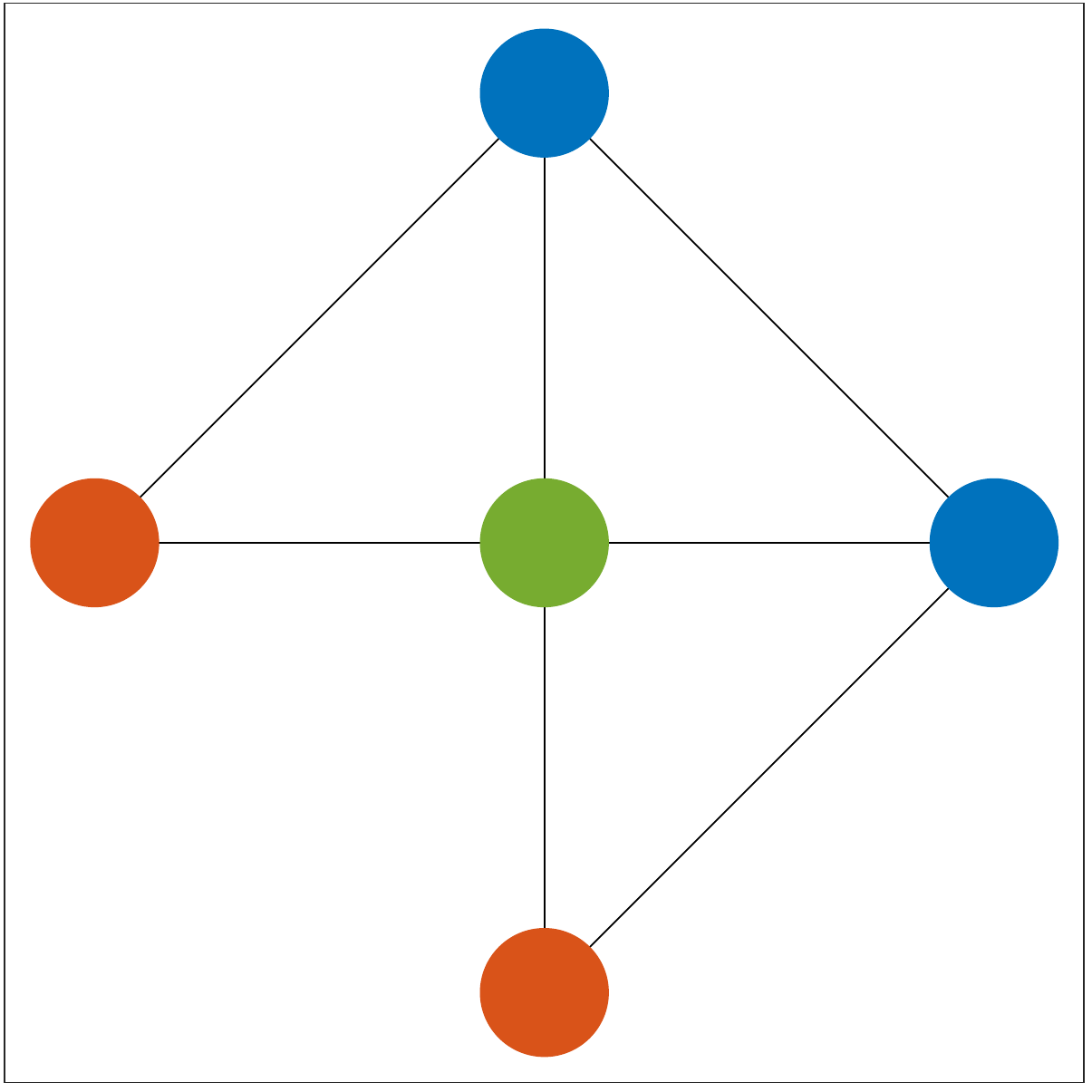}
  \caption*{$H_{\familyID,\graphletID}$}
\end{subfigure}
\hspace*{\fill}
   \renewcommand{\graphletID}{17}%
\begin{subfigure}{\figGraphletWidth}
  \includegraphics[width=\linewidth]{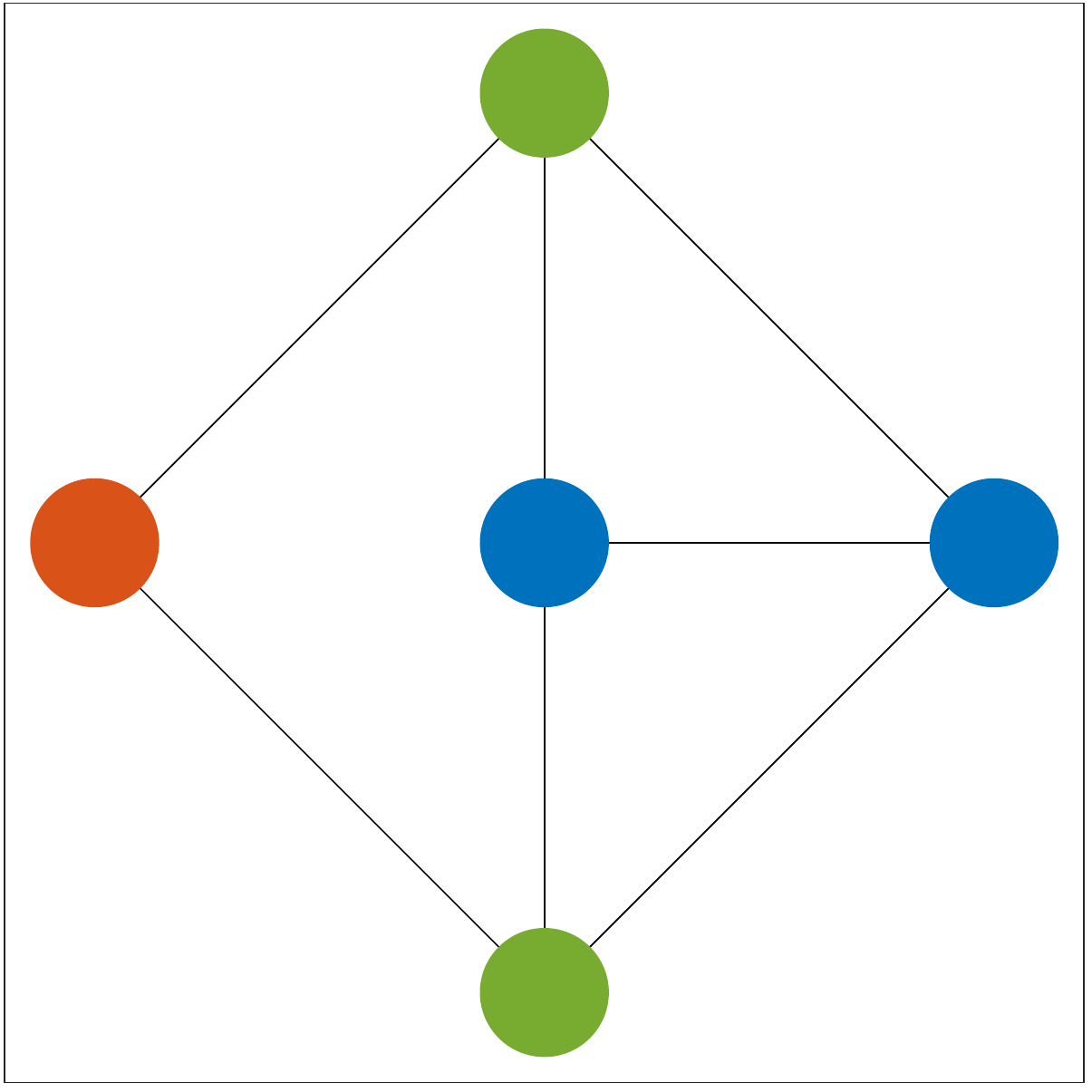}
  \caption*{$H_{\familyID,\graphletID}$}
\end{subfigure}
\hspace*{\fill}
   \renewcommand{\graphletID}{18}%
\begin{subfigure}{\figGraphletWidth}
  \includegraphics[width=\linewidth]{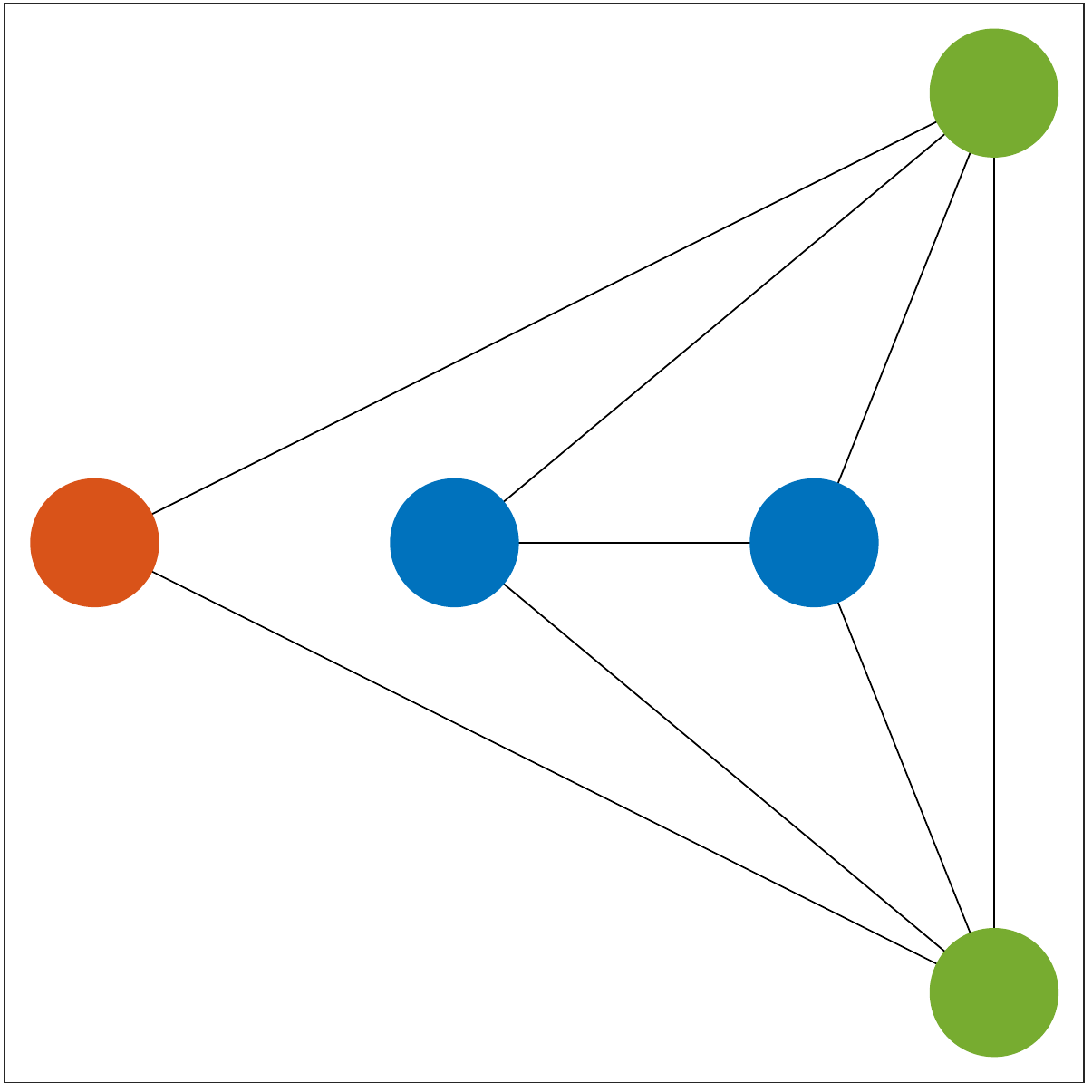}
  \caption*{$H_{\familyID,\graphletID}$}
\end{subfigure}
\hspace*{\fill}
   \renewcommand{\graphletID}{19}%
\begin{subfigure}{\figGraphletWidth}
  \includegraphics[width=\linewidth]{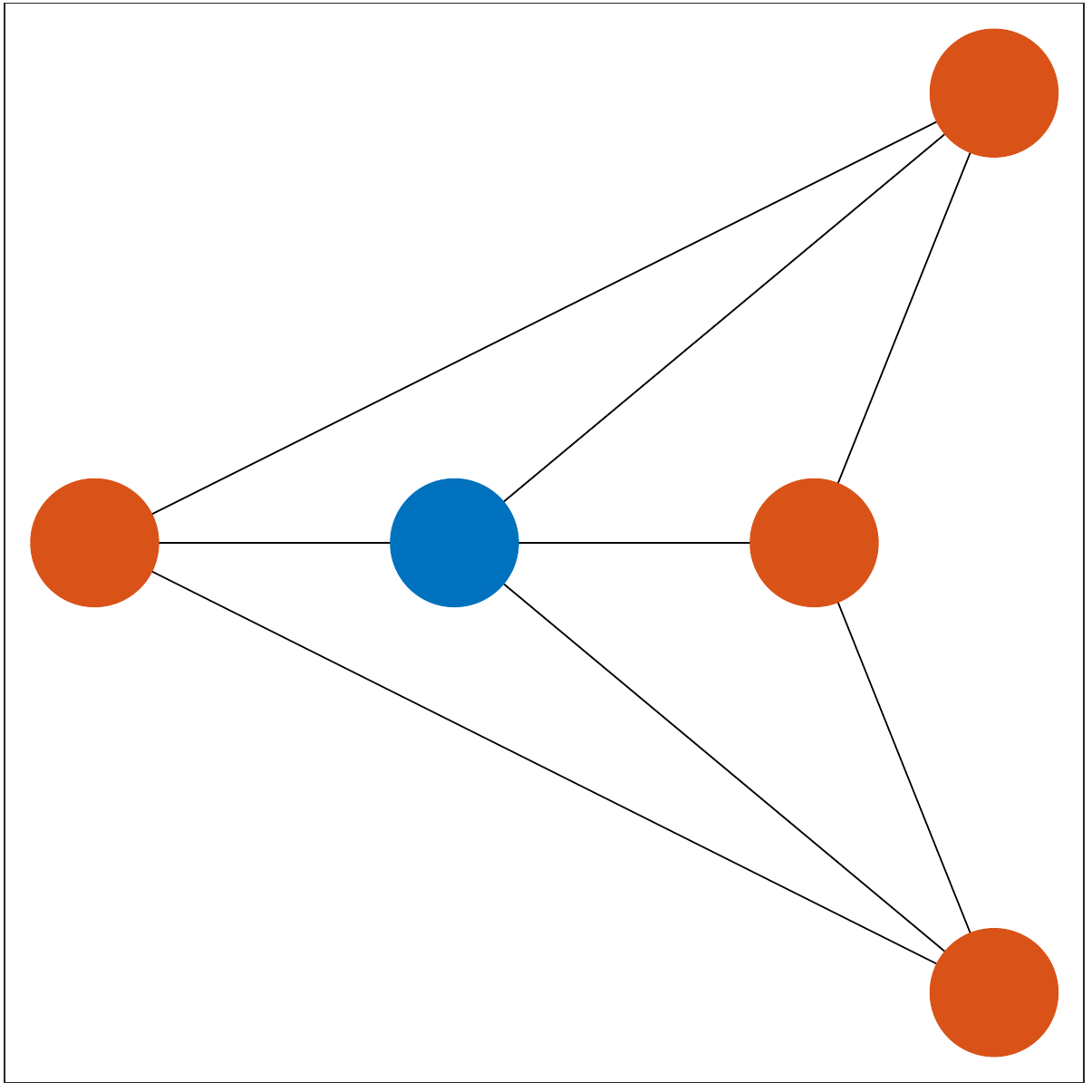}
  \caption*{$H_{\familyID,\graphletID}$}
\end{subfigure}
\hspace*{\fill}
   \renewcommand{\graphletID}{20}%
\begin{subfigure}{\figGraphletWidth}
  \includegraphics[width=\linewidth]{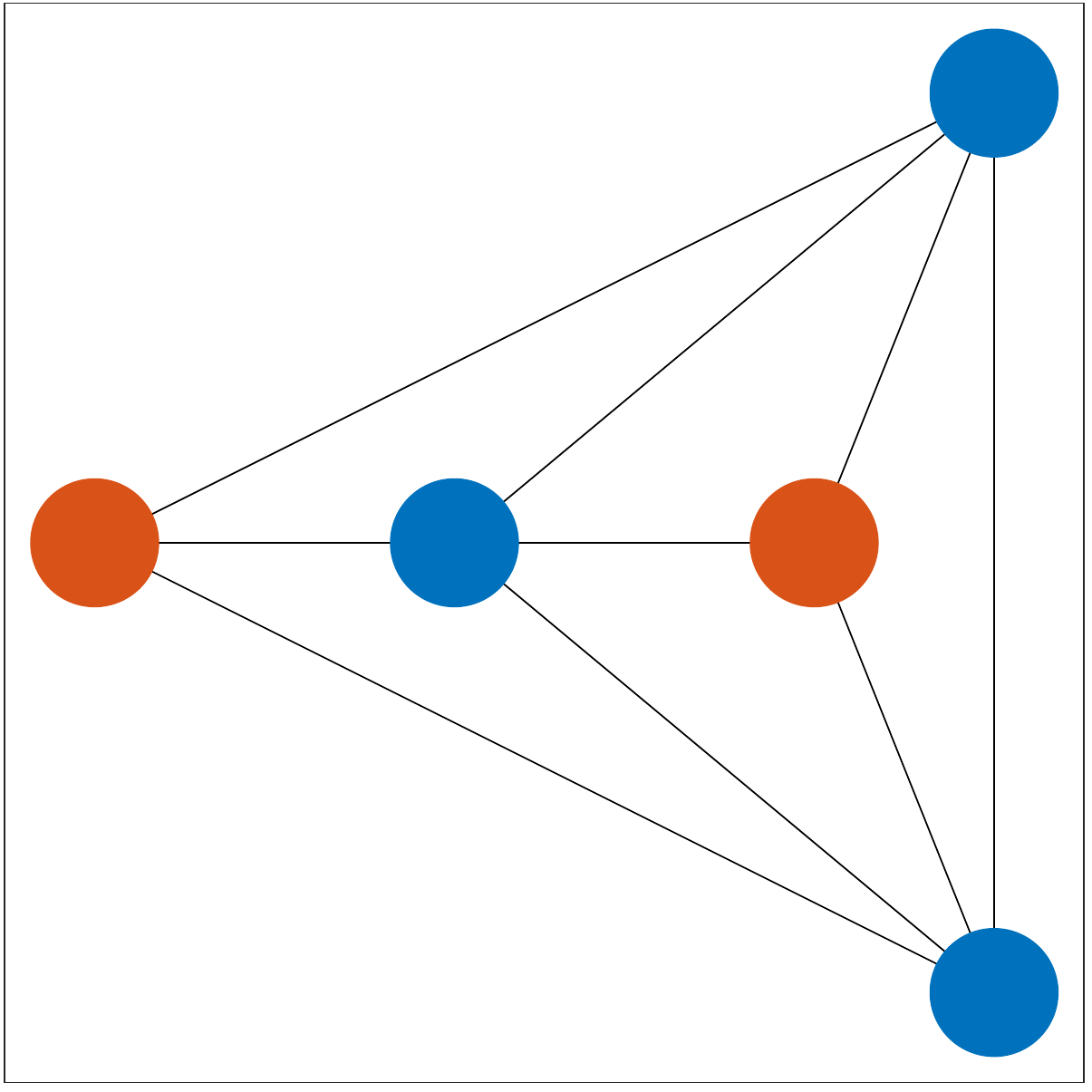}
  \caption*{$H_{\familyID,\graphletID}$}
\end{subfigure}
\hspace*{\fill}
   \renewcommand{\graphletID}{21}%
\begin{subfigure}{\figGraphletWidth}
  \includegraphics[width=\linewidth]{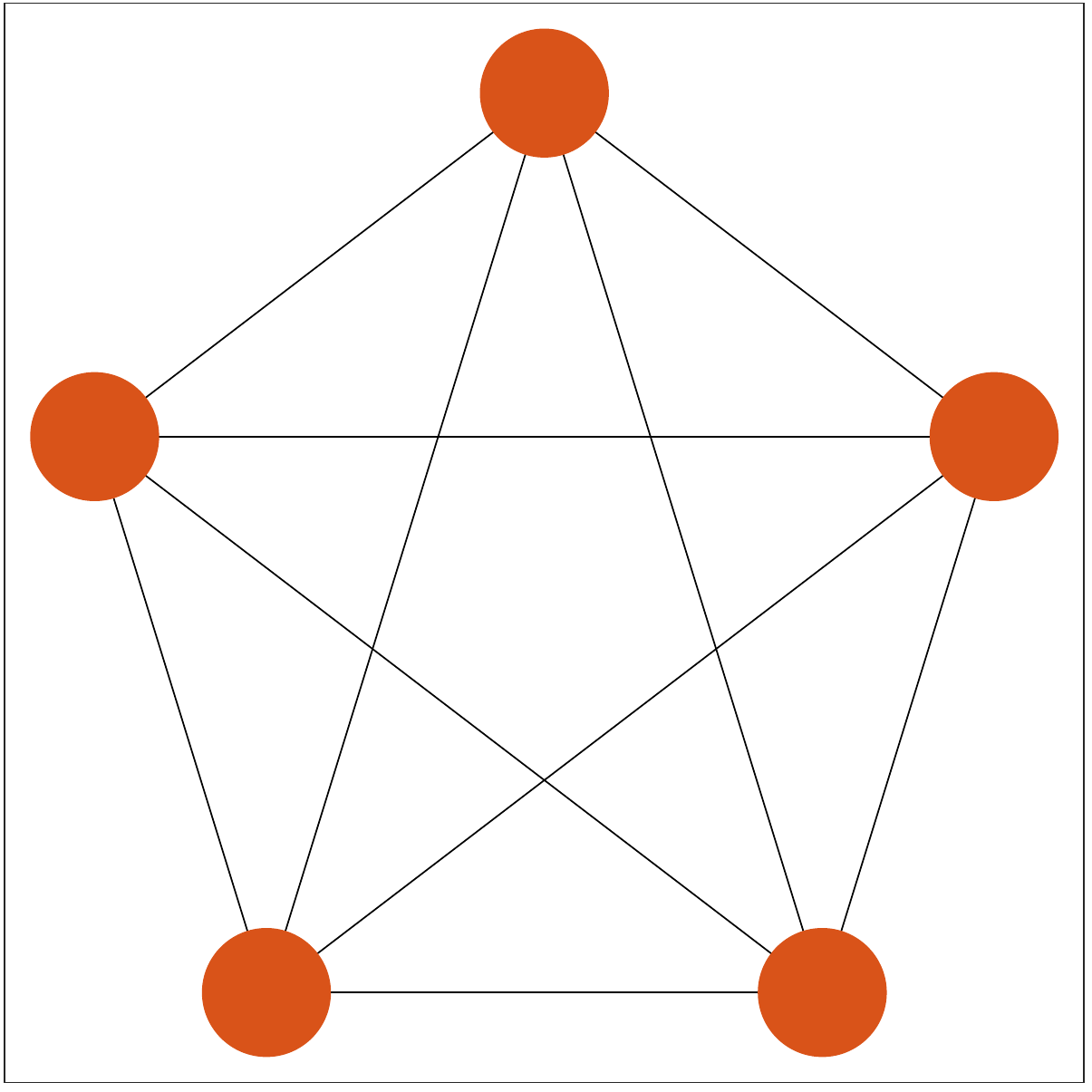}
  \caption*{$H_{\familyID,\graphletID}$}
\end{subfigure}
\hspace*{\fill}
  \\[0.5em]
  \caption{The penta-node graphlet family $\hat{\mathcal{H}}_{5}$ and
    family $\mathcal{H}_{5}$. The graphlets $H_{5,p}$,
    $1 \leq p \leq 21$, in $\hat{\mathcal{H}}_{5}$ are mutually
    non-isomorphic.  By the sub-channel
    decomposition, $\hat{\mathcal{H}}_{5}$ gives rise to 58
    orbit-specific graphlets in $\mathcal{H}_{5}$. Orbits are color
    coded red, blue, green and yellow.  For instance, the green orbit
    of graphlet $H_{5,7,3}$ has two nodes.  The graphlets are placed
    from left to right, top to bottom, by the ordering scheme {\sc
      seira} in \Cref{subsec:graphlet-dictionary}.}
  \label{fig:graphlet-portraits-5node}
\end{figure*}

A graph element, a.k.a. \emph{graphlet}, is a connected template graph
with a small number of nodes with or without a designated incidence orbit,
see \Cref{def:local-counts-at-vertices,def:orbit-specific-counts}.
All $s$-node graphlets with designated incidence orbits form a
natural family ${\cal H}_s$, $s\geq 1$. \Cref{fig:16-graphlet-portraits}
displays the graphlets in the
first four families ${\mathcal{H}}_{s}$, $s=1,2,3,4$;
\Cref{fig:graphlet-portraits-5node}, the graphlets in family
$\mathcal{H}_{5}$.
All $s$-node graphlets without orbit partition form the family
$\hat{\mathcal{H}}_s$, in which all graphlets are mutually
non-isomorphic. It is beneficial to utilize smaller graphlets as much
as possible for graph encoding.  We therefore consider graph encoding
with graphlets up to a certain number $t$ of nodes.
The length of the frequency-vector code is the sum of the chosen
family sizes.
For any $s>2$, the size of $\mathcal{H}_{s}$ is
$| {\cal H}_{s} | = \sum_{H_p\in \hat{\mathcal{H}_s}} a_p > |
\hat{\cal H}_s | $, where $a_p$ is the number of orbits in pattern
template $H_p$.  \Cref{tab:number-graphlets} lists the family sizes up to $8$
nodes. In practice, a small number of graphlet families gives a
desirable code length.

Any graphlet collection is a partially ordered set (poset) with the
binary relation defined by subgraph inclusion: graphlet $H_i$ precedes
graphlet $H_j$, denoted by $H_i \prec H_j$, if $H_i$ is a proper
subgraph of $H_j$. The union of the first $t$ families
$\mathcal{L}_{t} = \bigcup_{s\leq t} \mathcal{H}_s$ is a lattice.
\Cref{fig:hasse-2to5} shows the Hasse diagram of $\mathcal{L}_5$, in
which we include the null graph $\emptyset$.  The lattice with some of
the families removed is a sub-lattice.  Alternatively, any family
$\mathcal{H}_s$, $s\leq t$, together with $\emptyset$, is a
sub-lattice.  Lattice $\mathcal{L}_5$ is a sub-lattice of a larger
lattice with more graphlets included. Let $G$ be a source graph. By
graph encoding with the graphlets in the first $t$ families, the
frequency vector at any vertex $v\in V(G)$ is defined on the lattice
$\mathcal{L}_{t}$. The singleton count at any vertex is always
$1$, which sums to the total count of nodes in a graph or a
subgraph. At each vertex, the lattice $\mathcal{L}_{t}$ with the
singleton removed defines the neighborhood architecture. Lattice
$\mathcal{L}_2$ with the singleton removed is the conventional
neighborhood. Extending the degree $d(v)$, the frequency vector at $v$
over $\mathcal{L}_t$, $t>2$, quantitatively encodes the multi-order
topology structures at the vertex.

We elaborate on a few important details about lattice $\mathcal{L}_t$
and its counterpart $\hat{\mathcal{L}}_t$.  We specify a graphlet $H$
in $\mathcal{L}_t$ with an index triplet $(s,p,\sigma)$, $s$ is the
number of nodes in $H$, $p$ identifies with a unique topological
pattern in $\mathcal{H}_s$, and $\sigma$ identifies with a unique
orbit of $H$.
That is, the graphlet $H_{s,p}$ in $\hat{ \mathcal{H}}_s$ is expanded,
by orbit partition, to the subset
$\mathcal{H}_{s,p} = \{ H_{s,p,\sigma} \} $ in $\mathcal{H}_s$.
\Cref{fig:sub-lattices-quad-node} shows the particular relationship
between $\mathcal{H}_{4}$ and $\hat{\mathcal{H}}_{4}$.
In $\mathcal{L}_{t}$ or $\hat{\mathcal{L}}_{t}$, the length of the
path from any graphlet $H$ to the null element is equal to $m(H) + 1$,
where $m(H)$ is the number of edges in $H$.  The lattice height is
$m(K_t) + 1$, where $K_t$ is the $t$-node clique.  In terms of
neighborhood architectures, lattice $\mathcal{L}_t$ provides more room
for encoding and differentiating orientational information.

\begin{figure}[!t]
  \centering
  \begin{tikzpicture}
    \node[anchor=south west,inner sep=0] (image) at (0,0) {
      \includegraphics[width=.8\linewidth]{%
        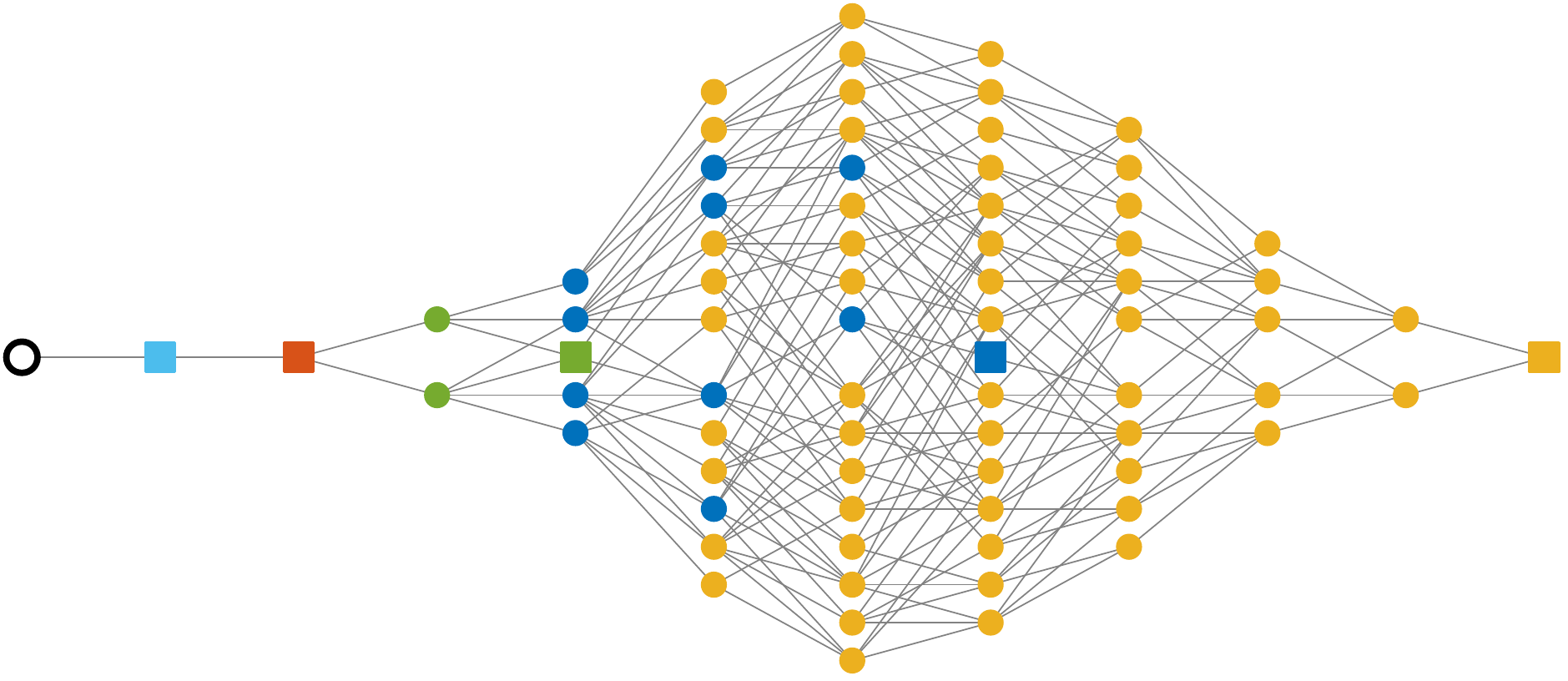}
    };
    \draw[line width=0.45mm, ->] (1,0.3) -- (2,0.3);
  \end{tikzpicture}
  \caption{The Hasse diagram of lattice
    $\mathcal{L}_{5} = \{\mathcal{H}_{s},\, 1\leq s \leq 5\}$ formed
    by the subgraph inclusion relationship among graphlets in the first five
    families. Each node element in the diagram represents a graphlet,
    except the null graph $\emptyset$ depicted by an un-filled circle
    to the leftmost.  Each edge, directed from left to right,
    represents the covering relationship. Cliques $K_{s}$,
    $1\leq s \leq 5$, are depicted with square markers in the central
    row; non-clique graphlets, with filled circles. The graphlets are
    color-coded by the families: ${{\cal H}}_{1}$ in cyan,
    ${{\cal H}}_{2}$ in red, ${{\cal H}}_{3}$ in green,
    ${{\cal H}}_{4}$ in blue, and ${{\cal H}}_{5}$ in yellow.  The
    number of graphlets in each family is in
    \Cref{tab:number-graphlets}. Graphlets with the same path length
    to the null element are placed in the same vertical layer (at the
    same location along the horizontal or $x$ axis); they
    also have the same number of edges. Sub-lattices can be extracted
    from the lattice, see \Cref{fig:sub-lattices-quad-node} for two
    instances. The lattice $\mathcal{L}_{5}$ itself is a sub-lattice
    of a larger one with more graphlet families included. }
  \label{fig:hasse-2to5}
\end{figure}

We introduce a total ordering scheme, {\sc seira}, with simple induction rules that
are self-contained and extendable.  The scheme preserves the inclusion
relationship between any two directly comparable graphlets and places
a sequential order between any two graphlets non-comparable by
inclusion.  The ordering between any two graphlets $H_{s,p.\sigma}$
and $H_{s', p', \sigma'}$ is determined by the lexicographic ordering
detailed below.
\begin{list}{}{\setlength{\leftmargin}{1.7em}
    \setlength{\itemsep}{-.33em}
    \setlength{\topsep}{0em}
  }
\item[(a)] The integers $s$ and $s'$ are naturally ordered.
\item[(b)] When $s=s'$, we assign a unique integer index to each
topologically unique pattern in the same family a unique integer
index. That is, $p=p'$ if and only if
$H_{s,p,\sigma} \cong H_{s,p', \sigma'}$. We let $p < p'$ if
$m(H_{s,p, \sigma}) < m(H_{s,p', \sigma'})$, i.e., $H_{s,p,\sigma}$
has a shorter path to the null graph on the Hasse diagram.  For
example, $K_{1,2}$ is placed ahead of $K_3$. In the case of a tie,
$m(H_{s,p, \sigma}) = m(H_{s,p', \sigma'})$, we break the tie by the
first discrepancy in the frequency sequences (in non-decreasing
ordering) drawn from $ \{ f( H_k | H_{s, p, \sigma})\}$ and
$\{ f( H_k | H_{s, p', \sigma'}) \} $, where $H_k$ is already placed
ahead of $H_{s,p,\sigma}$ and $H_{s,p', \sigma'}$ by {\sc seira} itself.
\item[(c)] When $(s,p) = (s',p')$, we assign each orbit a unique
  integer.  That is, $\sigma = \sigma'$ if the two graphlets are the
  one and the same. We let $\sigma < \sigma'$ if
  $f( H_k| H_{s,p,\sigma}) (v)$ at $v\in \sigma$ is lower than
  $f( H_k| H_{s,p,\sigma'}) (v')$ at $v'\in \sigma'$ where
  $H_k$ is as described in (b). 
\end{list}
By the use of frequency codes with precedent graphlets, {\sc seira} makes
quantitative comparisons between two graphlets that are
non-comparable by subgraph inclusion.  The ordering scheme is used in graphlet
placements in
\Cref{fig:16-graphlet-portraits,fig:graphlet-portraits-5node,fig:sub-lattices-quad-node}
and in matrices composed of graphlet frequencies in
\Cref{sec:feature-conversion,sec:g-surf}.  In general, {\sc seira} can
be applied to any collection of graphlets.

\begin{figure*}[!t]
  \newcommand{\mc}[1]{\multicolumn{#1}{c}{\phantom{x}}}
  \newcommand{\mcv}[2]{\multicolumn{#1}{c}{#2}}
  \centering
  \hspace*{\fill}
  \begin{subfigure}{0.21\linewidth}
    \centering
    \begin{tikzpicture}
      \node[anchor=south west,inner sep=0] (image) at (0,0) {
        \includegraphics[width=\linewidth]{%
          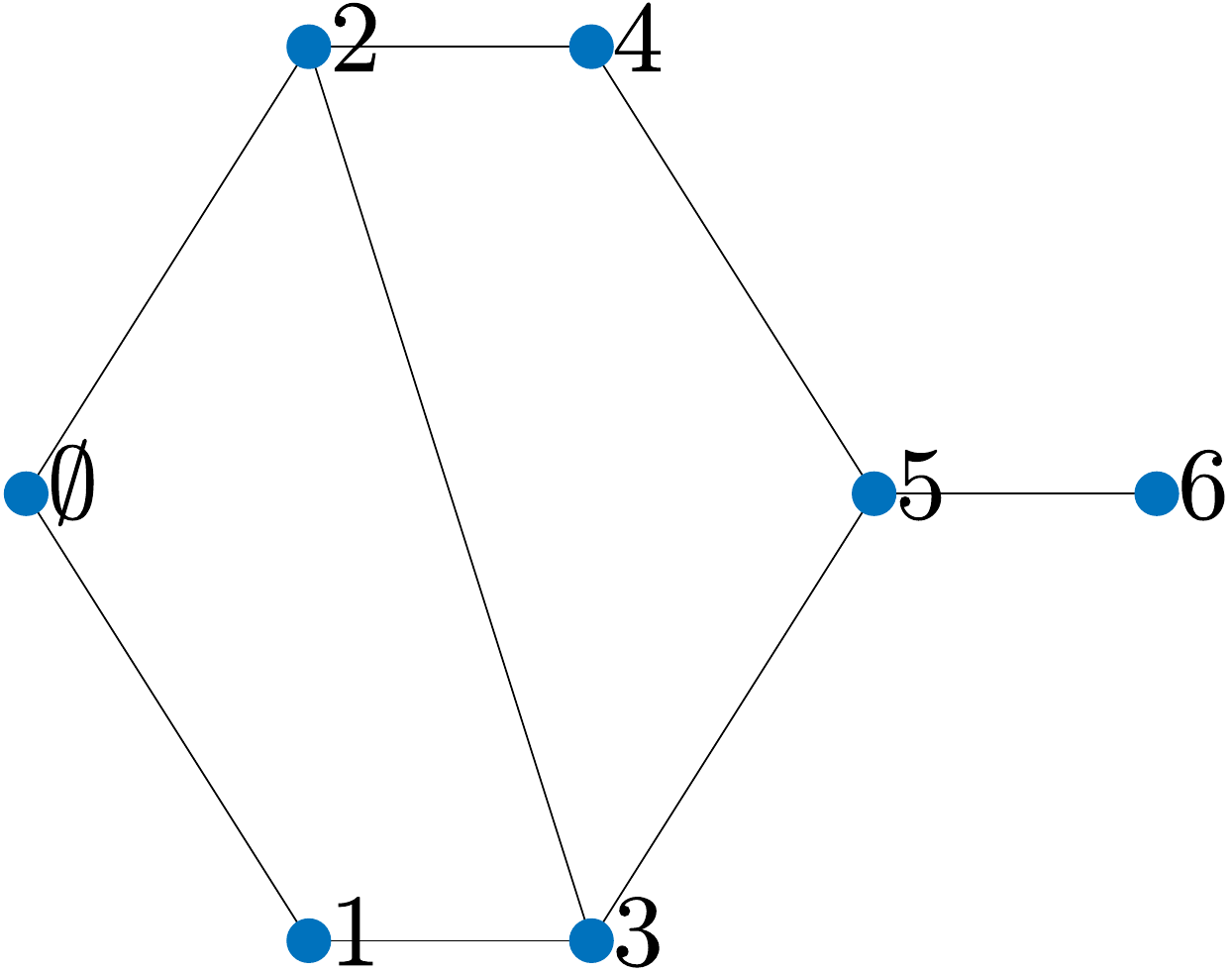}
      };
      \draw[line width=0.45mm, ->] (0.1,-.5) -- (0.9,-.5);
    \end{tikzpicture}
  \end{subfigure}
  \begin{subfigure}{0.23\linewidth}
    \centering
    \resizebox{.7\textwidth}{!}{
      \begin{tabular}{|rrrrrr|}
        \hline
        2 & 2 & 3 & 1 & 2 & 1 \\ \hline
        \multicolumn{6}{c}{\phantom{x}} \\ \hline
        1 &   & 1 &   & 2 & 4 \\
          & 1 & 2 & 4 & 6 & 12 \\
          &   & 1 &   & 4 & 12 \\
          &   &   & 1 & 1 & 3 \\
          &   &   &   & 1 & 6 \\
          &   &   &   &   & 1 \\
        \hline
      \end{tabular}
    }

  \end{subfigure}
  \hspace*{1cm}
  \begin{subfigure}{0.23\linewidth}
    \centering
    \begin{tikzpicture}
      \node[anchor=south west,inner sep=0] (image) at (0,0) {
        \includegraphics[width=\linewidth]{%
          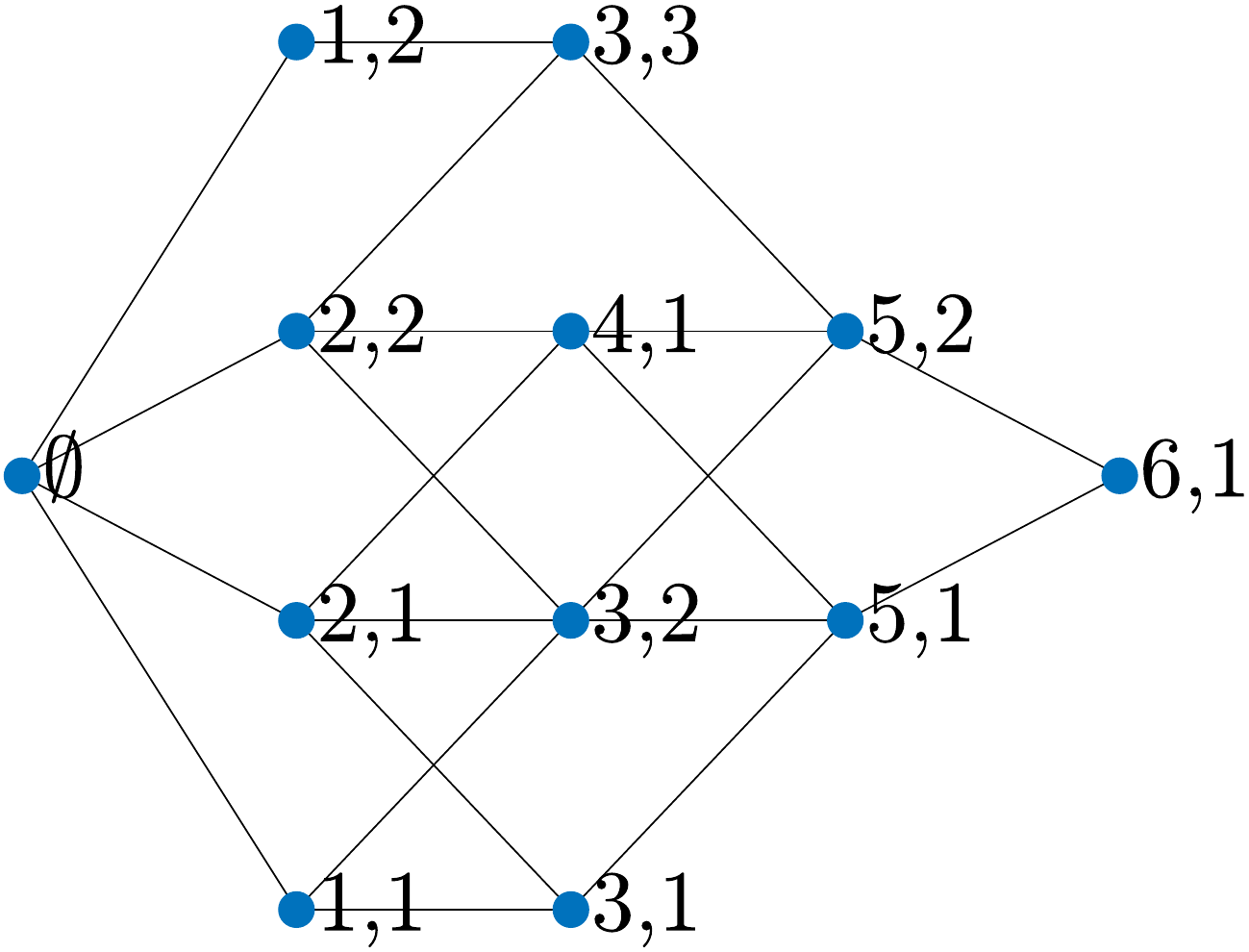}
      };
      \draw[line width=0.45mm, ->] (0.1,-.5) -- (0.9,-.5);
    \end{tikzpicture}
  \end{subfigure}
  \begin{subfigure}{0.23\linewidth}
    \resizebox{\linewidth}{!}{
      \begin{tabular}{|p{0.2cm}p{0.2cm}|p{0.2cm}p{0.2cm}|p{0.2cm}p{0.2cm}p{0.2cm}|p{0.2cm}|p{0.2cm}p{0.2cm}|p{0.2cm}|}
        \hline
        1 &   & \mc{2} & 1 & 1 & \mc{2}                  & 2          & \mcv{1}{1} & 3 \\
          & 1 & \mc{2} &   &   & \mcv{1}{1} & \mc{1}     &            & \mcv{1}{1} & 1 \\ \cline{1-4}
          &   & 1  &   & 2 & 1 & \mc{1}     & \mcv{1}{2} & 4          & \mcv{1}{2} & 6 \\
          &   &    & 1 &   & 1 & \mcv{1}{2} & \mcv{1}{2} & 2          & \mcv{1}{4} & 6 \\ \cline{3-7}
          & \mc{2} &   & 1 &   &            & \mc{1}     & 2          & \mc{1}     & 3 \\
          & \mc{2} &   &   & 1 &            & \mc{1}     & 2          & \mcv{1}{2} & 6 \\
          & \mc{2} &   &   &   & 1          & \mc{1}     &            & \mcv{1}{2} & 3 \\ \cline{5-8}
          & \mc{3}     &   &   &            & 1          & 1          & \mcv{1}{1} & 3 \\ \cline{8-10}
          & \mc{6}                          &            & 1          &            & 3 \\
          & \mc{6}                          &            &            & 1          & 3 \\ \cline{9-11}
          & \mc{7}                                       &            &            & 1 \\ \hline
      \end{tabular}
    }
  \end{subfigure}
  \hspace*{\fill}
  \\[0.5em]
  \hspace*{\fill}
  \begin{subfigure}{0.21\linewidth}
    \caption*{$\hat{\mathcal{H}}_{4}$}
  \end{subfigure}
  \begin{subfigure}{0.23\linewidth}
    \caption*{$\mathbf{\hat{U}}_{4}$}
  \end{subfigure}
  \hspace*{1cm}
  \begin{subfigure}{0.23\linewidth}
    \caption*{${\mathcal{H}}_{4}$}
  \end{subfigure}
  \begin{subfigure}{0.23\linewidth}
    \caption*{$\mathbf{U}_{4}$}
  \end{subfigure}
  \hspace*{\fill}
  \caption{
    The sub-lattices of $\mathcal{L}_{5}$ for $\hat{\cal H}_{4}$ (left)
    and ${{\cal H}}_{4}$ (right) and frequency conversion matrices,
    see ${\cal L}_{5}$ in \Cref{fig:hasse-2to5}.
    Matrices $\mathbf{\hat{{U}}}_{4}$
    and $\mathbf{{U}}_{4}$, by \Cref{def:s-node-U-matrices}, are shown
    next to the corresponding sub-lattices.
    The $p$-th element of the row vector above $\mathbf{\hat{U}}_{4}$
    is the size of the $p$-th diagonal block of $\mathbf{U}_{4}$,
    i.e., the number of orbits in the pattern $H_{4,p}$.}
  \label{fig:sub-lattices-quad-node}
\end{figure*}

\section{Intrinsic connections in frequency vectors}
\label{sec:feature-conversion}

We disclose and describe in this section rich and intrinsically
structural connections among the graphlet frequency vectors.  These
connections are presented coherently and systematically for the first
time in simple and rigorous expressions. 

\subsection{Local transforms on graphlet lattices}
Net frequency counting, by
\Cref{def:global-counts,def:local-counts-at-vertices,def:orbit-specific-counts},
is subject to the constraint that the subgraphs must be induced. It is
shown in a precursor work~\cite{floros2020c} that with $s$-node graphlet
families, $s\leq 4$,

\begin{list}{}
  {
    \setlength{\topsep}{.5em}
    \setlength{\leftmargin}{1.5em}
  \setlength{\topsep}{0em}
  \setlength{\itemsep}{-0.33em}
  \setlength{\leftmargin}{1.75em}
 }
\item[(i)] counting the gross frequencies of an $s$-node family can be
  highly flexible, direct and efficient, especially on sparse
  networks, by utilizing the net or gross frequencies of $s'$-node
  families, $s'< s$, and the sparsity structure of a network. The $s$-node
  frequencies are non-linearly related to the precedent frequencies.
\item[(ii)] the net frequencies can be obtained with ease and
efficiency from gross frequencies within the same
family, and vice versa. Specifically, the conversions are by linear transforms.
\end{list}

The two statements extend to any $s$-node graphlet families, $s>2$.
They underscore the essential properties common to various formulas in the
literature for computing graphlet frequencies.  These properties were not
declared in previous works elsewhere due in part to the lack of
conceptual and computational distinctions between net frequencies and
gross frequencies or other possible intermediate frequencies. In the
present work, we generalize the finding in (ii) to any graphlet family
of $s$-nodes, $s>2$.  This finding is important because the linear
frequency conversion within any graphlet family leaves the non-linear
transforms in graphlet frequencies to that across different
families.

In the rest of the section, we focus on intrinsic relations, local to
every vertex, in the frequencies w.r.t.  graphlets with orbit
distinction, i.e., graphlets in families $\mathcal{H}_s$,
$s\geq 1$. The orbit distinction, however, is not critical; the
frequency relations can be translated to the graphlets in
$\hat{\mathcal{H}}_{s}$ families by the sub-channel (de)composition
property (\ref{eq:sigma-frequency-split}).  We denote the quantities
in $\hat{\mathcal{H}}_{s}$ families with an overhead hat accordingly.
Recall from \Cref{subsec:graphlet-dictionary} and
\Cref{fig:hasse-2to5} the lattices associated with each type of the
graphlet families and the relationship between them.  The essence in
frequency relationships is in the lattice properties.

\subsection{Intra-family relations}
\label{sec:intra-family-relations}

\begin{figure}
  \centering
  \begin{subfigure}{0.35\linewidth}
    \includegraphics[width=\linewidth]{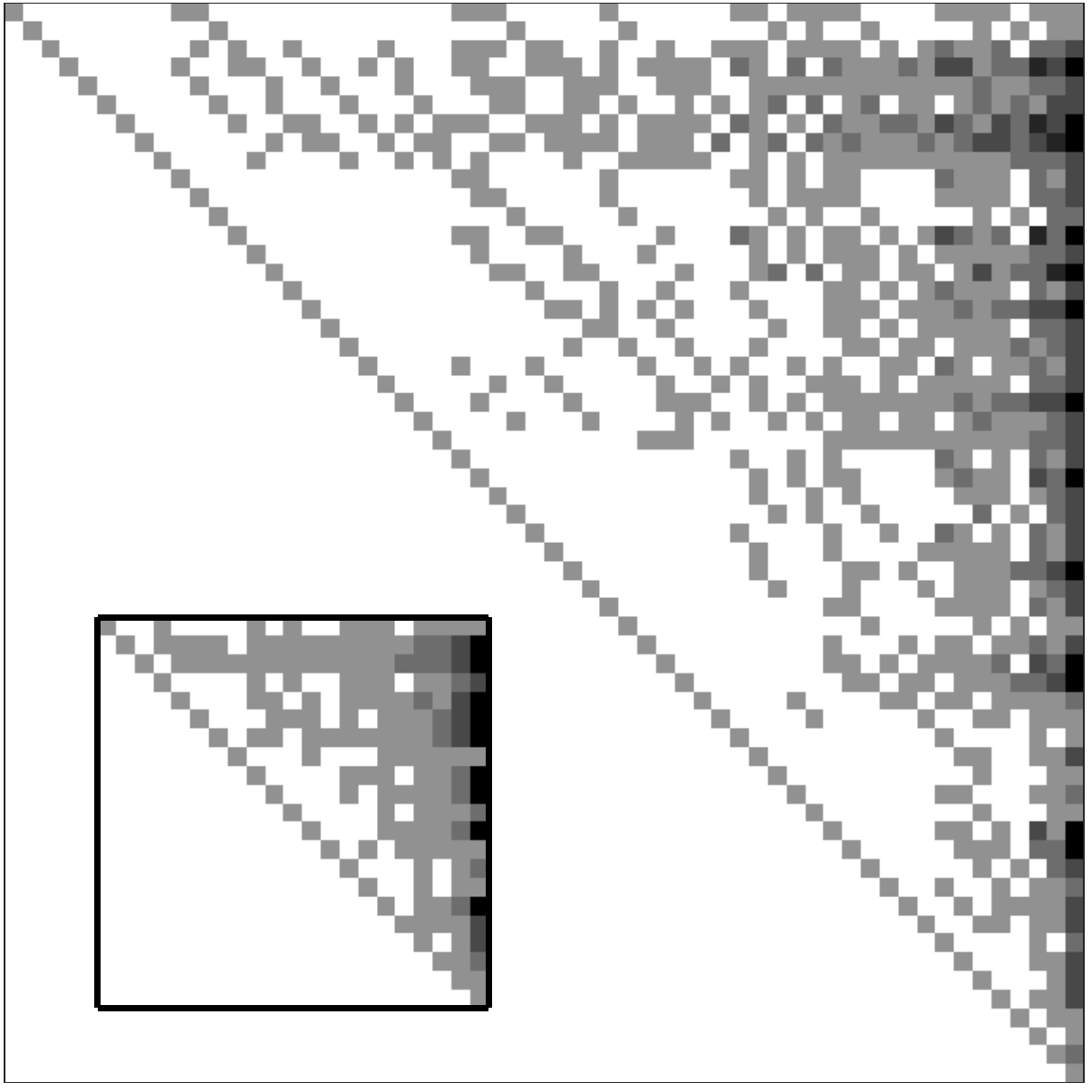}
  \end{subfigure}
  \caption{ Frequency conversion matrices ${\mathbf{U}}_{5}$
    and $\hat{\mathbf{U}}_{5}$ by \Cref{def:s-node-U-matrices}.
    Elements with larger integers are in darker pixels;
    with zero values, in white pixels. 
    Both matrices are upper triangular. Matrix
    $\hat{\mathbf{U}}_{5}$ ($21\times 21$ with $164$ nonzeros) is
    placed in the lower triangular space of $\mathbf{{U}}_{5}$
    ($58\times 58$ with $744$ nonzeros).
    The inverse of each matrix has the same sparsity pattern, by
    Theorem~\ref{thm:frequency-conversion}. }
  \label{fig:conversion-five}
\end{figure}

The net frequencies on a source graph $G$ are bounded from above by
the corresponding gross frequencies. We prove that with the graphlet
encoding system, the net and gross frequency vectors on a source graph
w.r.t. to any $s$-node family $\mathcal{H}_s$ are related by linear
transforms.

\begin{definition}{\rm (Intra-family gross-frequency matrices)} 
  \label{def:s-node-U-matrices}
  For any $s>1$, let $ { \cal H}_{s} $ be the family of $s$-node
  graphlets, each with a distinctive incidence orbit.  Define matrix
  $\mathbf{U}_s$ by the pairwise gross frequencies among the graphlets
  in the respective families as follows,
 \begin{equation}
   \label{eq:U-matrices}
   \begin{array}{l}
     \mathbf{U}_{s} ( i, j  )
     \, \triangleq \,
     g( H_i | H_j  )(v) ,
     \quad
     H_{i} , H_{j} \in  {\cal H}_s, 
   \end{array} 
 \end{equation} 
 where $v$ belongs to the designated orbit of $H_{j}$.
\end{definition}
Matrix $\mathbf{U}_s$ is intrinsically
triangular with unit diagonal values, by the subgraph inclusion in
each family and the reflexive property $g( H | H) = 1$ with any
graphlet $H$; it is therefore unimodular.  The matrix is made
upper triangular by any ascending (upward) ordering.  Its inverse
is also upper triangular and unimodular.
Matrices $\mathbf{U}_4$ and $\hat{\mathbf{U}}_4$ are shown in
\Cref{fig:sub-lattices-quad-node}.  Matrices $\mathbf{U}_5$ and
$\hat{\mathbf{U}}_5$ are depicted as gray images in
\Cref{fig:conversion-five}.

\begin{theorem}{\rm (Linear conversions of frequency vectors)} 
  \label{thm:frequency-conversion} 
  For $s>1$, let ${\cal H}_s$ and $\mathbf{U}_s$ be specified as in
  \Cref{def:s-node-U-matrices}. Then,
  \begin{list}{}
    {\setlength{\leftmargin}{0.1em} }
  \item[(a)] 
    for any source graph $G(V, E)$, at any vertex $v \in V$,
    \begin{equation}
      \label{eq:frequency-transform}
        \mathbf{U}_s \, f( {\cal H}_s | G)(v) = g( {\cal H}_s | G)(v);
    \end{equation}
\item[(b)] matrix $\mathbf{U}_s\Lambda_s$ is involutory, i.e.,
   \begin{equation}
    \label{eq:inverse-transform}
      \mathbf{U}_s^{-1} = \Lambda_s \mathbf{U} \Lambda_s,
    \end{equation}
  where $\Lambda_s = \mbox{\rm diag}(\lambda_i )$ with
  $\lambda_i = (-1)^{m(H_i)}$, $H_i \in \mathcal{H}_s$.
  \end{list}
\end{theorem}
By the theorem, net frequencies and gross frequencies w.r.t. the same
family $\mathcal{H}_s$ are exchangeable via substitution, with any $s \geq 1$.
Part (a) of the theorem is straightforward to verify by the
definitions of gross and net frequencies.  The involutory property of
$\mathbf{U}_s\Lambda_s$ is theoretically
interesting in its own right. It is practically useful because the
reverse conversions are in the same ready and easy way as the forward
conversions in \Cref{eq:frequency-transform}.
\begin{proof} A simple proof of (b) is based on the following fact.
  Consider $f(H_i|G)(v)$, the net frequency at $v$ with respect to
  template $H_i \in {\cal H}_s$. We can obtain the net frequency by
  removing the redundancy from the gross frequency in $g( H_i|G)(v)$.
  For any $H_j \succ H_i$ covered by $H_{j'}$ in the Hasse diagram,
  the redundancy in $g(H_i|G)(v)$ due to $H_j$-isomorphic and
  $H_{j'}$-isomorphic subgraphs incident at $v$ is
  $g( H_i| H_j ) g(H_j|G)(v) - g(H_{i} | H_{j'}) g( H_{j'} | G)(v)$. By
  nested removal of the redundancies due to all super-graphs of
  $H_i$, we arrived at, 
  \[
    \begin{array}{rl} 
      \displaystyle
      f(H_i| G)(v) & = g(H_i|G)(v)
                       \displaystyle 
                     - \sum_{U(i,j)>0, j\neq i} 
                     \lambda_i \lambda_j \, g( H_i| H_j ) \, g(H_j|G)(v)
      \\
      & = \lambda_i\, U_s(i,:)\, \Lambda \, g( {\cal H}_s | G)(v),
      \qquad    \forall v\in V , 
     \end{array} 
   \]
for any $H_i \in {\cal H}_s$.    
 \end{proof}

\begin{corollary}
  \label{corr:co-factor-matrix}
  Let $\mathbf{U}$ be a matrix in \Cref{def:s-node-U-matrices}.
  Denote by $\mathbf{U}_{[ij]}$ the sub-matrix obtained from
  $\mathbf{U}$ with row $i$ and column $j$ removed. Then,
  $\mathbf{U}_{[ij]}$ is non-singular if and only if
  $\mathbf{U}(j,i) \neq 0$.
\end{corollary}
\begin{proof}
  By (b) of Theorem \ref{thm:frequency-conversion} and Cramer's rule,
  $| \, \mathbf{U}(j,i) \,  | = | \, \mathbf{U}^{-1}(j,i)\,  | = |\,  \det(\mathbf{U}_{[ij]}) \, |$.
\end{proof}

\begin{proposition}
  \label{prop:intra-family}
  Let $G(V,E)$ be a source graph. For any $s>1$, $H_j\in \mathcal{H}_s$,
  at any $v\in V(G)$, 
  \begin{equation}
    \label{eq:intra-family-filter}
    f( H_j | G)(v) = 0
    \mbox{ \rm if }
    g (H_j |G)(v) < \mathbf{U}_{s}( i, j ) 
    \mbox{\rm\ for\ some }
    H_i \in \mathcal{H}_s,
    \      H_i \prec H_j . 
  \end{equation}
  The upper bounds in (\ref{eq:intra-family-filter}) on the precedent gross
  frequencies are independent of any source graph.
\end{proposition}

\subsection{Inter-family relations}
\label{sec:inter-family-relations}

\begin{table}[!t]
  \centering
  \caption{Inter-family net-frequency matrix
    $[\mathbf{W}_{2,4}; \mathbf{W}_{3,4}]$, by
    Definition~\ref{eq:UW-matrix-blocks},
    for filtering by algorithm \gsurf{}. }
  \label{tab:inter-family}
    \begin{tabular}{lccccccccccc}
      \toprule
      & $H_{4,1,1}$
      & $H_{4,1,2}$
      & $H_{4,2,1}$
      & $H_{4,2,2}$
      & $H_{4,3,1}$
      & $H_{4,3,2}$
      & $H_{4,3,3}$
      & $H_{4,4,1}$
      & $H_{4,5,1}$
      & $H_{4,5,2}$
      & $H_{4,6,1}$
      \\
      \midrule
      $H_{2,1,1}$ & 1 & 3 & 1 & 2 & 1 & 2 & 3 & 2 & 2 & 3 & 3 \\
      \hdashline
      $H_{3,1,1}$ & 2 & 0 & 1 & 1 & 2 & 1 & 0 & 2 & 2 & 0 & 0 \\
      $H_{3,1,2}$ & 0 & 3 & 0 & 1 & 0 & 0 & 2 & 1 & 0 & 1 & 0 \\
      $H_{3,2,1}$ & 0 & 0 & 0 & 0 & 0 & 1 & 1 & 0 & 1 & 2 & 3 \\
      \bottomrule
    \end{tabular}%
\end{table}

The relations in frequencies on a source graph across different
families are fundamentally non-linear.

\begin{example} For any graph $G(V,E)$, at any $v\in V$,
  \label{example:nonlinear-relation}
\[ 
    \label{eqn:non-linear-upward} 
      g(H_{4,2,2} | G)(v)  \! = g(H_{2,1,1} | G)(v) \, g(H_{3,1,1} | G)(v)
            -g(H_{3,1,1} | G)(v) - 2\, g(H_{3,2,1}|G)(v). 
\] 
\end{example}

We introduce certain useful inference rules by inequalities, regardless how the
frequencies with family $\mathcal{H}_s$ are computed from the
precedent families.

\begin{definition}{\rm (Inter-family pairwise-frequency matrices)}
  \label{def:inter-family-gross-matrices}
  For $s, s' >1$, let ${\cal H}_{s}$ and ${\cal H}_{s'} $ be two
  families of graphlets with distinctive incidence orbits.  Define
  matrix $\mathbf{U}_{s,s'}$ and $\mathbf{W}_{s,s'}$ by the pairwise
  gross and net graphlet frequencies, respectively, across the two
  families,
  \begin{equation}
    \label{eq:UW-matrix-blocks}
      \mathbf{U}_{s,s'} (\, i, j \, )
      \, \triangleq \,
      g\, ( H_{i}  | H_{j} )( v ) ,
      \qquad
      \mathbf{W}_{s,s'} (\, i, j \, )
      \, \triangleq \,
      f\, ( H_{i}  | H_{j} )( v ) ,
      \qquad
      H_{i} \!\in\! \mathcal{H}_s,\
      H_{j} \!\in\! \mathcal{H}_{s'},
  \end{equation}
  where $v$ is in the designated orbit of $H_{j}$.
\end{definition}
Clearly, when $s'=s$, $\mathbf{U}_{s,s'} = \mathbf{U_s}$ by (\ref{eq:U-matrices}). 

\begin{corollary}
  \label{coro::inter-graphlet-frequency-matrices}
  For $t>1$, let 
    $ \widetilde{\mathbf{U}} _t \triangleq \big[\mathbf{U}_{s,s'}\big]_{s,s'=1}^{t} $ 
  and $ \widetilde{\mathbf{W}} _t \triangleq \big[\mathbf{W}_{s,s'}\big]_{s,s'=1}^{t} $. 
  Then, $\widetilde{\mathbf{U}}_t$ and $ \widetilde{\mathbf{W}}_t$ 
  are block upper triangular and related as follows, 
  \begin{equation}
    \label{eq:UW-coversion}
    \widetilde{\mathbf{W}}_{t} =
    \mbox{\rm diag}( \mathbf{U}_1^{-1}, \cdots, \mathbf{U}_t^{-1} )
    \, 
    \widetilde{\mathbf{U}}_{t} .
  \end{equation}
  \end{corollary}
  \Cref{tab:inter-family} shows a submatrix of
  $\widetilde{\mathbf{W}}_{4}$.  Independent of source graphs,
  matrices $\widetilde{\mathbf{U}}_t$ and $\widetilde{\mathbf{W}}_t$,
  for any fixed $t$, can be precomputed once and for all.  The
  relation between them is a direct consequence of
  Theorem~\ref{thm:frequency-conversion}.

\begin{proposition}
  \label{prop:inter-family}
  Let $G(V,E)$ be a source graph. For any $s>1$, $H_j \in \mathcal{H}_s$,
  at any $v\in V(G) $,
  \begin{equation}
    \label{eq:inter-family-filter}
    f( H_j | G) (v) = 0
    \mbox{\rm  \ if \ }
    f( H_i |G)(v) < \mathbf{W}_{r,s}( i, j )
    \mbox{\rm \  for\ some\ }  H_i \in \mathcal{H}_r, \ \  r < s.
  \end{equation}
  The upper bounds in (\ref{eq:inter-family-filter}) on the precedent net
  frequencies are independent of any source graph.
\end{proposition}

Remarks. Myriad formulas in the graphlet literature had been used by
haphazard selections. We are able to characterize, categorize, relate
and interpret them in terms of the relationships among the encoding
elements as well as the relationships among graphlet freqencies on any
source graph.

\section{Algorithm \gsurf{}}
\label{sec:g-surf}

We present a novel algorithm, \gsurf{}, for systematic and efficient
generation of frequency maps on any source graph $G$ with 
$t$ graphlet families $\mathcal{H}_s$, $1<s\!\leq\! t$.
We first describe the baseline algorithm.  We then elaborate on
acceleration methods and show a significant cost reduction in a case
study with a real-world network. Additionally, we comment on
time and space complexities.

\subsection{The baseline algorithm}

Algorithm \gsurf{} takes at input:
\begin{inparaenum}[(i)]
  \item a source graph $G(V,E)$ with adjacency matrix $A$, and
  \item an integer $t>2$ specifying the graphlet families
    ${\cal H}_s$, $s\leq t$.
\end{inparaenum}
The algorithm renders at output the net frequency maps
$\{f({{\cal H}}_{1:t} | G)(v), v\in V(G)\}$.
\gsurf{} has the following basic steps.
Initially, the first net frequency maps are set over $V$,
\[
  f(K_{1}|G)(V) = 1, \quad f(K_{2} |G)(V) = d(V) = A \cdot e. 
  \] 
  Then, the algorithm iterates sequentially over the graphlet families
  $\mathcal{H}_s$, $s = 3, \dots, t$.
  Step $s$ has two substeps at every vertex $v\in V$: 
\[ \begin{array}{l} 
     \mbox{\tt Compute} \mbox{ gross frequencies}
     \\
     g({\cal H}_{s} | G)(v)  = \mbox{\tt upRec}( f({\cal H}_{1:(s-1)}|G) (v), A ) , 
     \\[0.7em]
     \mbox{\tt Convert} \mbox{ to net frequencies}
     \\ 
     f( {\cal H}_s | G)(v) = \mathbf{U}^{-1}_s \, g( {\cal H}_s | G)(v) . 
     \end{array} 
\]
At the first substep, the upward recursion function {\tt upRec} makes
use of the precedent frequencies and adjacency matrix $A$.
The upward recursion functions are typically non-linear.  There exist
various approaches for constructing the recursion
function~\cite{marcus2012,hocevar2014b,melckenbeeck2018,floros2020c},
see also Example~\ref{example:nonlinear-relation}.
The construction approaches also vary in how to exploit the structures
of $A$, including the sparsity.
For the second substep, the transform matrices $\mathbf{U}_{s}$,
$s \leq t$, are pre-computed once and for all.
At iteration step $s$, the frequency conversion is linear with $n(G)$,
the constant prefactor is proportional to $\mbox{nnz}(\mathbf{U}_{s})$.
The main cost lies in the upward recursion, which is
dominated by the computation of the clique frequencies
$g( K_s|G)(v) = f( K_{s} |G )(v) $.  We introduce next our cost reduction strategies.

\begin{figure*}[!t]
  \centering
  \begin{subfigure}{0.28\linewidth}
    \centering
    \includegraphics[width=\linewidth]{
      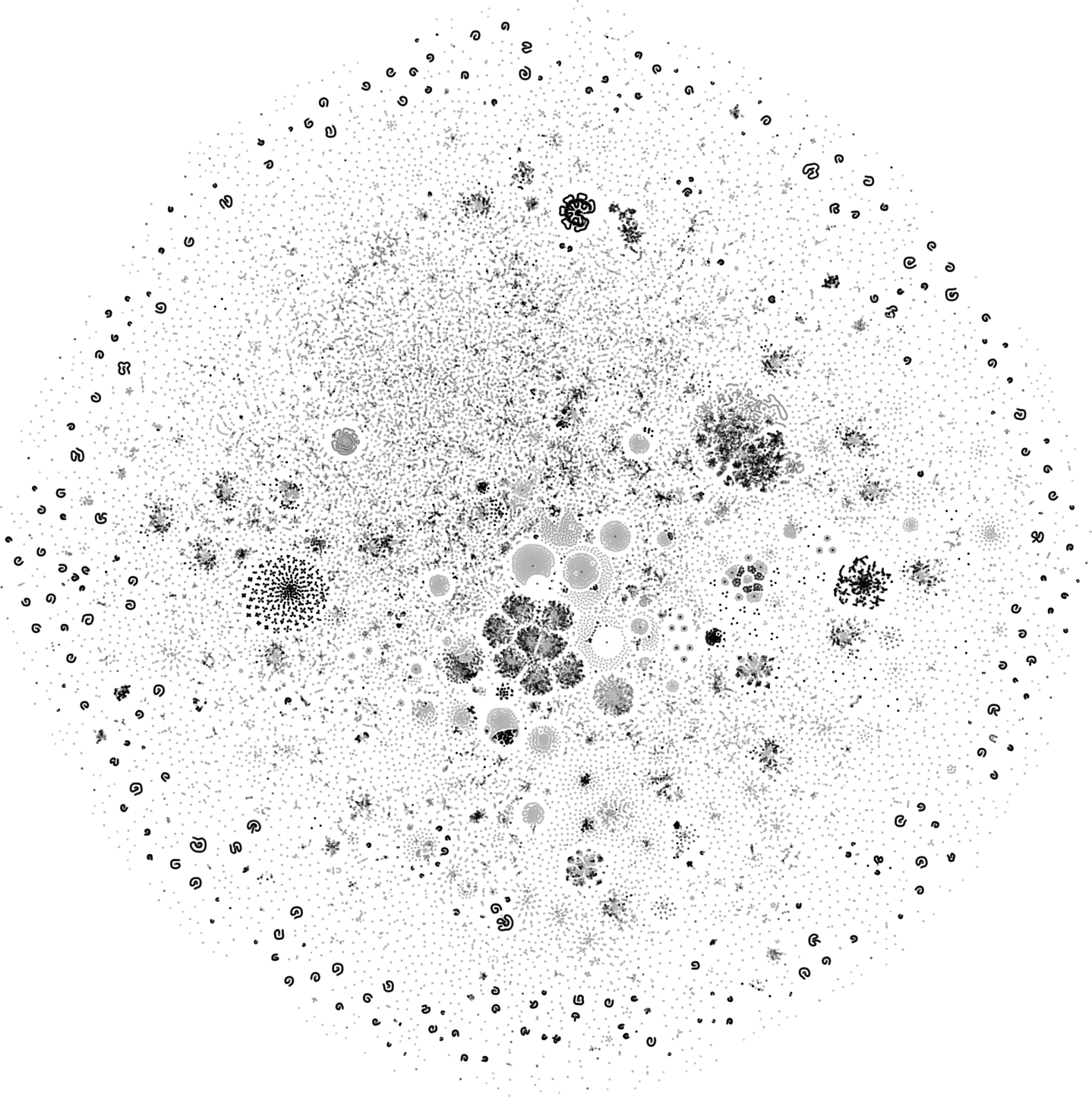}
    \caption{}
  \end{subfigure}
  \begin{subfigure}{0.28\linewidth}
    \centering
    \includegraphics[width=\linewidth]{
      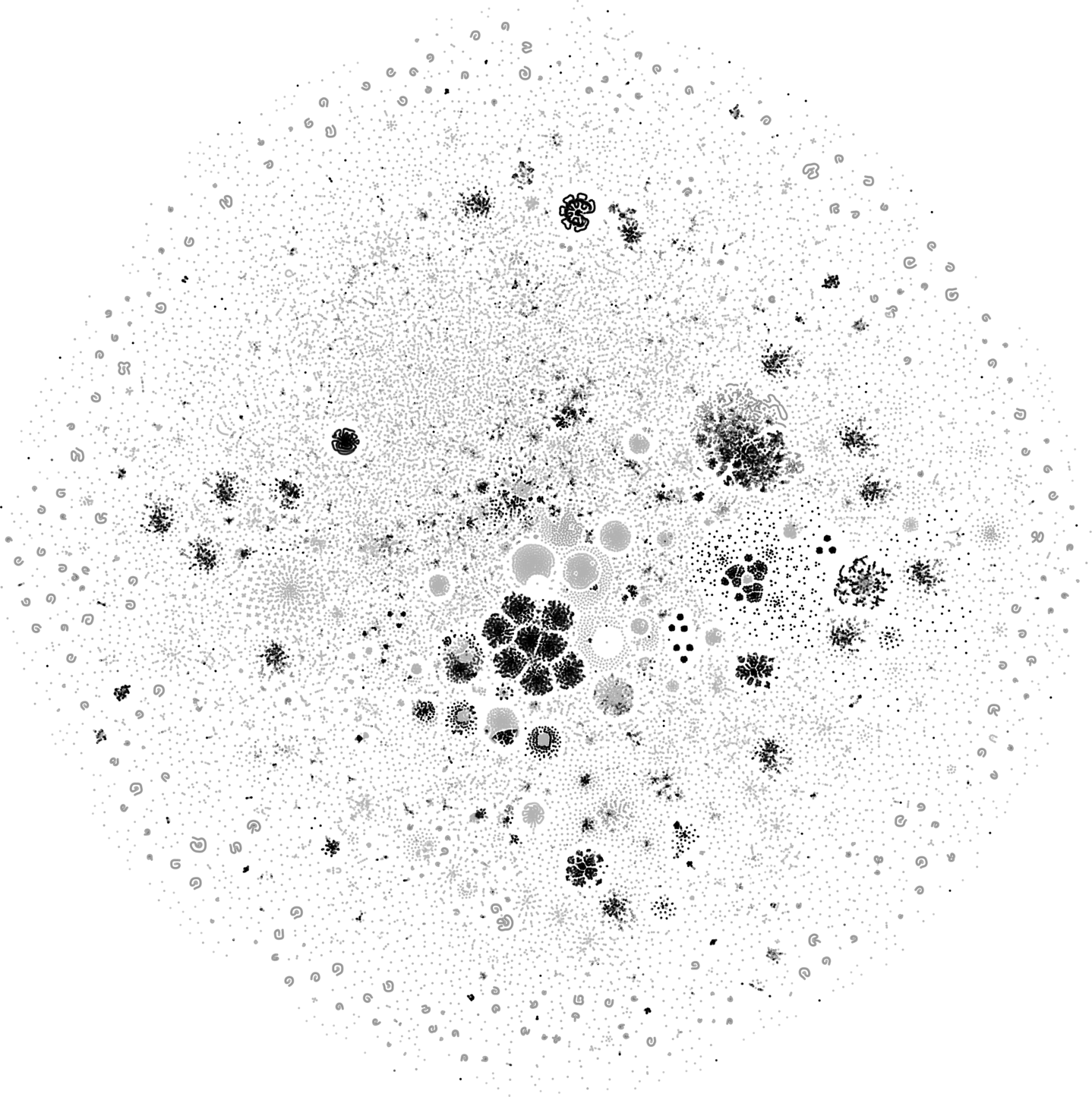}
    \caption{}
  \end{subfigure}
  \begin{subfigure}{0.28\linewidth}
    \centering
    \includegraphics[width=\linewidth]{
      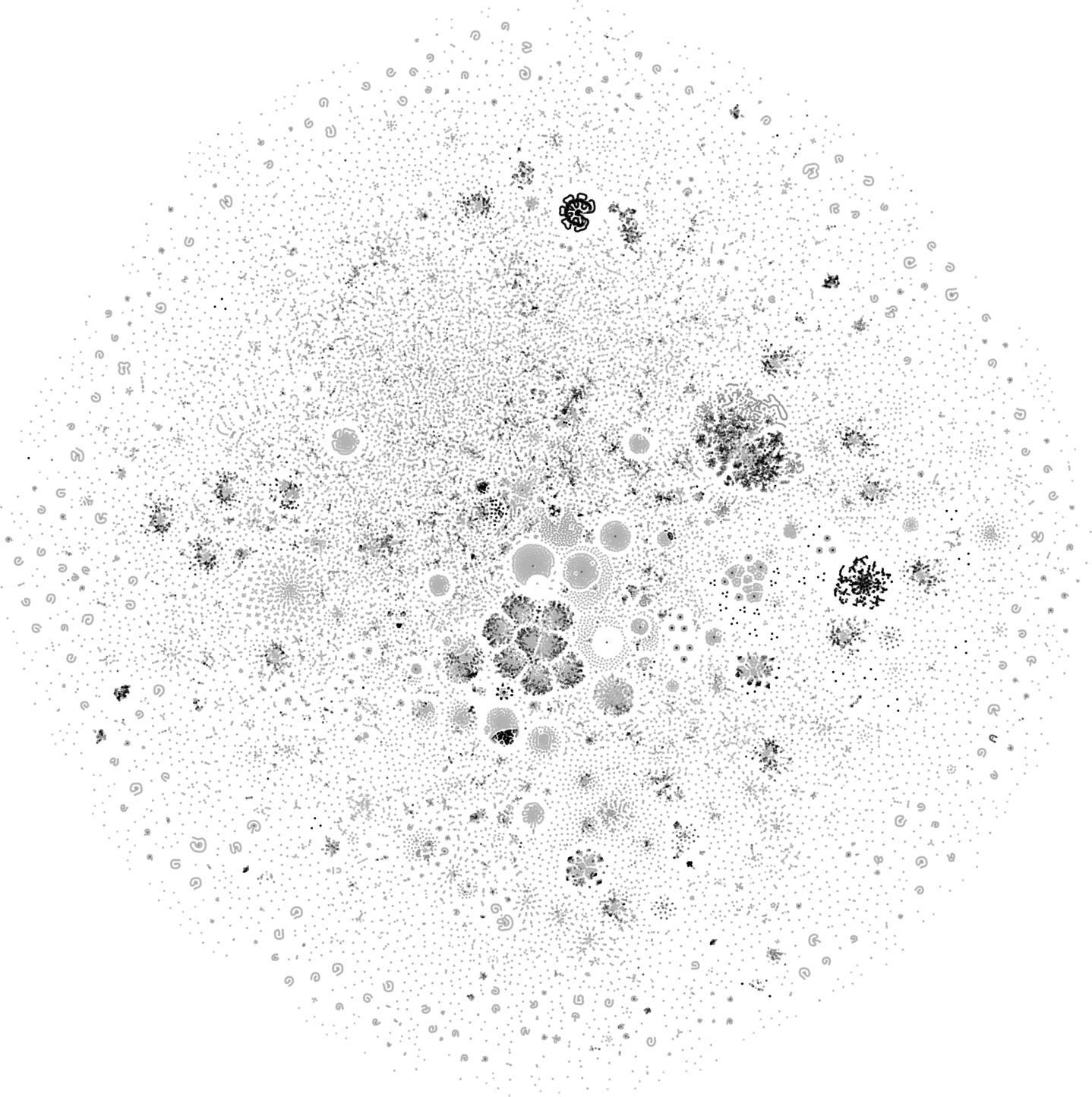}
    \caption{}
  \end{subfigure}
  \\[0.5em]
  \hspace*{\fill}
  \begin{subfigure}{.98\linewidth}
    \centering
    \includegraphics[height=0.12\textheight]{
      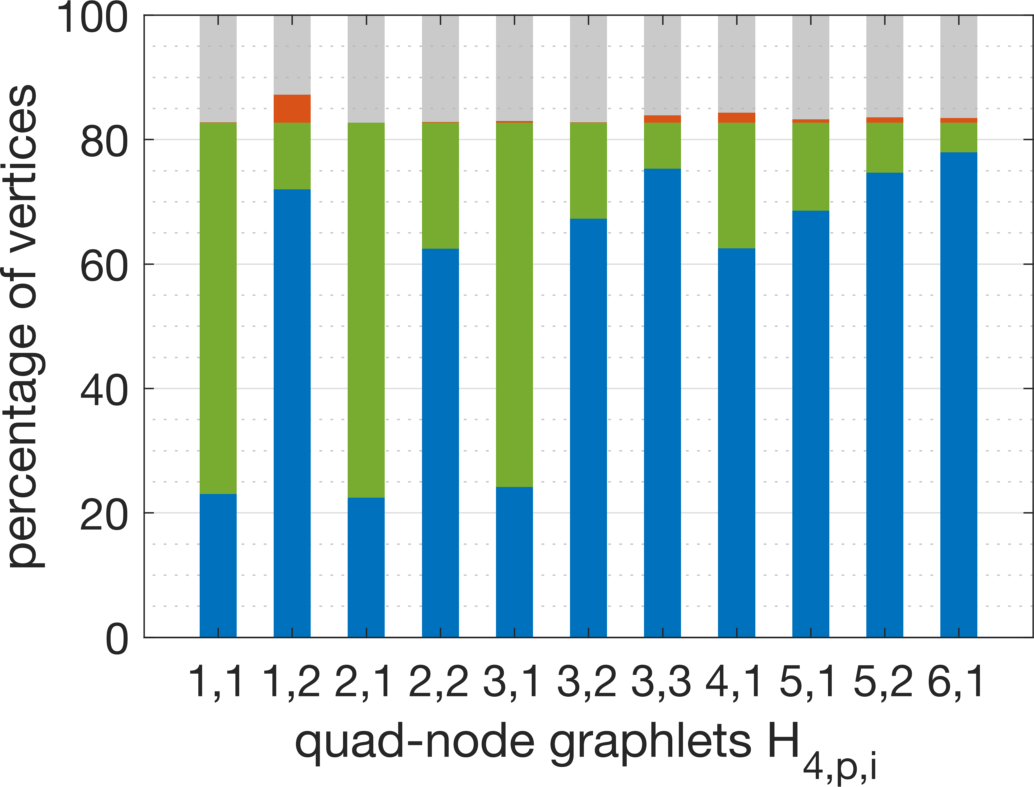}
    \includegraphics[height=0.12\textheight]{
      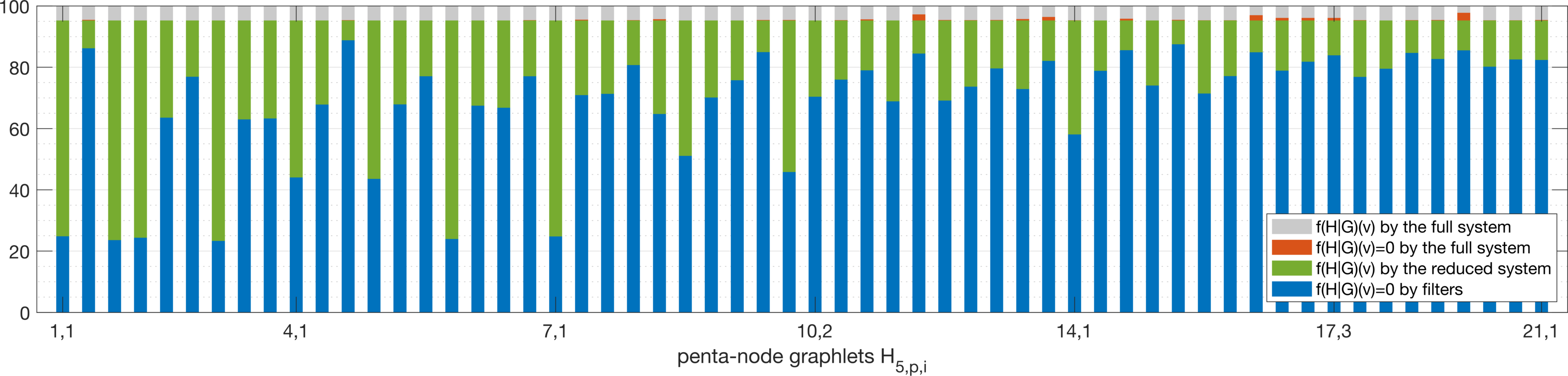}
  \end{subfigure}
  \hspace*{\fill}
  \\[0em]
  \hspace*{1.5em}
  \begin{subfigure}{.2\linewidth}
    \centering
    \caption{}
  \end{subfigure}
  \hspace*{0em}
  \begin{subfigure}{.6\linewidth}
    \centering
    \caption{}
  \end{subfigure}
  \caption{Significant reduction in computation cost by algorithm
    \gsurf{} on network \texttt{NotreDame\_www}. The network, described in
    \Cref{sec:case-study}, has \num{325 729} HTML document nodes and
    \num{757 365} links. (\textbf{Top}): Frequency maps, overlaid on a 2-D spatial embedding of the network~\cite{pitsianis2019}, with respect to
    graphlets
    \textbf{(a)} $H_{4,1,2}$, %
    \textbf{(b)} $H_{4,6,1}$ and %
    \textbf{(c)} $H_{5,4,3}$, respectively.
    Vertices with zero frequencies are shown in black; non-zeros, in gray.
    (\textbf{Bottom}): Reduction bar charts with (\textbf{d}) quad-node graphlets and (\textbf{e}) penta-node graphlets. The percentages of vertices
    (blue) with $f(H|G)(v) = 0$ identified by internal filters,
    (red) with $f(H|G)(v) = 0$ identified by full systems,
    (green) processed by reduced systems, and
    (gray) processed by full systems.
    The smaller the grey area is, the more reduction in computation cost.
    \textbf{Observation}: Reduced systems are used at about
    $\num{84}\%$ of the vertices with quad-node graphlets;
    $\num{95}\%$, with penta-node graphlets; with larger graphlets,
    more reduced systems and more cost reduction.  }
  \label{fig:results-real-world}
\end{figure*}

\subsection{Reduced systems}
\label{sec:reduced-systems}

At step $s$ we can obtain the clique frequency
$f( K_s |G)(v) = g( K_s|G)(v) $ at vertex $v$ economically if for some
$H_j \in \mathcal{H}_s\!\setminus\!K_s$ the frequency $f( H_j |G)(v)$
is known prior to the gross-to-net conversion.
For instance, if graph $G$ is found path-4 free by a linear-complexity
algorithm~\cite{corneil1985},
then $f( H_{4,2} |G)(v) =0$ at all vertices and $f(K_4|G)(v)$ can be
obtained economically by the following approach.
Here, ${\cal H}_s\!\setminus\! H$ denotes the set ${\cal H}_s$ with
$H$ removed.  With the additional information of a zero net
frequency with $H_j$ at $v$, we can reduce the full conversion system
(\ref{eq:frequency-transform}) at $v$ to the following one by variable
substitution,  with $H_{j} \neq K_{s}$, 
\begin{equation}
  \label{eq:reduced-systems}
  \mathbf{U}_{s[sj]} \, f( {\cal H}_s\!\setminus\! H_{j} | G)(v)
  = g( {\cal H}_s\!\setminus\! K_{s} | G)(v).
\end{equation}
Here, $\mathbf{U}_{s[ij]}$ is matrix $\mathbf{U}_{s}$ with row $i$ and
column $j$ removed, and it is non-singular if
$\mathbf{U}_s(j,i) \neq 0$, by \Cref{corr:co-factor-matrix}.
In general, under the conditions that $\mathbf{U}_s(j,i) \neq 0$ and
$i\neq j$, the reduced system with matrix $\mathbf{U}_{s[ij]}$ can be
used to infer $f(H_i|G)(v)$, with the knowledge of $f( H_j|G)(v)$,
$H_i, H_j\in \mathcal{H}_s$, without calculating the gross frequency
$g(H_i|G)(v)$.

\subsection{Frequency filtering}
\label{sec:graphlet-filtering}

The system reduction is not limited to external sources of frequency
information. The reduced systems are not necessarily uniform across
all vertices. We introduce how to infer zero-net frequencies
internally for system reduction and make use of precedent frequency
information local to each vertex.

At step $s$ of algorithm \gsurf{}, $M_{s}$ is used as a binary mask
over the vertex set $V$. The mask is initially set to zero at all
vertices. The algorithm sets $M_{s}(v) = 1$ at vertex $v$ when it is
recognized that $f(H|G)(v) = 0$ for some $H \in \mathcal{H}_{s}$.
The system (\ref{eq:frequency-transform}) can be reduced at any vertex
with $M_{s}(v)=1$.
The zero-frequency recognition takes place before and during the
computation of the gross frequencies.  When the net frequencies with
precedent graphlet families, $f({\cal H}_{1:(s-1)}|G)(v)$, become
available, matrices $\mathbf{W}_{s',s}$, $s' < s$, are used to detect
zero frequencies, based on Proposition~\ref{prop:inter-family}.
During the computation of gross frequencies, as more information
becomes available, the columns of $\mathbf{U}_s$ are used to detect
zero frequencies from precedent graphlets in the same family, based on
Proposition~\ref{prop:intra-family}.
The filtering is self-contained and adaptive to local information.

\subsection{Case study \& complexities}
\label{sec:case-study}

We present a case study with the real-world network
\texttt{NotreDame\_www}.\footnote{The network data is available at
  \url{https://sparse.tamu.edu/Barabasi/NotreDame_www}}
The network is the first shown to follow the celebrated
Bar{\'a}basi-Albert model~\cite{albert1999a}.
We treat it as undirected, without self-loops.
The resultant network has \num{325 729} HTML document nodes and \num{757
  365} links, with average degree $4.65$.
We generate quad-node and penta-node frequency maps using \gsurf{}.
Three of the maps are shown in \Cref{fig:results-real-world}.
There is a significant reduction in computation cost. The reduction is
measured by the percentage of the reduced systems over the entire
vertex set.  About $84\% $ of the local systems are reduced for
generating the frequency maps with quad-node graphlets; about $95\%$,
with the penta-node graphlets. The reduction relies entirely on
internal filtering.

We comment on time complexity and space complexity.  When the
graphlet size is bounded by $t>1$ and when the space complexity 
is of $O( n^{2})$, \gsurf{} is of time complexity $O(n^{3} )$, under
additional assumptions as follows. 
Matrix multiplications are used in the computation of gross
frequencies. Matrix powers $A^k$, $1<k\leq t$, may be computed by the
square-doubling technique. Asymptotic techniques for matrix
multiplications are not used. By this reasoning, the complexity is of
the same order as generating the triangle map alone.  The hardness level in
theoretical worst-case complexity is cubic in time and quadratic in
space. In the real world, large networks tend to be sparse.  The
sparsity shall be exploited to lower both time and space complexities
on average over all sparse matrices, which is shown feasible
in~\cite{floros2020c}.

\section{Discussion}
\label{sec:conclusion}

\subsection{Relations to previous works}
\label{sec:previous-work}
We relate our work to previous works on structures and counts of small
subgraphs in a graph or network.  To this end, we give a brief
overview of the previous works through the lens of our unifying
analysis as introduced in the preceding sections.  The overview is
intended to clarify key distinctions, connections and gaps among the
most relevant and notable works, certain lingering faults and limiting
factors, and the advance we have made.

Relevant previous works may be categorized first by problem types or
objectives and then by solution methodologies.  We name below a few
problems and describe typical or notable solution methods for each.
In a solution method to a particular problem, a connection or
translation to another problem may be utilized.

\noindent{\ (I)} %
Determine whether or not a graph $G$ is free of certain forbidden
connected subgraphs.  Such classical graph recognition problems ask
whether or not the total count of forbidden subgraphs is zero.
Familiar examples include triangle-free, path-free, cycle-free, or
diamond-free graphs~\cite{faudree1997,kloks2000,corneil1985}. A line
graph is free of $9$ small subgraph patterns with no more than $6$
nodes each~\cite{beineke1970,harary1960}.  Specifically, the $9$
patterns are graphlets $H_{4,1}$, $H_{5,17}$, $H_{5,20}$, $H_{6,24}$,
$H_{6,58}$, $H_{6,64}$, $H_{6,71}$, $H_{6,91}$ and $H_{6,99}$ by our
graph index system.  Solution methods for such recognition or
detection problems are typically pattern-specific, leveraging the fixed
pattern topology and adapting to the local and global connectivity
structure in a source graph. Such methods are driven toward optimal
time complexity. Remarkably, the algorithm for line-graph recognition
and root graph reconstruction by Lehot \cite{lehot1974} is of linear
complexity with $m(G)$.

\noindent{\ (II)} %
Compute the total net counts of small connected subgraphs.  Among the
notable solution methods are the work by Alon, Yuster and Zwick~\cite{alon1997},
the work by Kloks, Kratsch and M{\" u}ller~\cite{kloks2000}
and more recent work by Vassilevska Williams and Williams~\cite{williams2018a}.
These methods are common in their use of
fast matrix multiplications for global counting of prescribed
small subgraphs.  The connections between the template structures are
used for the total counts. They are appealing
with asymptotically low complexity. However, they remain impractical
due to the galactically large prefactor constants.

\noindent{\ (III)} %
Find the locations and local counts (i.e., listing) of small cliques in
a graph.  All nodes in a clique are symmetrical, i.e., any clique has
only one orbit. The net count and the gross count of a clique at each
vertex are equal.  An influential work on clique-listing is by Chiba
 and Nishizeki~\cite{chiba1985}. This work also makes a critical link between the subgraph
counting complexity and the graph arboricity, the latter is a measure
of graph connectivity.  Listing cliques up to a fixed size takes
polynomial time. The method by Chiba and Nishizeki makes use of sub-clique
listings.

\noindent{\ (IV)} %
Find the net frequency distributions of graphlets. Graphlets include
but are not limited to small cliques. A small graph template with more
than one orbits can be split into multiple graphlets. The original
concept and use of graphlets are by Pr\v{z}ulj, Corneil and
Jurisica~\cite{przulj2004}. To obtain the graphlet distributions, the
graphlet frequency maps are actually obtained first. Unfortunately,
the collocation relationships among the frequency maps were discarded
in separate extractions for the sequences and distributions, which were
examined subsequently for cross-correlations or associations.

Myriad formulas and procedures exist for computing the net frequencies
of
graphlets~\cite{hocevar2014b,marcus2012,floros2020c,melckenbeeck2016,ortmann2017,washio2003,lee2010,jiang2013,alhasan2018,ribeiro2020,bouhenni2021}.
They can be categorized into three types.
\begin{inparaenum}[(1)]
  \item A type-1 algorithm locates every induced subgraph of a fixed
  size $t$.  The pattern of each induced $t$-node subgraph and its induced
  subgraphs are then recognized by comparison to the chosen graphlet
  templates.  A type-1 algorithm is expensive, the dominant factor in
  the algorithm complexity is the number of the induced subgraphs
  $O(n^{t})$. Common information among overlapping subgraphs is not
  utilized.
  \item A type-2 algorithm starts at every vertex $v$ with the
  frequency at the smallest graphlets and proceeds by upward recursion
  to the frequencies with graphlets with one more node.  Type-2
  algorithms did not have automatically generated equations until a
  recent work by Melckenbeeck et al. in~\cite{melckenbeeck2016}.  The
  question is left wide open how to choose among various ways to
  expand by one node the neighborhood of every vertex $v$.
  \item A type-3 algorithm leverages neighborhood connections, or walks
  in network analysis language, by matrix and vector
  operations~\cite{hocevar2014b,floros2020c}. The interpretation of
  successive matrix-vector products is partially responsible for
  developing the notion of raw or gross frequencies and for
  investigating the relationship between gross and net frequencies.
\end{inparaenum} 

The problem with motif detection and discovery can be described
in terms of the above primary problems.

\subsection{Potential impacts}
\label{sec:potential-impact} 

The graph encoding system with graphlet frequencies can serve as the
ground for a graph calculus. A large problem is presented by small
subgraph attributes locally at vertex neighborhoods (differentiation)
as well as by their spatial and statistical distributions regionally
and globally (integration). A solution to the large graph problem can
then be facilitated by a divide and conquer approach. Such ideas and
approaches are not new. For example, at the center of the graph
matching problem is graph isomorphism.  Subgraph analysis is used to
facilitate and accelerate graph isomorphism tests~\cite{mckay2014}.

We make another important connection: each of the problems described
in \Cref{sec:previous-work} can find its robust, statistical counterparts that are
tolerant of errors, noise or perturbation.  Such robust property is
important to analysis and modeling of real-world networks.  A
statistical subgraph problem is not necessarily associated with a null
model.  Applied graph/network problems that can get direct benefits
from such graph calculus include community or anomaly detection over a
network; graph learning; sampling or pooling on a graph; and graph
indexing, search and retrieval which are based on comparison,
alignment and matching among networks.  The great potential of such
graph calculus cannot be underestimated.

\clearpage
\phantomsection
\label{sec:references}
\addcontentsline{toc}{section}{References}
\printbibliography

\end{document}